\documentclass[bj]{imsart}

\RequirePackage{amsthm,amsmath,amsfonts,amssymb}
\usepackage[numbers]{natbib}

\RequirePackage[colorlinks,citecolor=blue,urlcolor=blue]{hyperref}
\RequirePackage{graphicx}
\graphicspath{{figures/}}
\usepackage{xfrac}
\usepackage{bbm}
\usepackage{float}
\allowdisplaybreaks
\usepackage{mathtools}

\usepackage{booktabs} 
\usepackage{caption} 
\usepackage{subcaption} 
\usepackage{graphicx}
\usepackage{pgfplots}
\usepackage[all]{nowidow}
\usepackage[utf8]{inputenc}
\usepackage{tikz}
\usetikzlibrary{er,positioning,bayesnet}
\usepackage{multicol}
\usepackage{algpseudocode,algorithm,algorithmicx}
\usepackage[inline]{enumitem} 

\newcommand{\E}{\mathbb{E}}
\newcommand{\R}{\mathbb{R}}

\usepackage{tikz}
\usetikzlibrary{matrix}

\startlocaldefs
\theoremstyle{plain}

\newtheorem{theorem}{Theorem}[section]
\newtheorem{lemma}[theorem]{Lemma}
\newtheorem{assumption}{Assumption}
\theoremstyle{remark}
\newtheorem{definition}[theorem]{Definition}
\newtheorem{corollary}[theorem]{Corollary}
\newtheorem{example}{Example}
\newtheorem{remark}{Remark}

\numberwithin{theorem}{section}

\endlocaldefs

\begin{document}

\begin{frontmatter}
\title{Kernel-weighted specification testing under general distributions}
\runtitle{Specification testing under general distributions}

\begin{aug}
\author[A]{\fnms{Sid} \snm{Kankanala}\ead[label=e1]{sid.kankanala@yale.edu}}
\author[B]{\fnms{Victoria} \snm{Zinde-Walsh}\ead[label=e2]{victoria.zinde-walsh@mcgill.ca}} 
\address[A]{Department of Economics, Yale University, \printead{e1}}

\address[B]{Department of Economics, McGill University, 
\printead{e2}} 
\end{aug}
\begin{abstract}
Kernel-weighted test statistics have been widely used in a
variety of settings including non-stationary regression, survival analysis, propensity score and panel data models.  We develop the limit theory for a kernel-weighted specification test of a parametric conditional mean when the law of the regressors may not be absolutely continuous to the Lebesgue measure and admits non-trivial singular components. In the special case of absolutely continuous measures, our approach weakens the usual regularity conditions. This result is of independent interest 
and may be useful in other applications that utilize kernel smoothed statistics. Simulations illustrate the non-trivial impact of the distribution of the conditioning variables on the power properties of the test statistic.
\end{abstract}


\end{frontmatter}


\section{Introduction}
Kernel-weighted statistics are widely used for inference on the
functional form of a density, conditional distribution and conditional mean. In testing for a parametric specification of a conditional mean, kernel-based tests have been used in
the traditional regression context \citep{Zhang,mammen1,horowitz2001}. Those types of statistics are also employed in
various extensions such as regression quantiles \citep{zheng1998consistent,mammen2}, semi-parametric  models \citep{chen2009goodness}, propensity score \citep{ShaikhVytlacil}, panel data \citep{lin}, non-stationary regression
\citep{gao,phillips} and survival analysis \citep{muller2019goodness}. If $F_{X}$ represents the law of the regressors, it can always be   expressed in its Lebesgue decomposition: \begin{align}
&  F_X     =\rho_{d} F_{X}^d  +\rho_{a.c.}%
F_X^{a.c.}  + \rho_{s}  F_X^{s}
\;,\label{decomp}\\ &
 \rho_{d}, \: \rho_{s}, \: \rho_{a.c.}  \in [0,1] \; \; , \; \; \rho_{d} + \rho_{s} + \rho_{a.c.}=1 \; \; , \nonumber
\end{align}
where $F_{X}^d$ is a discrete measure, $F_{X}^{a.c.}$ is absolutely
continuous to the Lebesgue measure and $F_{X}^s$ is a
  singular continuous measure. While kernel-weighted statistics have been investigated in several distinct applications, not much is known about the statistical properties of procedures based on these statistics when the distribution is ``contaminated'' with non-trivial singular components. To the best of our knowledge, all the available analyses in the literature have assumed that the Lebesgue decomposition of $F_{X}$ does not admit any singular components, usually with additional smoothness regularity conditions
imposed on the density function.
  
  In this paper, we study a class of kernel-weighted U-statistics that frequently arise in the analysis of goodness-of-fit testing. A general limit theory is provided that allows for the Lebesgue  decomposition of $F_{X}$   to contain singular  components. Our generalization of the standard limit theory is partially motivated by a desire to understand the finite sample properties of these statistics when the distribution is absolutely continuous with a density  that is either non-smooth or possesses large derivatives.  Indeed, in finite samples, singular distributions tend to exhibit characteristics similar to such measures. For example, the claw and its variants in \cite{Marron} are Gaussian mixtures with a  density that exhibits sharp continuous  ``spikes'' at several points in the support. While these distributions are absolutely continuous with a
smooth density, the derivatives may be so large as to make them resemble distributions with
singular components in finite samples.

\begin{figure}[H]
  \begin{subfigure}{0.49\textwidth}
    \centering
    \includegraphics[height=4cm,width=\linewidth]{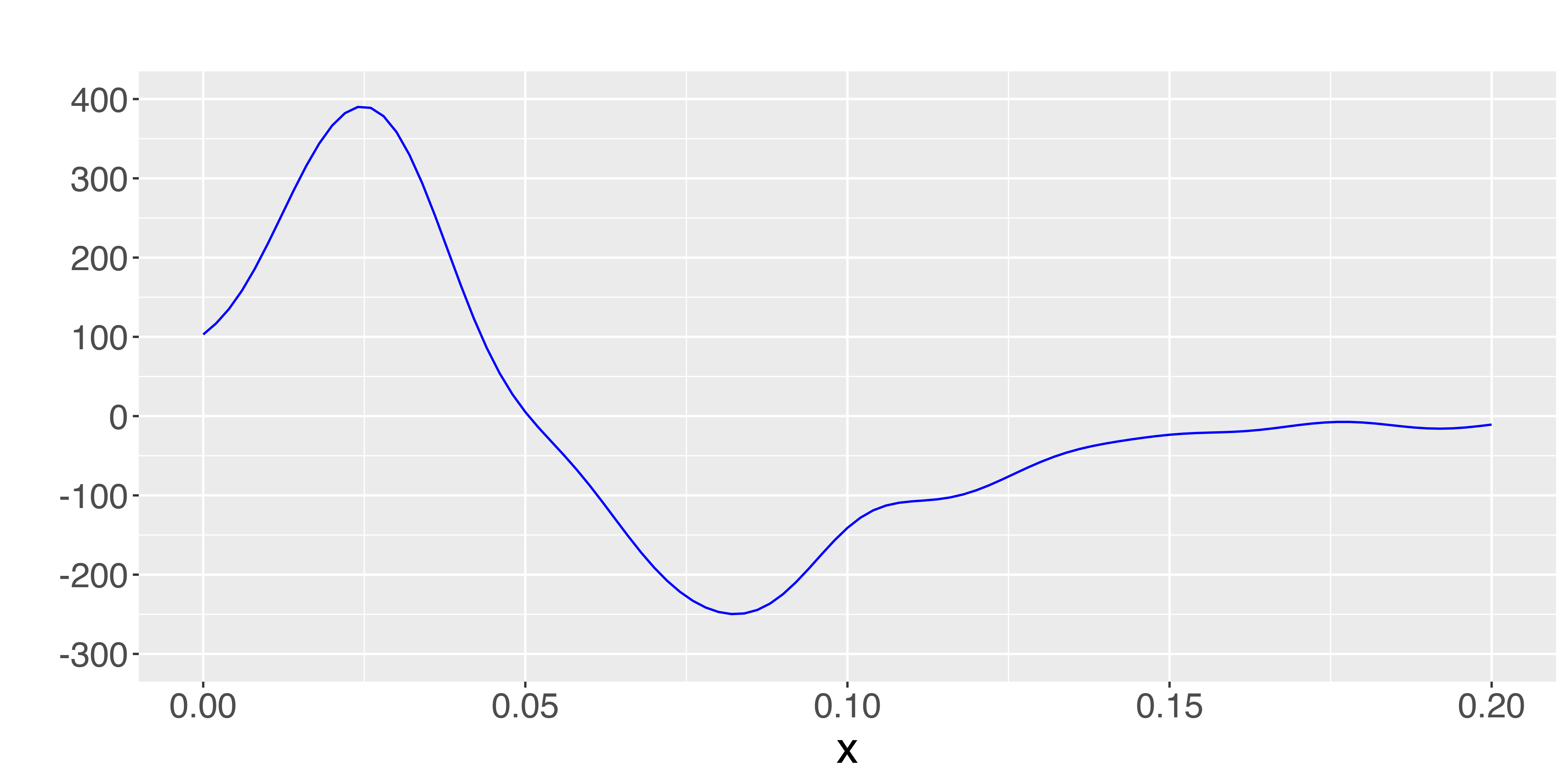}
    \caption{fuel}
    \label{0a}
  \end{subfigure}
  \begin{subfigure}{0.49\textwidth}
    \centering
    \includegraphics[height=4cm,width=\linewidth]{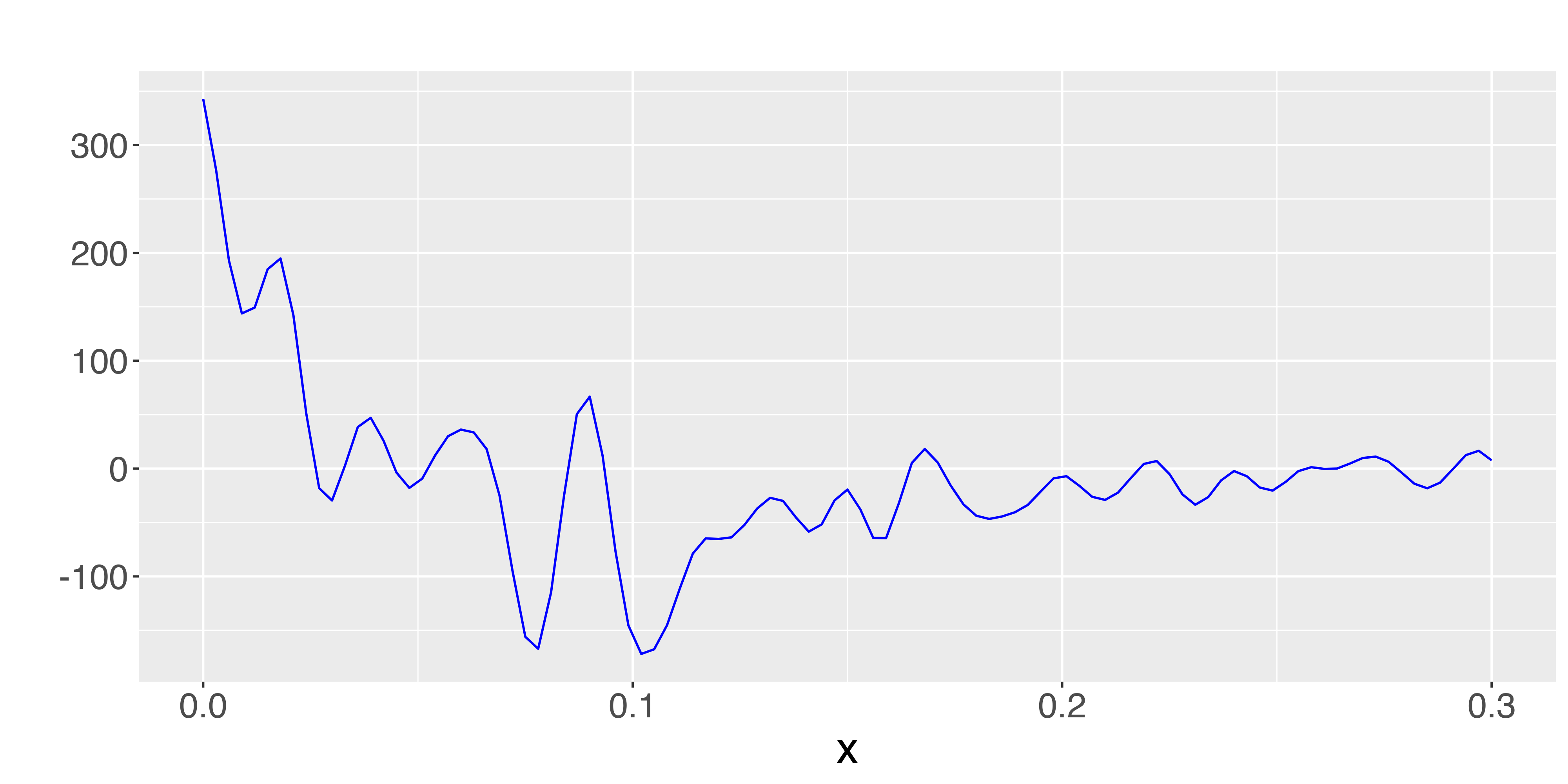}
    \caption{food out}
    \label{0b}
  \end{subfigure}
  \caption{Estimated Gaussian kernel  density derivative of the  household budget share allocated to fuel and food outside from the 1995 U.K. Family Expenditure Survey (see e.g. \cite{blund}), with smoothing bandwidth chosen using  least-squares cross-validation  \cite{wolfsmooth}.}
  \label{fig0}
\end{figure}

As illustrated in Figure \ref{fig0}, the estimated density derivative of covariates may be quite large: within the same data set, the budget shares for alcohol, travel and leisure were significantly more volatile  and exhibited estimated derivatives in the thousands. While the standard limit theory can account for this phenomenon by imposing an arbitrarily large upper bound on the density and/or its derivative, we find it more appropriate to model these situations with the possibility of a non-trivial singular component.

\begin{figure}[H]\centering
\subfloat{\label{a}\includegraphics[height=3cm,width=.45\linewidth]{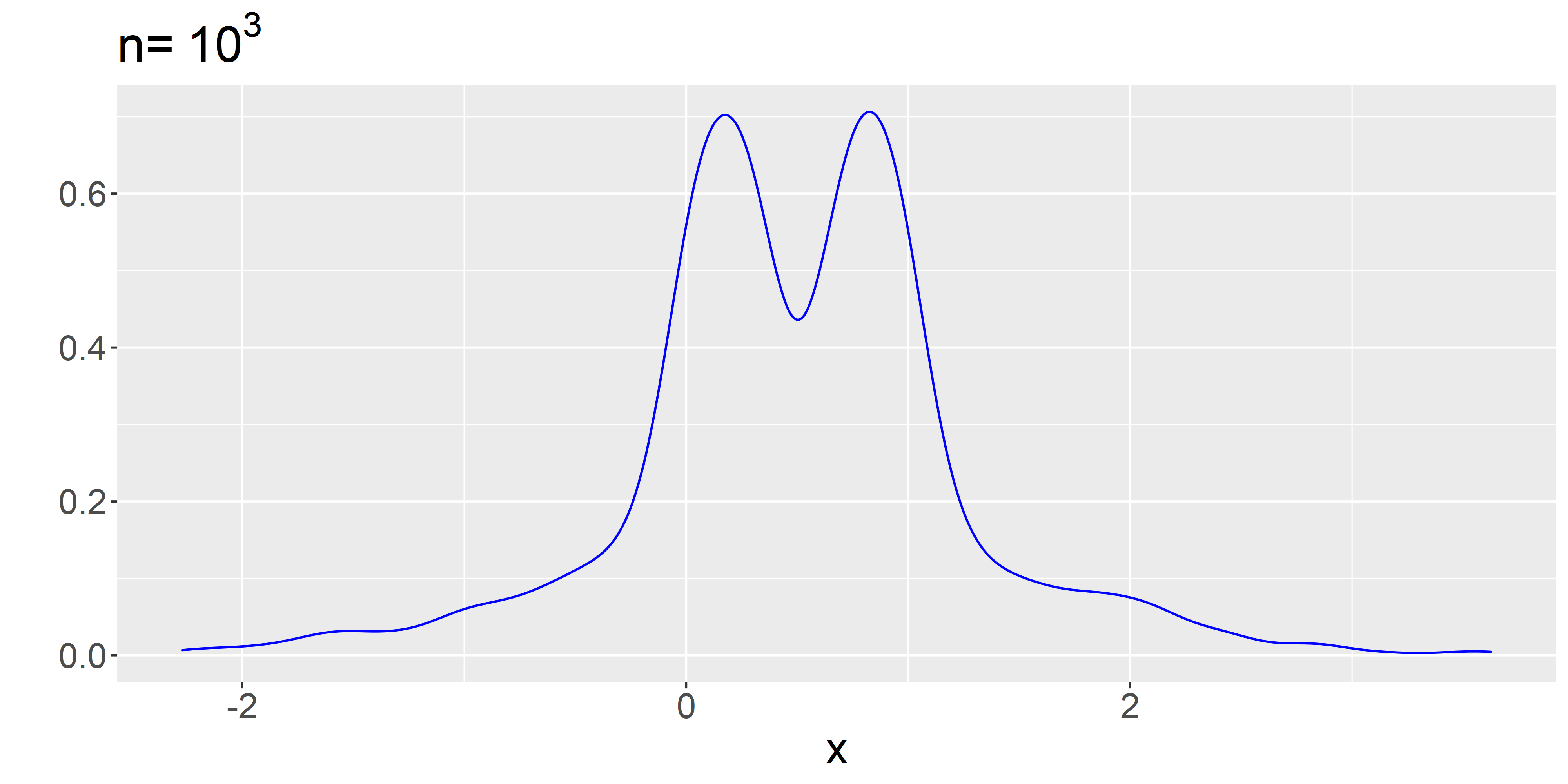}}\hfill
\subfloat{\label{b}\includegraphics[height=3cm,width=.45\linewidth]{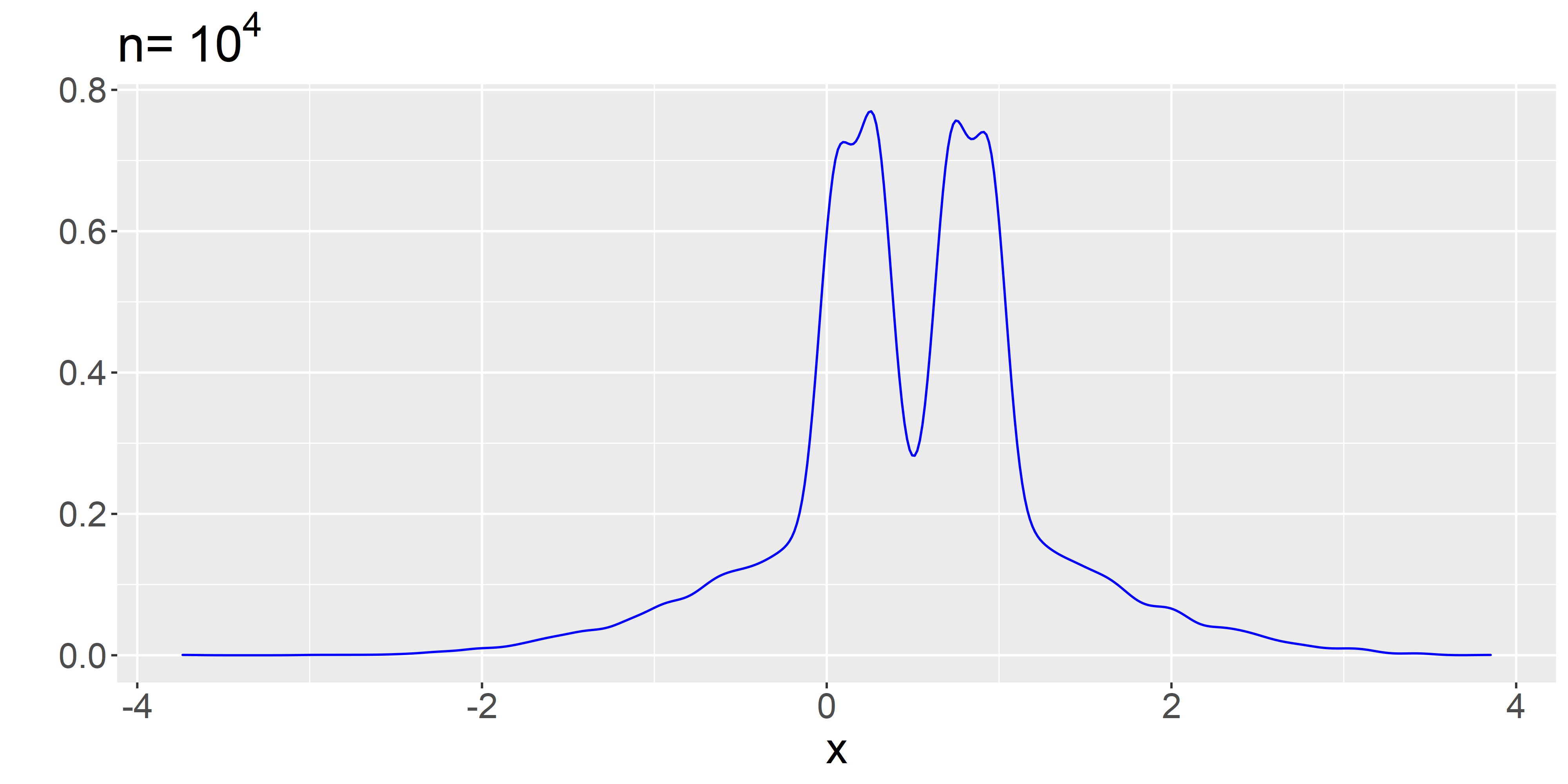}}\par 
\subfloat{\label{c}\includegraphics[height=3cm,width=.45\linewidth]{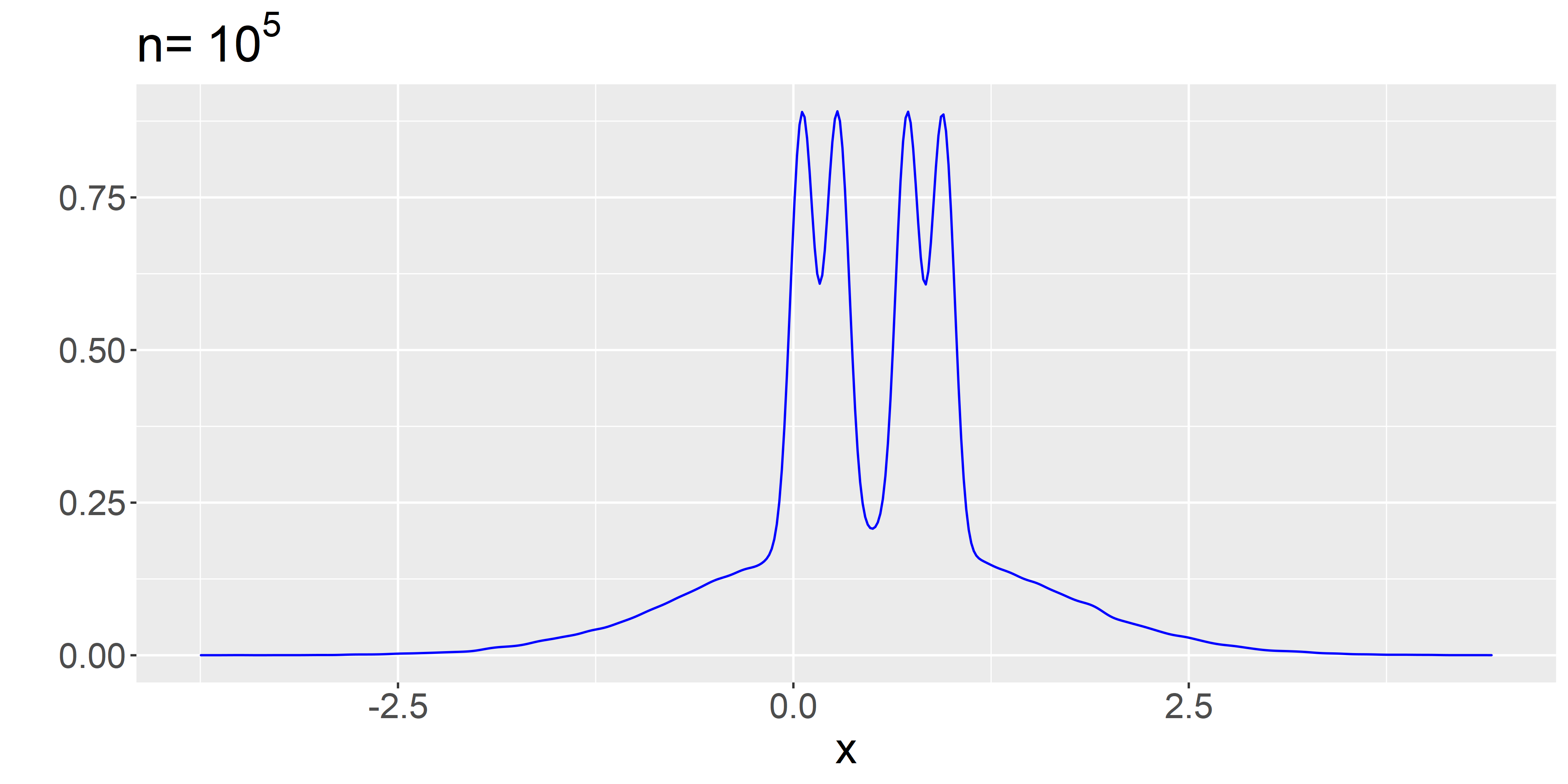}}
\caption{Realizations of a kernel density estimate with random sample of size $n$ drawn from an equal weight mixture  of $N(0.5,1)$ and the standard Cantor distribution (with smoothing bandwidth selected by the default rule in R).}
\label{fig}
\end{figure}

The main results of this paper develop the limit theory for kernel-based inference statistics when the distribution of the conditioning variables admits a Lebesgue decomposition  with singular components. The class of distributions that our theory covers include absolutely continuous measures with a bounded density, measures supported on lower dimensional subspaces,  continuous measures that contain discrete marginals (e.g. mixed variables and discrete regressors),   normalized Hausdorff measures on self-similar fractals and any Ahlfors-David regular measure. All our main results also apply to countable mixtures generated by such measures.  Even in the special case where $F_{X}$ is absolutely continuous, our analysis only requires a bounded density and avoids imposing any further  regularity conditions.

Denote by $\iota$ a vector of ones and the $F_{X}$ probability of a cube  of radius $h > 0 $ and centred at $x$  by $F_X(x-h \iota , x +h \iota)  = \mathbb{P}(\| X- x\| _{\infty} \leq h )  $. Our approach provides interpretable conditions that are based on local features of the underlying distribution through the expected small ball probability:  \begin{align}
\label{cube}   
r(h) = \E \big[   F_{X} (X -h \iota , X +h \iota)      \big].
\end{align}
If $F_{X}$ is a probability measure on $\R^q$ with a bounded density, $r(h) \asymp h^{q}$ as $h \downarrow 0 $ but  the rate is always slower when the Lebesgue decomposition admits singular components.
The idea
behind our results is to utilize a version of integration by parts to extract local features of $F_{X}$ that depend on $r(h)$ when the density does not exist (at the cost of using sufficiently differentiable
kernels). This approach is of independent interest and may  be useful in other applications that utilize kernel-based statistics. 

Theorem \ref{asymptotics} establishes the general result for asymptotic normality of a degenerate kernel-weighted U-statistic. For some mixture subclasses that include singular distributions of reduced Hausdorff
dimension $sq$ for some $s \in (0,1)$, exact rates of convergence for the U-statistic are obtained. These results form the basis for the study of self-normalizing statistics that  arise in goodness-of-fit tests of a parametric regression function. In Theorem \ref{asymptotics2}, we establish the limit distribution  for such statistics under the null hypothesis. In Theorem \ref{pow1}, we develop the local power analysis of the test statistic under a Pitman sequence of local alternatives \begin{align*}
H_{1}:Y_{}=g_{}\left(  X_{},\beta_{0}\right)  +\gamma_{n}\delta\left(
X_{}\right)  +u_{} \; \; \;\; ,
\end{align*} 
where $\gamma_n \downarrow 0$ and $\delta(.)$ is a fixed drift function that determines the direction of approach to the null model.
We characterize the fastest possible rate $\gamma_n$ at which alternatives can approach and yet remain distinguishable from the null.  If the Lebesgue decomposition of $F_{X}$ contains singular components, we show that it is possible for the alternatives to approach the null at a rate faster than in the fully absolutely continuous case. The novel feature of
the result is the interplay between the rate of approach $\gamma_n$,  the direction of  approach $\delta(.)$ to the null model and the singular components of the
distribution.  In particular, power does not exist at the fastest possible rate if the support of $\delta(.)$ does not sufficiently ``touch'' areas where the local singularity of the measure (the rate at which $ h \rightarrow   \mathbb{P}(\| X- x\| _{\infty} \leq h )   $ decays) coincides with $r(h)$. In Theorem \ref{pow2}, we  provide further details on the mechanism through which this interplay can influence the local power of the test statistic.

The paper is organized as follows. Section \ref{sec2} provides the model framework and assumptions. Section \ref{sec3} develops the main limit results. Section \ref{sec4} provides simulation evidence on the sensitivity of the kernel test statistic to the distribution of the conditioning variables. Section \ref{sec5} concludes. The supplemental file \cite{sidvsup} contains  additional proofs and technical results that were omitted in the main text.

\section{Framework and assumptions}  
\label{sec2}

Consider the nonlinear regression model \begin{equation}
Y=  \bar{g}(X) +u \;, \; \; \;  \; \; \;  \text{ } \E\left(  u|X\right)  =0,\label{model}%
\end{equation}
where $X\in\mathbb{R}^{q}$ is a vector of regressors, $Y$ is a scalar dependent variable and $u$ is an unobserved error. Given a family of parametric regression functions  $\{ g(x,\beta) : \beta \in \Theta \}$ indexed by a finite dimensional parameter $\beta \in \Theta \subseteq \R^p$, we address the common problem of testing the null hypothesis  $$ H_0 : \mathbb{P}( \bar{g}(X) = g(X,\beta_0)) = 1 \; \; \;  \text{for some} \;\; \beta_0 \in \Theta .$$
Denote by $\hat{u}$ the vector of estimated residuals $\hat{u}%
_{i}=Y_{i}- g(X_i,\hat{\beta})$ (e.g.  using non-linear least squares or maximum likelihood). To test $H_0$, we make use of the kernel smoothed statistic  \begin{equation}
\hat{I}_{n}=\frac{1}{n(n-1)}\sum_{i=1}^{n}\sum_{j\neq i}\hat{u}_{i}\hat{u}%
_{j}K\left(  \frac{X_{i}-X_{j}}{h_{n}}\right)  \;,  \label{I(n)}
\end{equation}
where $K(.)$ is a non-negative symmetric kernel function and $h_n \downarrow 0$ is a deterministic bandwidth sequence.  The usual goodness-of-fit statistic (see e.g. \cite{gao,phillips,Zhang}) uses a self-normalized form \begin{equation}
\label{selfn} \hat{\tau}_n =   \frac{n \hat{I}_n}{ \sqrt{ \hat{\sigma}_n^2}} \; ,
\end{equation}
where $ n^{-2}  \hat{\sigma}_n^2$ is an estimator of the variance of $ \hat{I}_n$.  In the literature, this test statistic falls under the class of smoothing-based tests   (see \cite{gonz2013} for a comprehensive review). Let $U_n$ denote the version of $\hat{I}_n$ obtained by replacing $\hat{u}_i$ with the unobserved error:  \begin{equation}
U_{n}=\frac{1}{n(n-1)}\sum_{i=1}^{n}\sum_{j\neq i}u_{i}u_{j}K\left(
\frac{X_{i}-X_{j}}{h_{n}}\right) .\label{U(u)}
\end{equation}

\subsection{Notation}  Denote by $\iota$ a vector of ones and $c \iota$ a vector of $c \in \R$. For any vector $x \in \R^q$, denote the coordinate components by $x=(x_1,\dots,x_q)$. Let $\| .\|_{2}, \|. \|_{\infty} $ and $\| . \|_{\text{op}}$ denote the Euclidean, infinity and operator
norm, respectively. For positive sequences $(a_n,b_n)$, we use $a_n \lessapprox b_n$ to denote $\limsup_{n \rightarrow \infty} a_n / b_n < \infty$ and $a_n \asymp b_n$ to denote $a_n \lessapprox b_n \lessapprox a_n$. Let $L^p(X)$ denote the usual equivalence class of $p$ integrable (with respect to $F_X$) functions  that are measurable with respect to the $\sigma$ algebra generated by $X$. Denote the absolute continuity of a measure $\nu$ with respect to a measure $\lambda$ by $\nu \ll \lambda$.  Let $\E$ and $\mathbb{P}$  denote the usual expectation and probability operators.  In the special case where $F_{X}$ can be expressed as a mixture that includes component $F_{t}$, the notation $\underset{X  \stackrel{}{\sim} F_{t} }{\mathbb{P}} $ and $\underset{X  \stackrel{}{\sim} F_{t} }{\mathbb{E}}$ will be used to indicate that the operators are defined with respect to $X\sim F_{t}$.  We use 
$\underset{d}{\rightarrow}$ to  denote convergence in distribution.

\subsection{Assumptions} 

\begin{assumption}
\label{data}  $  Z_i =(u_i, X_i) \in \R \times \R^q $ is a sequence of independent and identically distributed (i.i.d) random vectors.
\end{assumption}

\begin{assumption}
\label{error} (i) $\mathbb{E}\left(  u_{}|X_{}\right)  =0$. (ii) The functions $\mu_{2}(X) = \E[u^2|X]$ and $\mu_{4}(X) = \E[u^4|X]$ are bounded away from infinity: $\mathbb{P}(\mu_{l}(X) \leq B) =1$ for some $B <\infty $ and $l=2,4$. (iii) The function $\mu_{2}(X)$ is bounded away from zero: $\mathbb{P}(\mu_{2}(X) \geq b) = 1$ for some $b > 0 $.
 
\end{assumption}

\begin{assumption}
\label{kernel} (i) The function $K(.)$ is a product kernel: $K (x) =
\prod_{i=1}^{q} k (x_{i})$ for every $x=(x_{1},\dots,x_{q}) \in\mathbb{R}^{q}%
$. (ii) The kernel function $k: \R \rightarrow \R$ is non-negative, continuous, symmetric around zero, strictly decreasing on $[0,1]$ and has support $ \Delta =  [-1,1]$. (iii) $k(\,\cdot\,)$ is twice continuously differentiable  on the interior $ \Delta^\mathrm{o} = (-1,1)$ and the  derivatives admit a continuous extension to $\Delta$.
\end{assumption}

Assumption \ref{data} could be generalized but is made here to facilitate the
focus on the distribution of the conditioning variables. Assumptions
\ref{data} and \ref{error}(i) imply that the statistic $U_{n}$ in $\left(
\ref{U(u)}\right)  $ is degenerate. Assumptions \ref{error}(ii) and \ref{error}(iii) are made for convenience and could be weakened further. Of course, Assumption \ref{error}(ii-iii)  applies under homoscedasticity.

Assumptions \ref{kernel}(i-iii) are standard and satisfied by e.g. the
Epanechnikov kernel $k_{E}(t) = \frac{3}{4}  (1-t^2) \mathbbm{1} \{ |t| \leq 1  \} $
and Quartic kernel $  k_{Q}(t) = \frac{15}{16} (1-t^2)^2 \mathbbm{1} \{ |t| \leq 1 \} $. The support assumption on $k(\,\cdot\,)$ could be modified to allow for any compact interval without changing any of our main results.  This could be generalized even further to admit a wider class of kernel
functions (e.g. Gaussian kernels) where the assumption of compact support is replaced with a rate of decay. However, such an analysis will typically involve some interplay between the decay rate of the kernel and the tails of the distribution. Here, we simplify to highlight the impact of the distribution. Assumption \ref{kernel}(iii) is satisfied if $k'(\,\cdot\,), \: k''(\,\cdot\,)$ exist and are uniformly continuous on $ \Delta^\mathrm{o} = (-1,1)$. It ensures that in the event a density does not exist,  one can use (as we explain below) integration by parts to find the
local behavior of  moments of the U-statistic.

\section{Main results} 
\label{sec3}
In Section \ref{sec3.1}, we provide results about the moments of kernel smoothed statistics. In Section \ref{sec3.2} and \ref{sec3.3}, a class of distributions is defined over which asymptotic normality for the U-statistics in (\ref{I(n)}, \ref{U(u)}) is subsequently established. Section \ref{sec3.4} develops the limit theory 
and local power analysis within the context of specification testing.
\subsection{Derivations and bounds for moments}
\label{sec3.1}
From the seminal work of Hall \cite{Hall84, hallhyde}, it is known that limit theory for $U_n$ (\ref{U(u)})  can be established by appealing to a version of the martingale central limit theorem. Indeed, by defining \begin{align}
&  H_{n}(Z_{1},Z_{2})=u_{1}u_{2}K\bigg(\frac{X_{1}-X_{2}}{h_{n}}%
\bigg)\;,\label{h_n}\\
&  G_{n}(Z_{1},Z_{2})=\mathbb{E}\big[H_{n}(Z_{1},Z_{3})H_{n}(Z_{2}%
,Z_{3})|Z_{1},Z_{2}\big]\;,\label{G}%
\end{align}
it is shown in  \cite[Theorem 1]{Hall84} that $ n U_n / \sqrt{2 \E(H_n^2)}\underset{d}{\rightarrow} N(0,1) $, provided that the moments satisfy \begin{equation}
\frac{\E(G_n^2) + n^{-1} \E(H_n^4) }{ \{  \E(H_n^2)  \}^2   } \xrightarrow[n \rightarrow \infty]{} 0.   \label{hallcond}
\end{equation}
As noted in the literature (e.g. \cite[pp. 154-155]{koltch}), Condition (\ref{hallcond}) (or its variants in other applications) is typically difficult to interpret as it depends non-trivially on the underlying distribution of the regressors. It is shown in \cite{Zhang} (see also \cite{mammen1,mammen2,ShaikhVytlacil} for related applications) that when the distribution is absolutely continuous and certain smoothness regularity conditions hold on the density,  Condition (\ref{hallcond}) reduces to the usual restriction on the bandwidth: $h_n \downarrow 0$ and $n h_n^q \uparrow \infty$.  

Our starting point in generalizing beyond absolutely continuous measures is to derive the distributional restrictions  that are implicitly imposed through Condition (\ref{hallcond}). This requires us to express and bound the moments that appear in (\ref{h_n}, \ref{G}) in terms of interpretable functionals of $F_{X}$. This will be the subject of several subsequent Lemmas that appear below.

\begin{definition}
Let  $\mathcal{O} = (l_1,u_1) \times \dots \times (l_q,u_q) \subset \R^q  $ where  $\{(l_{i},u_{i}) \}_{i=1}^{q}$ denote open intervals of finite length. We say a function $g : \mathcal{O} \rightarrow \R$ is sufficiently differentiable on $\mathcal{O}$ if  the mixed partials $$  \frac{\partial^q g(x)}{\partial x_{1} \dots \partial x_{q}} ,  \frac{\partial^{q-1} g(x)}{\partial x_{ 2} \dots \partial x_{q} } , \dots , \frac{\partial g(x)}{\partial x_{ q}} $$ exist and admit continuous extensions to the closure of $\mathcal{O}$.  In this case, we denote $$ \partial_{x} g(x) = \frac{\partial^q g(x)}{ \partial x_{1} \dots \partial x_{q}} .$$
\end{definition}
Given $f \in L^1(X)$, a straightforward application of Fubini's theorem shows that integration  (with respect to $f  dF_{X} $) of a compactly supported sufficiently differentiable function  admits a representation as a standard integral with respect to the Lebesgue measure. The integrand in this case is the $\partial_{t} g(.)$-weighted $f dF_{X} $ measure of a  Euclidean cube.
\begin{lemma}
\label{aux} Let $f \in L^1(X)$ and $\{(l_{i},u_{i}) \}_{i=1}^{q}$ denote open intervals of finite length.  Suppose $g: \R^q \rightarrow \R$ is bounded, continuous, sufficiently differentiable on $\mathcal{O} = (l_1,u_1) \times \dots \times (l_q,u_q)$ and has support contained in the closure of $\mathcal{O}$ . Additionally, for $q > 1 $ and every $2 \leq k \leq q $, suppose that the mixed partial $$  \frac{\partial^{q-k+1} g(x)}{ \partial x_{k} \dots \partial x_{q-1}  \partial x_{q}}  $$
vanishes at $x_j = u_j  \, $ for  $ j < k $. Then \begin{equation}
 \int_{\R^q} f(x)g(x)d F_{X} (x) =(-1)^{q}\int_{  \R^q  } \prod_{i=1}^q \mathbbm{1}  \big \{ l_i \leq t_i \leq u_i   \big \}    \Omega
_{f}(l,t)\partial_{t}g(t)dt\;, \label{M1}%
\end{equation}
where $\Omega_{f}(l,t)$ denotes
\begin{equation}
 \Omega_{f}(l,t)=\int_{  \R^q }    \prod_{i=1}^q  \mathbbm{1} \{ l_i \leq x_i \leq t_i  \} f(x) d F_{X} (x) . \label{M2}%
\end{equation}
\end{lemma}

The next Lemma aims to interpret the moments appearing in Condition (\ref{hallcond}) through repeated applications of Lemma \ref{aux}. We begin by introducing some convenient notation. Let $\mu_{2}\left(
t\right)  ,\mu_{4}\left(  t\right)  $ be as in Assumption \ref{error}. Given $x,s,t \in \R_{}^q$, we define the cube centered at $x$ with directions $(s,t)$ to be
\[
B\left(  x-s,x+t\right)  =\left\{  y \in \R^q  :x_{i}-s_{i}\leq y_{i}\leq x_{i}+t_{i} \; \: \forall  \; i=1,...,q\right\}  .
\]
Define
\begin{align}
&  \Omega_{l}(x-s,x+t)=\int\limits_{B\left(  x-s,x+t\right)  }\mu_{l}%
(y)d F_{X}(y)\;\;\;\;\;\;\;\;\;l=2,4  \label{omega}\\
&  F_{X}(x-s,x+t)=\int\limits_{B\left(  x-s,x+t\right)  } d F_{X}\;.     \label{Fx}
\end{align}
The following Lemma expresses the moments in terms of functionals of $\Omega_{l}$.
\begin{lemma}
\label{moments} Suppose $\mu_{2}(X) \in L^1(X)$ and Assumptions (\ref{data}, \ref{kernel}) hold. Then 
\begin{enumerate}
\item[(i)]
\begin{align*}
  \mathbb{E}[H_{n}^{2}(Z_{1},Z_{2})]=\mathbb{E}\bigg(\mu_{2}(X)\int_{\left[  0,1\right]  ^{q}}\Omega
_{2}(X-h_{n}v,X+h_{n}v)\partial_{v}K^{2}(-v)dv\bigg)  \; ,
\end{align*}
\item[(ii)]
\begin{align*}
  \mathbb{E}[H_{n}^{4}(Z_{1},Z_{2})]=\mathbb{E}\bigg(\mu_{4}(X)\int_{\left[  0,1\right]  ^{q}}\Omega
_{4}(X-h_{n}v,X+h_{n}v)\partial_{v}K^{4}(-v)dv\bigg) \; ,
\end{align*}
\item[(iii)]
\begin{align*}
\mathbb{E}[G_{n}^{2}(Z_{1},Z_{2})]  &
 \lessapprox \mathbb{E}\bigg(  \mu_2(X)  \{ \Omega_2(X - h_n \iota , X +  h_n \iota)    \}^2 \int_{[-2,2]^q} \Omega_{2} ( X - 2 h_n \iota , X + h_n u     )     \\
&  \;\;\;\;\; \;   \times  \left| \partial_{u} \bigg[ \int_{  [-1,1]^q
}   \big(    \partial_{v} 
 \big[ K(v)K\big(v-u\big) \big]   \big)^2 \:  \mathbbm{1} \big \{  v-u \in [-1,1]^q      \big \}     dv \bigg]  \right|   du\bigg).
\end{align*}
\end{enumerate}
\end{lemma}

\begin{proof}
[Proof of Lemma \ref{moments}] 
\begin{align*}
\mathbb{E}[H_{n}^{2}(Z_{1},Z_{2})] &  =\mathbb{E}\bigg[\mu_{2}(X_{1})\mu
_{2}(X_{2})K^{2}\bigg(\frac{X_{1}-X_{2}}{h_{n}}\bigg)\bigg]\\
&  =\mathbb{E}\bigg[\mu_{2}(X_{2})\int_{\mathbb{R}^{q}}\mu_{2}(x)K^{2}%
\bigg(\frac{x-X_{2}}{h_{n}}\bigg)d F_{X} (x)\bigg].
\end{align*}
Define
\[
f(x)=\mu_{2}(x),\;g(x)=K^{2}\bigg(\frac{x-X_{2}}{h_{n}}\bigg).
\]
Let $X_2^i$ denote the $i^{th}$ coordinate of $X_2$. Conditional on $X_{2}$, $(f,g)$ satisfy the hypothesis of Lemma \ref{aux} with  $ \mathcal{O}  =  (X_{2}^{1}%
-h_{n},X_{2}^{1}+h_{n}) \times \dots \times (X_{2}^{q}%
-h_{n},X_{2}^{q}+h_{n}) $.  Applying Lemma \ref{aux} yields
\begin{align*}
\int_{\mathbb{R}^{q}}\mu_{2}(x)K^{2}\bigg(\frac{x-X_{2}}{h_{n}}\bigg)dF_{X}(x)
&  = (-1)^q \int_{\mathbb{R}^{q}} \mathbbm{1} \{ t \in \mathcal{O}  \}   \Omega_{2}(X_{2}-h_{n}\iota,t)\partial_{t}%
K^{2}\bigg(\frac{t-X_{2}}{h_{n}}\bigg)dt\\
&  = (-1)^q \int_{\left[  -1,1\right]  ^{q}}\Omega_{2}(X_{2}-h_{n}\iota,X_{2}%
+h_{n}v)\partial_{v}K^{2}(v)dv\;,
\end{align*}
where the last equality follows from the change of variables $t\rightarrow
X_{2}+h_{n}v$. It follows that
\[
\mathbb{E}[H_{n}^{2}(Z_{1},Z_{2})]=\mathbb{E}\bigg(\mu_{2}(X)(-1)^{q}%
\int_{\left[  -1,1\right]  ^{q}}\Omega_{2}(X-h_{n}\iota,X+h_{n}v)\partial
_{v}K^{2}(v)dv\bigg).
\]
Define $(v_{1},v_{-1})$ to be the partitioned
vector $(v_{1},\dots,v_{q})$ with $v_{-1}=\left(  v_{2},...,v_{q}\right)  $.
For any fixed choice of $v_{-1}\in\mathbb{R}^{q-1}$, we have
\begin{align*}
&  \int_{[-1,1]}\Omega_{2}(X-h_{n}\iota,X+h_{n}(v_{1},v_{-1}))\partial_{v_{1}}k^{2}(v_{1})dv_{1}\\
& \qquad\qquad =     \int_{[0,1]} \Omega_{2}(X-h_{n}\iota,X+h_{n}(v_{1},v_{-1}))\partial_{v_{1}}k^{2}(v_{1})dv_{1}  \\ 
  & \qquad\qquad\qquad\qquad
  +  \int_{[-1,0]}  \Omega_{2}(X-h_{n}\iota,X+h_{n}(v_{1},v_{-1}))\partial_{v_{1}}k^{2}(v_{1})dv_{1}  \\
& \qquad\qquad = \int_{[0,1]} \Omega_{2}(X-h_{n}\iota,X+h_{n}(v_{1},v_{-1}))\partial_{v_{1}}k^{2}(v_{1})dv_{1} \\
  & \qquad\qquad\qquad\qquad
  -\int_{[0,1]}\Omega_{2}(X-h_{n}\iota,X+h_{n}(-v_{1},v_{-1}))\partial_{v_{1}}k^{2}(v_{1})dv_{1} \\
& \qquad\qquad = \int_{[0,1]} \Omega_{2}(X-h_{n}(v_{1},\iota),X+h_{n}(v_{1},v_{-1}))\partial_{v_{1}}k^{2}(v_{1})dv_{1}\;,
\end{align*}
where the second  equality follows from the change of variables $v_1 \rightarrow -v_1$ and $  \partial_{v_1}  k^{2}%
(-v_{1})=- \partial_{v_1} k^{2}(v_{1})$ ($k$ is a symmetric function).
Iterating this procedure from $v_{1}$ to $v_{q}$ yields
\[
\int_{\lbrack-1,1]^{q}}\Omega_{2}(X-h_{n}\iota,X+h_{n}v)\partial_{v}%
K^{2}(v)dv=\int_{[0,1]^{q}}\Omega_{2}\big(X-h_{n}v,X+h_{n}v\big)\partial
_{v}K^{2}(v)dv.
\]
The expression for $\E(H_n^2)$ follows from substituting  $(-1)^q \partial_{v} K^2(v) = \partial_{v} K^2(-v)$. The derivation for $\E(H_n^4)$ is similar. The derivation for $\E(G_n^2)$  follows from repeated applications of Lemma \ref{aux} (further details provided in the supplementary file  \cite{sidvsup}).
\end{proof}
Given the form of the integrand that defines $\E(H_n^2)$ in Lemma \ref{moments}, it is expected (by Lebesgue's differentiation theorem)  that $\E(H_n^2) \asymp h_n^q$ whenever $F_{X} \ll $ Lebesgue measure. If the Lebesgue decomposition of $F_{X}$ admits singular components, the following Corollary shows  that $h_n^q$ may at least be used as a conservative lower bound on the rate.

\begin{corollary}
\label{momlimit}
Suppose $\mu_{2}(X) \in L^2(X)$,  Assumptions (\ref{data},\ref{kernel}) hold and $h_n \downarrow 0$. \begin{enumerate}
\item[(i)] If $F_X$ is absolutely continuous with respect to the Lebesgue measure and admits a density function $f_X \in L^{\infty}(X)$, then
\begin{equation*}
h_{n}^{-q}\mathbb{E}\left[ H_n^2(Z_1,Z_2) \right]  \xrightarrow[n \rightarrow \infty ]{}\mathbb{E}\bigg(\mu_{2}%
^{2}(X)f_{X}(X)\int_{[-1,1]^q}K^{2}(v)dv\bigg) . %
\end{equation*}

\item[(ii)] If additionally Assumption \ref{error}(iii) holds, then for any distribution $F_X$ we have that
\begin{equation*}
 \liminf_{n \rightarrow \infty}  h_{n}^{-q}  \mathbb{E}\left[ H_n^2(Z_1,Z_2) \right]   \geq \begin{cases}      \mathbb{E}\bigg[\mu_{2}%
^{2}(X)f_{X}(X)  \int \limits_{[-1,1]^q}K^{2}(v)dv\bigg]   & F_{X} \ll \; \text{Lebesgue} \; ,  \\ \infty & \text{else}   .         \end{cases}
\end{equation*}
\end{enumerate}
\end{corollary}

In particular, Corollary \ref{momlimit} extends the standard result (see e.g. \cite{Zhang}) for the limiting behavior of $\E(H_n^2)$ when the
Lebesgue density exists (although here we do not
assume that it is continuous) and demonstrates that $h_n^{-q} \E(H_n^2)$ diverges when there
are singular components. The next Lemma provides bounds on the moments that will be
instrumental in verifying Condition (\ref{hallcond}) and as a consequence the limit behavior of the U-statistic in $\left(
\ref{U(u)}\right)$. \begin{lemma}
\label{moments copy(2)} Let Assumptions \ref{data}-\ref{kernel} hold. Then,
given any $\varepsilon\in(0,1)$, we have
\begin{align*}
&  (i)\;\;\;b_{1}\mathbb{E}\big[  F_{X}(X-h_{n}\varepsilon\iota
,X+h_{n}\varepsilon\iota)  \big]\leq \mathbb{E}\left[ H_n^2(Z_1,Z_2) \right] \leq B_{2}\mathbb{E}\big[
F_{X}(X-h_{n}\iota,X+h_{n}\iota)  \big] \; ,\\
&  (ii)\;\;\;  \mathbb{E}\left[ H_n^4(Z_1,Z_2) \right]  \leq B_{3}\mathbb{E}\big[   F_{X}(X-h_{n}\iota,X+h_{n}%
\iota)  \big] \; ,\\
&  (iii)\;\;  \mathbb{E}\left[ G_n^2(Z_1,Z_2) \right]  \leq B_{4} \E \big[   \big\{ F_X(X - 2 h_n \iota ,X + 2 h_n \iota)   \big\}^3      \big]
\end{align*}
where $b_{1},B_{2},B_{3},B_{4}>0$ are finite universal constants that depend on $\mu_{2}\left(  x\right)  ,\mu_{4}\left(  x\right)  $
(Assumption \ref{error}) and  functionals of the kernel function  $k(\,\cdot\,)$.
\end{lemma}

\begin{proof}
[Proof of Lemma \ref{moments copy(2)}] 
Note that $\partial_{v} K^2(-v)\geq 0 $ for every $v \in [0,1]^q$. Fix any $\varepsilon
\in (0,1)$. From the expression defining $\E[H_n^2]$ in Lemma \ref{moments}, we obtain that
\begin{align*}
\mathbb{E}\left[ H_n^2(Z_1,Z_2) \right]  &
=\mathbb{E}\bigg(\mu_{2}(X)\int_{[0,1]^{q}}\Omega_{2}(X-h_{n}v,X+h_{n}%
v)  \partial_{v}K^2(-v) dv\bigg)\\
&  \geq\mathbb{E}\bigg(\mu_{2}(X)\int_{[\varepsilon,1]^{q}}\Omega_{2}%
(X-h_{n}v,X+h_{n}v)  \partial_{v} K^2(-v) dv\bigg)\\
&  \geq M_{1}\mathbb{E}\big[ \mu_{2}(X)  \Omega_{2} \big(X-\left(
h_{n}\varepsilon\right)  \iota,X+\left(  h_{n}\varepsilon\right)
\iota  \big)   \big] \\
&  \geq b^{2}M_{1} \E \big[  F_X \big(X -   (h_n \varepsilon) \iota , X +(h_n \varepsilon) \iota     \big)       \big]
\end{align*}
where $b$ is as in Assumption \ref{error} and $M_{1}=\int_{[\varepsilon,1]^{q}}  \partial_{v}K^2(-v) dv > 0$. The derivations for the other bounds are provided in the supplementary file  \cite{sidvsup}.

\end{proof}
The bounds on the moments derived in this section are in terms of the expected small ball probability $\E \big[   F_{X} (X -h \iota , X +h \iota)      \big]$. We next turn to defining classes of distributions where these bounds can be used
to provide limit properties of the statistic.

\subsection{Classes of distributions}
\label{sec3.2}
We begin by making an assumption that delineates a class of distributions for which the asymptotic normality of the kernel statistic will be established. As we show below, this
class encompasses some well-known distributions. We will refer to $F_{X}$ as a continuous measure if  $ F_X(x - h \iota , x + h \iota) = \mathbb{P} \big(  \| X - x \|_{\infty} \leq h \big) \downarrow 0  $ as $h \downarrow 0$ for every $x $ in the support of $F_{X}$.  \begin{assumption}
\label{d-class}
$F_{X}$ is a continuous measure that satisfies
\begin{align*}
& (i) \; \; \;\;\text{for some} \; \varepsilon \in (0,1) \; , \;  \limsup_{h \downarrow 0} \frac{\mathbb{E}\left[  F_{X}(X-h\iota,X+h\iota
)\right]  }{\mathbb{E}\left[  F_{X}(X-\left(  h\varepsilon\right)
\iota,X+\left(  h\varepsilon\right)  \iota)\right]  } < \infty \; ,
\\  & (ii) \; \; \;\;\
\lim_{h  \downarrow 0} \frac{\mathbb{E}\left[  \big\{ F_{X}(X-h\iota,X+h\iota)
 \big\}^3  \right]  }{\left(  \mathbb{E}\left[  F_{X}(X-h\iota,X+h
\iota)\right]  \right)  ^{2}} = 0 \: . %
\end{align*}

\end{assumption}
A stronger pointwise version of Assumption \ref{d-class}(i)  is commonly known as the ``doubling'' condition in the literature (see e.g. \cite{dbl1,ussr}). Specifically, doubling measures are exactly those that satisfy \begin{equation}  \mathbb{P} \bigg( F_{X}\big(X-h_{{}}\iota,X+h_{{}}\iota \big) \leq C F_{X} \big( X-0.5 h
\iota,X+0.5 h\iota \big) \bigg) = 1    \label{bou1}
\end{equation}
for some  universal constant $  C <\infty $ and all sufficiently small $h > 0$.

A stronger condition which implies both parts of Assumption \ref{d-class} is that $F_{X}$ be Ahlfors-David regular (see e.g. \cite{ahl1}) on the support of $F_X$. These are precisely the measures where there exists a $s \in (0,1] $ for which \begin{equation} \mathbb{P} \big( C^{-1} h^{sq} \leq  F_{X}(X-h_{{}}\iota,X+h_{{}}\iota) \leq D h^{sq} \big) = 1 \label{bou2}  \end{equation}
holds for some  universal constants $C,D < \infty $ and all sufficiently small $h > 0 $. 

We note that while conditions (\ref{bou1}, \ref{bou2}) are sufficient for Assumption \ref{d-class} to hold, they are not necessary.  In particular, the existence of a universal $C > 0 $ that satisfies (\ref{bou1}, \ref{bou2}) may be excessively restrictive when $X$ is not compactly supported.  Nonetheless, Assumption \ref{d-class} may still hold in this case as the assumption only depends on expectations of the distribution. To expand on this point, we introduce a rich class of distributions that extends beyond the absolutely continuous and Ahlfors-David regular subclass. Define $$ \underline{F_{X}} (x,s)  =  \liminf_{h \downarrow 0} \frac{F_X(x - h \iota , x + h \iota)}{(2h)^{sq}} . $$
\begin{definition} \label{ds} For every $s \in (0,1]$, let $\mathcal{D}\left(  s\right)  $ denote the class of probability measures that satisfy
\begin{equation}
(i)  \; \; \; \; \; \;   \mathbb{P} \bigg(      \frac{ F_{X}(X-h\iota,X+h\iota)}{(2h)^{sq}} \leq M_{F_{X}} \bigg)  = 1 \; \; , \; \; (ii) \; \; \; \;  \;  \E \big[ \, \underline{F_{X}} (X,s)    \big]  > 0    \label{fractal}%
\end{equation}
for some constant $ M_{F_{X}} < \infty $ and all sufficiently small $h > 0 $. By varying  the singularity exponent $s$, we denote the class of all such distributions by $\mathcal{D} = \bigcup_{s \in (0,1]} \mathcal{D}(s)$.
\end{definition}
The following four examples demonstrate that the class $\mathcal{D}$ includes a wide range of distributions that
may underlie various cases of interest in economics, finance and natural sciences.

\begin{example}[Absolutely continuous measures] \label{ex1}
Suppose $F_{X}$ is absolutely continuous with respect to the Lebesgue measure on $\R^q$ and admits a density function $f_{X} \in L^{\infty}(X)$.  By Lebesgue's differentiation theorem, we obtain that$$   \lim_{h \downarrow 0} \frac{F_X(x - h \iota , x + h \iota)}{(2h)^{q}}  = f_X(x)$$
almost everywhere with respect to the Lebesgue measure. As a consequence, $F_X \in \mathcal{D}(1)$ where (\ref{fractal}) holds with $ \underline{F_{X}}(x,1) = f_X(x), \; M_{F_{X}} = \| f_X \|_{L^\infty}$. 
 \label{ex1}
\end{example}
In particular, Example \ref{ex1} allows for absolutely continuous measures that admit a discontinuous density function.

\begin{example}[Self-Similar Fractals]
Consider a  contraction mapping $%
\phi :\mathbb{R}^{q}\rightarrow \mathbb{R}^{q}$ such that $\Vert \phi
(x)-\phi (y)\Vert _{2}=r\Vert x-y\Vert _{2}$ for all $x,y\in \mathbb{R}^{q}$
and some fixed $r\in (0,1)$. Let $\mathcal{S}=\{\phi _{1},\dots ,\phi _{N}\}$
denote a family of contraction maps with contraction ratios $\{r_{1},\dots
,r_{N}\}$. There is a unique compact set $\mathcal{K}$ (see e.g. \cite[Proposition 2.30]{ergodic}) that is invariant with respect to $\mathcal{S}$, in the sense that $\mathcal{K} = \bigcup_{i=1}^N \phi_{i}(\mathcal{K})$. Denote the similarity dimension of $\mathcal{S}$ by the unique $D$ for which $%
\sum_{i=1}^{N}r_{i}^{D}=1$. $\mathcal{S}$ is said to satisfy the open set condition
(OSC) if there exists a nonempty open set $O$ such that $\bigcup_{i=1}^{N}%
\phi _{i}(O)\subseteq O$ and $\phi _{i}(O)\cap \phi _{j}(O)=\emptyset$ for $%
i\neq j$.  For any $\mathcal{S}$ that satisfies the OSC, it is known (see e.g. \cite[Section 5]{hutch}) that $(i)$ the Hausdorff dimension of $\mathcal{K}$ is the similarity dimension $D$ and $(ii)$ the probability measure induced from the restriction of the $D$-dimensional Hausdorff measure to $\mathcal{K}$ is Ahlfors-David regular (\ref{bou2}) with $sq = D$. Frequently referenced
examples include the Cantor set (the Cantor measure is the restriction of the $D=\log _{3}(2) \approx 0.631$ Hausdorff measure),
Sierpi\'{n}ski's Triangle ($D\approx 1.585$) and the Koch Snowflake ($%
D = \log _{3}(4) \approx 1.262$). \label{ex2}
\end{example}
Fractal measures  are useful for modelling data that describe
processes that are similar at different scales, used frequently in the natural sciences and the analysis of spatial data (see e.g. \cite{burrough1981fractal,davies1999fractal}).

\begin{example}[Measures supported on a low dimensional subspace] Suppose $ X \sim N(0, \Sigma)$, where $\Sigma \in \R^{q \times q}$ is singular. Let $r > 0 $ denote the rank of $\Sigma$. We can write $\Sigma =  Q D Q'$ where $Q \in \R^{q \times q}$ is an orthogonal matrix and   $D  = \text{Diag}(\sigma_{1}^2 , \dots , \sigma_{r}^2 , 0 , \dots 0) $ for some positive constants $(\sigma_{1}^2 , \dots, \sigma_{r}^2)$. Hence, $Q'X \overset{d}{=}  Z $ where $Z  \sim N(0,D)  $ and the support of  $F_{X}$ is the $r$ dimensional subspace $S = \{ x \in \R^q : Q'x  \in \R^{r} \times \{0 \}^{q-r}   \} $. For any fixed $t \in S$, we have that $ \mathbb{P} ( \|  X - t  \|_{2} \leq h ) = \mathbb{P} \big(   \sum_{i=1}^r  (   Z_i - [Q't]_{i}    )^2   \leq h^2         \big) $. From Example \ref{ex1} and the equivalence of the $( \|. \|_{2}, \|. \|_{\infty} )$ norms, we obtain (\ref{fractal})  with  $ s = r/q  $.
\label{ex3}
\end{example}

Example \ref{ex3}  can be generalized further to allow for a general $r$-rectifiable (see e.g. \cite{preiss,de2006lecture,moore1950density}) measure.\footnote{A measure $\nu(.)$ on $\R^q$ is  $r$-rectifiable if there exists a  Borel measurable function $f(.)$ and a countable collection of $r$-dimensional $C^1$ submanifolds $\{M_i \}_{i=1}^{\infty}$  such that $ \nu(A) = \sum_{i=1}^{\infty}  \int _{M_i \cap A} f(x) d \text{Vol}^r(x)$ for every Borel set $A$, where $\text{Vol}^r$ is the natural $r$-dimensional volume measure that a $C^1$ submanifold inherits as a subset of $\mathbb{R}^q$.} These are low dimensional measures in the sense that there exists a countable collection of $r$-dimensional $C^1$ submanifolds $\{ M_i \}_{i=1}^{\infty}$  such that  $ \mathbb{P}\big( X \in    \bigcup_{i=1}^{\infty} M_i     \big) = 1$.
\begin{example}[Discrete Regressors and mass points]
\label{ex4}
Suppose $ X = (W,Z)  $ where the law of $W$  is a continuous measure $F_{W}$ on $\R^{r}$ for some $1 \leq r  \leq q $ and the law of $Z$ assigns positive mass to some $z^{\star} \in \R^{q-r}$. Suppose $F_{W}$ admits a  density $f_{W} \in L^{\infty}(W)$ and  the conditional distribution measure $ \nu(.) =  \mathbb{P}(W \in .  \,|Z = z^{\star})$ admits a  density $f_{W|Z=z^{\star}}$. For any  fixed $x=(w,z)$,  we have  that

 \begin{align*} & F_X(x-h \iota , x + h \iota) = \mathbb{P}(\| X - x \|_{\infty} \leq h) \leq  \mathbb{P} \big( \| W-w    \|_{\infty} \leq h  \big) .
\end{align*}
From Example \ref{ex1}, it follows that the first requirement of (\ref{fractal}) holds with $s=r/q$. For the second requirement, we note that if $x=(w,z^{\star})$ is in the support of $F_{X}$, then \begin{align*} F_X(x-h \iota , x+h \iota)  & \geq  \mathbb{P}(\| W-w \|_{\infty} \leq h , Z= z^{\star}  ) \\ &  = \mathbb{P}(\| W-w \|_{\infty} \leq h |Z=z^{\star}) \mathbb{P}(Z=z^{\star}) .
\end{align*}
Denote the conditional (at $Z= z^{\star}$) measure of a cube with radius $h$ and centered at $w$ by $F_{W|Z = z^{\star}} (w - h \iota , w + h \iota) =   \mathbb{P}(\| W-w \|_{\infty} \leq h |Z=z^{\star}) $.   Since the conditional measure admits a  density, the second requirement of (\ref{fractal}) follows from

 \begin{align*}
   \E \bigg[ \liminf_{h \downarrow 0}   \frac{ F_{X}(X-h\iota,X+h\iota)}{(2h)^{r}}      \bigg]  &   \geq \{ \mathbb{P}(Z=z^{\star}) \}^2 \underset{ W|Z=z^{\star} }{\E} \bigg[ \liminf_{h \downarrow 0}   \frac{ F_{W|Z=z^{\star}}(W-h\iota,W+h\iota)}{(2h)^{r}}      \bigg]     \\ & = \{ \mathbb{P}(Z=z^{\star}) \}^2 \underset{ W|Z=z^{\star} }{\E} \big[  f_{W| Z = z^{\star}}   \big]   \\ & > 0.
\end{align*}
\end{example}
In Example \ref{ex4}, $Z$ may be fully discrete (the support of $Z$ is a countably infinite set $\mathcal{Z} \subset \R^{q-r}$) or $Z = (Z_1,Z_2)$ where $Z_1$ is fully discrete and $Z_2$ is a mixture of a continuous and discrete variable. Example \ref{ex4} can be generalized in a straightforward way to allow for the continuous measure $F_{W}$ to contain singular components (the argument is identical if we insist that $ F_{W} $ and $W|Z=z^{\star}$ are elements of $\mathcal{D}(t)$ for some $t \in (0,1]$). To the best of our knowledge, Example \ref{ex4} extends the known results in the literature (e.g. \cite{hsiao2007}) to allow for  discrete regressors with countably infinite support, mixed regressors and continuous regressors whose joint law may have singular components .  In this case, our approach also  highlights that there are possible advantages to viewing the joint distribution of mixed data containing discrete and continuous variables as a continuous singular measure (which can be analyzed directly).

 The next result shows that every $ F_X \in  \mathcal{D} $ satisfies Assumption \ref{d-class}. Moreover, $\mathcal{D}$ is closed under mixtures, so that any mixture combination of its elements (such as the distributions in Example \ref{ex1}-\ref{ex4}) satisfies Assumption \ref{d-class} as well.

\begin{lemma}
\label{d-lemma} 
$(i)$ Every $F_X \in \mathcal{D} = \bigcup_{s \in (0,1]} \mathcal{D}(s)$  satisfies Assumption \ref{d-class}. $(ii)$  Suppose $F_{t} \in \mathcal{D}(s_t)$ (with constant $M_t$ as in Definition \ref{ds}) for every $t$  in some countable  index set $T$, $ s = \inf_{t \in T} s_t > 0 $ and the set $\{ t \in T : s_t = s \}$ is non-empty. If $(\alpha_t)_{t \in T}$  is a weight vector satisfying $\alpha_t > 0 $, $\sum_{t \in T} \alpha_t = 1$ and $\sum_{t \in T} \alpha_t M_t < \infty$, then  $ F_{X} = \sum_{t \in T} \alpha_t \: F_{t} \in \mathcal{D}( s )$.

\end{lemma}

In particular,  by mixing absolutely continuous measures with elements of  $ \bigcup_{s \in (0,1)} \mathcal{D}(s)$, we obtain a large class of distributions that admit  a non-trivial Lebesgue decomposition  and satisfy Assumption \ref{d-class}. Finally, we remark that it is not necessary to argue for Assumption \ref{d-class} through membership in $\mathcal{D}$. In more complicated setups, a measure may not charge cubes in a way that satisfies Definition \ref{ds}. In such cases, Assumption \ref{d-class} should be verified directly.

\subsection{Limit theory}
\label{sec3.3}

In the next result, we establish asymptotic normality for the U-statistic $U_{n}$ in
(\ref{U(u)}) for every distribution $F_X $ that satisfies Assumption \ref{d-class}.  For the remainder of Section \ref{sec3}, we assume that the bandwidth sequence $(h_n)_{n=1}^{\infty}$ satisfies $h_n \downarrow 0$ and $n h_n^q \uparrow \infty$.

\begin{theorem}
\label{asymptotics} If Assumptions \ref{data}-\ref{d-class} hold, then

\[
\frac{nU_{n}}{\sqrt{2\mathbb{E}\left(  H_{n}^{2}\right)  }}%
\underset{d}{\rightarrow}N\left(  0,1\right)  .
\]
Furthermore, if $ F_{X} \in\mathcal{D}\left(  s\right)  $ for some $s \in (0,1]$
and $\sigma^{2}=\underset{n\rightarrow\infty}{\lim}2h_{n}^{-sq}\mathbb{E}%
\left(  H_{n}^{2}\right)  $ exists, then
\[
nh_{n}^{-sq/2}U_{n}\underset{d}{\rightarrow}N (  0,\sigma^{2} )  .
\]

\end{theorem}

\begin{proof}
[Proof of Theorem \ref{asymptotics}] 
The result follows from Theorem $1$ of \cite{Hall84} if Condition
 (\ref{hallcond}) holds. We aim to verify that $$  \text{(a)} \; \; \;  \;  \frac{\E(G_n^2)}{[ \E(H_n^2)  ]^2} \downarrow 0 \; \; ,  \; \; \; \text{(b)}  \;\; \; \; \frac{\E(H_n^4)}{n [ \E(H_n^2)]^2} \downarrow 0. $$
 From Assumption
\ref{d-class}(i), there exists $\varepsilon \in (0,1)$ and  $c^{\ast},\bar{h}>0$
such that
\begin{align}
   c^{\ast}\mathbb{E}\left[  F_{X}(X-h\iota,X+h\iota)\right]     \leq\mathbb{E}\left[  F_{X}(X-h\varepsilon\iota,X+h\varepsilon
\iota)\right]  \label{a5eq}
\end{align}
holds for all $h \in (0, \bar{h})$.
 \begin{enumerate}
 \item[(a)] From Lemma \ref{moments copy(2)}, we obtain  that $$  \frac{\E(G_n^2)}{[ \E(H_n^2)  ]^2} \lessapprox \frac{\E \big[   \big\{ F_X(X - 2 h_n \iota ,X + 2 h_n \iota)   \big\}^3      \big]  }{\big\{ \mathbb{E}\big[  F_{X}(X-h_{n}\varepsilon\iota
,X+h_{n}\varepsilon\iota)  \big] \big\}^2} .  $$
Let $\gamma\in\mathbb{N}$ be such that $\varepsilon^{-\gamma}>2$. Since
$h_{n} \downarrow 0$, we have $h_{n}\varepsilon^{-\gamma}<\bar{h_{{}}}$ for all
sufficiently large $n$. Repeated applications of equation (\ref{a5eq}) yield
\begin{align}
 \mathbb{E}\left[  F\left(  X-h_{n}\varepsilon\iota,X+h_{n}\varepsilon
\iota\right)  \right]  )   &  \geq c^{\ast}\mathbb{E}\left[  F_{X}(X-h_{n}%
\iota,X+h_{n}\iota)\right]  \nonumber \\
&  \geq(c^{\ast})^{2}\mathbb{E}\left[  F_{X}(X-h_{n}\varepsilon^{-1}%
\iota,X+h_{n}\varepsilon^{-1}\iota)\right]  \nonumber \\
&  \dots \nonumber \\
&  \geq(c^{\ast})^{\gamma+1}\mathbb{E}\left[  F_{X}(X-h_{n}\varepsilon
^{-\gamma}\iota,X+h_{n}\varepsilon^{-\gamma}\iota)\right] \nonumber  \\
&  \geq(c^{\ast})^{\gamma+1}\mathbb{E}\left[  F_{X}(X-2h_{n}\iota
,X+2h_{n}\iota)\right]   . \nonumber 
\end{align}
From the preceding inequality and Assumption \ref{d-class}(ii), we obtain \begin{align*}
 \frac{\E(G_n^2)}{[ \E(H_n^2)  ]^2} \lessapprox \frac{\E \big[   \big\{ F_X(X - 2 h_n \iota ,X + 2 h_n \iota)   \big\}^3      \big]  }{\big\{ \mathbb{E}\big[  F_{X}(X-2h_{n}\varepsilon\iota
,X+2h_{n}\varepsilon\iota)  \big] \big\}^2}  \downarrow 0.
 \end{align*}
 
 \item[(b)]  
 
From Lemma \ref{moments copy(2)} and (\ref{a5eq}) we obtain   \begin{align*}
\frac{\E(H_n^4)}{ [ \E(H_n^2)]^2} \lessapprox  \frac{\E\big[  F_X(X-h_n \iota, X+ h_n \iota)   \big]}{ \{ \E\big[  F_X(X-h_n  \varepsilon \iota, X+ h_n \varepsilon \iota)   \big]     \}^2} \lessapprox \frac{1}{ \E\big[  F_X(X-h_n \iota, X+ h_n \iota)   \big]}.
\end{align*}
Let $f_{X}$ denote the density  if $F_{X}  \ll $ Lebesgue measure.  To lower bound the denominator, Fatou's Lemma and a straightforward application of  \cite[Theorem 7.15]{rudin} yields \begin{align*}   \liminf_{n\rightarrow \infty} (2h_n)^{-q}  \E\big[  F_X(X-h_n \iota, X+ h_n \iota)   \big] &  \geq  \E \bigg[  \liminf_{n \rightarrow \infty}   \frac{F_X(X - h_n \iota , X + h_n \iota)}{(2h_n)^q}      \bigg]  \\ &       = \ \begin{cases}    \E[f_X(X)] & F_{X} \ll \; \text{Lebesgue} \; ,  \\ \infty & \text{else} . \end{cases}    
\end{align*}
In particular $   \E\big[  F_X(X-h_n \iota, X+ h_n \iota)   \big] \gtrapprox  h_n^{q}$ always holds. The claim follows from the assumption that   $n h_n^{q} \uparrow \infty$.
 
 \end{enumerate}

\end{proof}

\begin{remark}[On the existence of $\sigma^2$] Let $\lambda$ denote the Lebesgue measure on  $[0,1]^q$. Suppose Assumption \ref{error} holds and $F_X \in \mathcal{D}(s)$ for some $s \in (0,1]$. By Lemma \ref{moments}, we have that $$ \mathbb{E} \big( H_{n}^{2} \big)=\mathbb{E}\bigg(\mu_{2}(X)\int_{\left[  0,1\right]  ^{q}}\Omega
_{2}(X-h_{n}v,X+h_{n}v)\partial_{v}K^{2}(-v)dv\bigg). $$ 
This suggests that $\sigma^{2}=\underset{n\rightarrow\infty}{\lim}2h_{n}^{-sq}\mathbb{E}%
\left(  H_{n}^{2}\right)  $  may be well defined whenever   \begin{equation} f(x,v) = \lim_{h \downarrow 0} \frac{
\Omega_{2} \big( x - hv  , x + hv            \big)}{ (2h)^{sq} } \label{lim}
\end{equation}
exists almost everywhere (with respect to $F_X \otimes \lambda    ) $. For absolutely continuous measures and low dimensional measures such as in Example \ref{ex3}, this follows immediately from an application of Lebesgue's differentiation theorem.  In the general case, it is known  from \cite{preiss} that if $ \bar{f}(X) = f(X,\mathbf{1})$ exists $F_X$ almost everywhere with $ \mathbb{P}( \bar{f}(X) \in (0, \infty)) =1$, then $s=r/q$ for some integer $r \in \mathbb{N}$ and there exists a countable collection of $r$-dimensional $C^1$ submanifolds $\{ M_i \}_{i=1}^{\infty}$  such that  $ \mathbb{P}\big( X \in    \bigcup_{i=1}^{\infty} M_i     \big) = 1$. In particular, when $sq$ is not an integer, the limit in (\ref{lim}) does not always exist and in this case $ (2h_n)^{-sq} \Omega_{2} \big( x - (h_nv) \iota , x + (h_nv) \iota           \big) $ typically oscillates  between its limit inferior and superior.
\end{remark}

\begin{remark}[On bandwidth constraints]
\label{bandrem}
Our limit theory   depends on the assumption that the bandwidth sequence $(h_n)_{n=1}^{\infty}$ satisfies $n h_n^{q} \uparrow \infty $. While this is standard for the absolutely continuous case, a closer inspection of our proofs reveal that all our main results go through under the assumption that $ n  \E\big[  F_X(X-h_n \iota, X+ h_n \iota)   \big]  \uparrow \infty $. From the general bound $  \E\big[  F_X(X-h_n \iota, X+ h_n \iota)   \big] \gtrapprox h_n^{q} $ (with strict dominance in the presence of singular components), we see that this is a weaker requirement on the bandwidth. In particular, when singular components exist, the bandwidth can approach zero at a faster rate than in the fully absolutely continuous case. In most cases, this more relevant bandwidth restriction cannot be used directly as the exact rate depends on knowledge of the singular contamination. One situation where the weaker constraint can be interpreted directly is the case of mixed data with absolutely continuous and discrete regressors (Example \ref{ex4}). In this case, it reduces to $n h_n^{r} \uparrow \infty$ (where $r$ denotes the dimension of the absolutely continuous regressors).
\end{remark}

\begin{remark}[On discrete measures and Assumption \ref{d-class}]
\label{discrete-m}
 As illustrated in Example \ref{ex4},  Assumption \ref{d-class} allows for  discrete distributions and mass points in a subset of the marginals, provided that at least one of the marginals is a continuous distribution. If there is a mass point in the entire distribution (equivalently the Lebesgue decomposition of $F_X$ contains a discrete measure), Assumption \ref{d-class}(ii) will be violated. In this case, the limit distribution of the U-statistic $U_{n}$ in
(\ref{U(u)})  is  typically Non-Gaussian. Indeed, suppose $F_X      =\rho_{d} F_{X}^d  +\rho_{a.c.}%
F_X^{a.c.}  + \rho_{s}  F_X^{s} $ where $\rho_d > 0$ and $F_X^d$ is a discrete measure with finite support $ \mathcal{S} =  \{ x_1 , \dots , x_D   \}$. First, we observe that scaling by $1/\sqrt{\E(H_n^2)}$ does not provide any additional rate self-normalization as it converges to a positive limit:

$$ \lim \limits_{n \rightarrow \infty} \E(H_n^2)  = K^2(0) \rho_d^2 \underset{X_1,X_2  \stackrel{i.i.d}{\sim} F_{X}^d }{\mathbb{E}} \big[  \mu_2(X_1) \mu_2(X_2) \mathbbm{1} \{ X_1 = X_2  \}  \big]  > 0 .  $$ It is straightforward to verify that the U-statistic can be expressed as
\begin{align*}
     U_n & =  \frac{ K(0) }{n(n-1)} \sum_{i=1}^n \sum_{j \neq i}  u_i u_j \mathbbm{1} \{ X_i = X_j\} +  \frac{ 1 }{n(n-1)} \sum_{i=1}^n \sum_{j \neq i} u_i u_j  K \bigg(  \frac{X_i - X_j}{h_n}   \bigg)  \mathbbm{1} \{ X_i \neq  X_j\} \\ & = \frac{ K(0) }{n(n-1)} \sum_{i=1}^n \sum_{j \neq i} u_i u_j \mathbbm{1} \{ X_i = X_j\} + o_{ \mathbb{P}}(n^{-1}).
\end{align*}
The first term on the right is non-trivial when the Lebesgue decomposition of $F_X$ contains a discrete measure. Moreover, it can be viewed as a degenerate U-statistic with a fixed  symmetric kernel  function $h[  (x_1,u_1) , (x_2,u_2)  ] = u_1 u_2 \mathbbm{1} \{ x_1 = x_2  \} $. From \cite[Theorem 12.10]{vdv}, it follows that $ n U_n \underset{d}{\rightarrow} K(0)   \sum_{k=1}^{\infty} \lambda_k (Z_k^2 -1) $ where    $Z_k \stackrel{i.i.d}{\sim} N(0,1) $
and $(\lambda_k)_{k=1}^{\infty}$ are eigenvalues corresponding to the integral operator defined by $h(\,\cdot\,)$.
\end{remark}

\subsection{Specification Testing}
\label{sec3.4}

Consider the feasible statistic $\hat{I}_{n}$ in (\ref{I(n)}) that differs
from $U_{n}$ in that the  residuals $\hat{u}_{i} = Y_i - g(X_i, \hat{\beta})$ replace the true errors $u_{i}$. To construct the goodness-of-fit test statistic $\hat{\tau}_{n}$ in (\ref{selfn}), we will also require a feasible analog of $2 \E(H_n^2)$.  To that end, define \begin{equation}
\hat{\sigma}_{n}^{2}=\frac{2}{n(n-1)}\sum_{i=1}^{n}\sum_{j\neq i}\hat{u}%
_{i}^{2}\hat{u}_{j}^{2}K^{2}\left(  \frac{X_{i}-X_{j}}{h_{n}}\right)
\;. \label{sigma-hat}%
\end{equation}
Let $\nabla_{\beta} g(x,  \beta)$ and $\nabla_{\beta}^2 g(x, \beta)$ denote the Gradient and Hessian of $g(x,\beta)$ (with respect to $\beta$), respectively. To derive the limit theory under $H_0$, we impose the following regularity conditions on the model and parameter estimator.
\begin{assumption}
\label{greg}(i) $ \| \hat{\beta} - \beta_0  \|_{2} = O_{\mathbb{P}}(n^{-1/2}) $. (ii) In a neighborhood $\mathcal{N}$ of $\beta_0$,  the map $ \beta \rightarrow g(x,\beta)$ is twice continuously differentiable for every $x$ in the support of $F_X$. (iii) In a neighborhood $\mathcal{N}$ of $\beta_0$, $\| \nabla_{\beta} g(x,\: .)  \|_{2}$ and $\| \nabla_{\beta}^2 g(x,\: .) \|_{op}  $  are dominated by functions $M(X) \in L^2(X)$ and $G(X) \in L^{4/3}(X)$, respectively.
\end{assumption}
Assumption \ref{greg} is commonly imposed in the literature (e.g. \cite{Zhang, phillips}). For distributions in $\mathcal{D}$, all our results in this section also hold under the weaker requirement that $G(X) \in L^1(X)$. The theorem below establishes the limiting distribution for the self-normalized goodness-of-fit statistic $\hat{\tau}_{n}$ in (\ref{selfn}) under $H_0$.
\begin{theorem}
\label{asymptotics2}
Suppose the null hypothesis $H_0$ and  Assumptions \ref{data}-\ref{greg} hold. 	Then $$\frac{n \hat{I}_{n}}{\sqrt{\hat{\sigma}_{n}^{2}}}\underset{d}{\rightarrow}N\left(
0,1\right)  .
$$
Furthermore, if $ F_{X} \in\mathcal{D}\left(  s\right)  $ for some $s \in (0,1]$
and $\sigma^{2}=\underset{n\rightarrow\infty}{\lim}2h_{n}^{-sq}\mathbb{E}%
\left(  H_{n}^{2}\right)  $ exists, then
\[
nh_{n}^{-sq/2} \hat{I}_n   \underset{d}{\rightarrow}N (  0,\sigma^{2} )  .
\]

\end{theorem}
Thus, the  test statistic converges weakly to a standard Gaussian and does not depend on any nuisance parameters. When singular components exist, the rate at which the distribution of the statistic
approaches the limit Gaussian could be quite slow, in particular for
$\E(H_n^2)$ converging to zero slowly (e.g. for small $s$ in $\mathcal{D}\left(
s\right)  )$.  If knowledge of the singular components is assumed, one can in principle choose a bandwidth sequence that approaches zero at a faster rate than usual (see Remark \ref{bandrem}) to improve on the rates.

To investigate the  asymptotic power of the test, we  consider the sequence of local alternative models \begin{equation}
H_{1}:Y_{}=g_{}\left(  X_{},\beta_{0}\right)  +\gamma_{n}\delta\left(
X_{}\right)  +u_{} \; ,  \label{reg}
\end{equation}
where $\gamma_n \downarrow 0$ is a deterministic sequence of constants and $\delta(.) $ is a real-valued drift function that determines the direction of approach to the null model.

We will continue to assume (as is standard) that Assumption \ref{greg}(i) holds under $H_1$. Indeed, when the residuals are computed using non-linear least squares (NLS) and the usual regularity conditions to ensure consistency under $H_0$ hold, the estimator continues to admit (under $H_1$) the asymptotic linear expansion \begin{align*}
\sqrt{n}(\hat{\beta} - \beta_0) & =   -\{ \E (   [\nabla_{\beta} g(X,\beta_0) ] [\nabla_{\beta} g(X,\beta_0) ]'    ) \}^{-1} \bigg(  \frac{1}{\sqrt{n}} \sum_{i=1}^n \nabla_{\beta} g(X_i,\beta_0) u_i       \bigg) + o_{\mathbb{P}}(1) \\ & = O_{\mathbb{P}}(1).
\end{align*}
To facilitate the derivation we make an assumption on the moments additional to Assumption \ref{d-class}.
\begin{assumption}
\label{d-class2}
$F_{X}$ is a continuous measure that satisfies
\begin{align*}
\lim_{h  \downarrow 0} \frac{\mathbb{E}\left[  \big \{ F_{X}(X-h\iota,X+h\iota)
 \big \}^2  \right]  }{\left(  \mathbb{E}\left[  F_{X}(X-h\iota,X+h
\iota)\right]  \right)  ^{3/2}} = 0.
\end{align*}
\end{assumption}

This assumption is similar to Assumption \ref{d-class}(ii). By arguing as in Lemma \ref{d-lemma}, it is straightforward to deduce that Assumption \ref{d-class2} holds for every distribution in the class $\mathcal{D}$. 

The next theorem provides the local power analysis of the specification test.  In general, it depends on the interplay between $(i)$ the rate $\gamma_n$ at which the local alternatives approach the null, $(ii)$ the drift function $\delta(.)$ that determines the direction of approach to the null model and $(iii)$ the distribution $F_{X}$.

\begin{theorem} \label{pow1}   Suppose the alternative hypothesis $H_1$ and Assumptions \ref{data}-\ref{d-class2} hold. Suppose  $\delta(.) \in L^{\infty}(X)$  is uniformly continuous on the support of $F_X$. 
 \begin{enumerate}
\item[(i)] If  $\gamma_n = o( n^{-1/2} \{ \E[ F_{X}(X - h_n \iota , X + h_n \iota) ]   \}^{-1/4} )  $, then $$ \frac{n \hat{I}_{n}}{\sqrt{\hat{\sigma}_{n}^{2}}} \overset{d}{=} N(0,1) + o_{\mathbb{P}}(1).  $$
\item[(ii)] If $\gamma_n \asymp  n^{-1/2} \{ \E[ F_{X}(X - h_n \iota , X + h_n \iota) ]   \}^{-1/4}$ and
 \begin{align}
 \liminf_{h \downarrow 0} \frac{ \E \big[  \delta^2(X) F_X(X - h \iota , X +h \iota)               \big]}{\E[ F_X(X -h \iota , X + h \iota) ]} > 0 \; ,  \label{liminf}
\end{align}
then there exists a determinstic sequence of constants $(L_n)_{n=1}^{\infty}$ such that  \begin{align*}
 & \frac{n \hat{I}_{n}}{\sqrt{\hat{\sigma}_{n}^{2}}} \overset{d}{=} N(0,1) + L_n + o_{\mathbb{P}}(1)  \; \;, \; \; \liminf_{n \rightarrow \infty} L_n > 0.
\end{align*}
However, if  \begin{align} \limsup_{h \downarrow 0} \frac{ \E \big[  \delta^2(X) F_X(X - h \iota , X +h \iota)               \big]}{\E[ F_X(X -h \iota , X + h \iota) ]} = 0 \; , \label{limsup} \end{align}
then, $$ \frac{n \hat{I}_{n}}{\sqrt{\hat{\sigma}_{n}^{2}}} \overset{d}{=} N(0,1) + o_{\mathbb{P}}(1).  $$
\end{enumerate}
\end{theorem}

Thus, under a sequence of local alternatives, the fastest rate at which asymptotic power may exist is $ n^{-1/2} \{ \E[ F_{X}(X - h_n \iota , X + h_n \iota) ] \}^{-1/4} $. Clearly (\ref{liminf}) is satisfied whenever $\delta^2(X)$ is bounded away from zero (up to a $F_{X}$ null set). If $F_{X}$ is absolutely continuous with density $f_X \in L^{\infty}(X)$, we have $\E[ F_X(X -h \iota , X + h \iota) ] \asymp h^q$  and (\ref{liminf}) reduces (by Lebesgue differentiation) to the usual condition  $$   \lim_{h \downarrow 0} \frac{ \E \big[  \delta^2(X) F_X(X - h \iota , X +h \iota)               \big]}{(2h)^q}   =       \E[\delta^2(X) f_X(X)] > 0 . $$ 

 The implications of Theorem \ref{pow1} are more complex for distributions with  singular components when $\delta^2(X)$ is not bounded away from zero. Loosely speaking, (\ref{liminf}) says that the support of $\delta(.)$ must intersect nontrivially with a subset of the support of  the distribution where the local singularity is maximized \big(support points $x$ where the decay rate of $ F_{X}(x - h_n \iota , x + h_n \iota)$ coincides with the rate for $\E[ F_{X}(X - h_n \iota , X + h_n \iota)]$ \big). The following example illustrates the essential idea of how the direction of approach may influence the local power properties of the test statistic.

\begin{example}[On the direction of approach] \label{supportr}
Suppose  $ F_X =  0.5 F_{N(0,1)} +  0.5 F_{C}   $ where $F_{N(0,1)}$ is standard Gaussian and $F_{C}$ denotes the usual Cantor measure on $[0,1]$. Let $ s = \log _{3}(2)\approx 0.631 $ and denote the support of $\delta(.) \in L^{\infty}(X)$ by $\mathcal{S}_{\delta}$. In this case, $F_{X} \in \mathcal{D}(s)$ and $ \E[ F_X(X -h_n \iota , X + h_n \iota) ] \asymp h_n^{s}  $. If there exists $\epsilon > 0 $ such that $ \mathcal{S}_{\delta} \subseteq  (- \infty , - \epsilon] \cup [1 + \epsilon , \infty)  $, then for any sequence $h_n \downarrow 0$ we have that \begin{align*}
\frac{ \E \big[  \delta^2(X) F_X(X - h_n \iota , X +h_n \iota)               \big]}{ \E[ F_X(X -h_n \iota , X + h_n \iota) ]} \lessapprox   \frac{ \E \big[  \delta^2(X) F_{N(0,1)}(X - h_n \iota , X +h_n \iota)               \big]}{ \E[ F_X(X -h_n \iota , X + h_n \iota) ]}     \lessapprox  \frac{h_n^{}}{h_n^{s}}  = o(1).
\end{align*}
In particular, condition (\ref{limsup}) holds and there is no power for alternatives that approach at rate $\gamma_n \asymp n^{-1/2} \{ \E[ F_{X}(X - h_n \iota , X + h_n \iota) ]   \}^{-1/4}  $  in the $\delta(.)$ direction. 
\end{example}
\begin{remark}[On the power with singular components] \label{singcom}
If $F_{X}$ is absolutely continuous, it is well known (see e.g.  \cite[Theorem 3]{Zhang}) that the fastest possible rate of approach to the null model is given by  $ n^{-1/2} h_n^{-q/4} $.   If $F_{X}$ contains singular components, we have $h_n^{-q}  \E\big[  F_X(X-h_n \iota, X+ h_n \iota)   \big] \uparrow \infty $. By Theorem \ref{pow1}, it then follows that for measures with singular components, the alternative sequence can approach the null at a rate faster than in the fully absolutely continuous case (provided that the support of $\delta(.)$  sufficiently ``touches'' the singular measures support for (\ref{liminf}) to hold).
\end{remark} 
We now examine the situation where  (\ref{limsup}) holds. In this case, there is no power for alternatives that approach the null model in the $\delta(.)$ direction with rate \begin{equation} \label{ratefast} \gamma_n^* = n^{-1/2} \{ \E[ F_{X}(X - h_n \iota , X + h_n \iota) ]   \}^{-1/4} .  \end{equation} A natural question then is what is the minimum rate $\gamma_n > \gamma_n^*$ for which power exists? In general $\gamma_n$ depends non-trivially on the  interaction between the support of $\delta(.)$ and the local singularity of the measure. To illustrate the essential idea, we consider the case where $F_{X} $ can be represented as  a finite mixture of distributions from $\mathcal{D}$. That is, there exists a finite index set $T$ such that \begin{align}
\label{mud} F_{X} = \sum_{t \in T} \alpha_{t} F_{t} \; \; , \;  \alpha_{t} >0 \; \; \text{and} \; \; \sum_{t \in T} \alpha_{t} = 1 \; , \; \; F_{t} \in \mathcal{D}(s_t) \; \; \text{for some} \; \;s_t \in (0,1] .
\end{align}
In this case, $\E\big[  F_X(X-h_n \iota, X+ h_n \iota)   \big]  \asymp h_n^{ sq} $ where $s = \min_{t \in T} s_t$. We consider the case where  $\delta(.)$ concentrates away from the support of the components $\{ F_t : s_t =s \}$ or equivalently, the support of $\delta(.)$ does  not interact with areas where the local singularity of the measure is at its maximum.

 Given a  distribution $F_{t}$, its support can always be expressed as \begin{align} \label{support} \mathcal{S}_{t} =  \{x \in \R^q :  F_t(x-r \iota , x +r \iota) > 0 \; \; \text{for every} \; r > 0     \} .  \end{align}

Let $\mathcal{S}_{\delta} = \{ x \in \R^q : \delta(x) \neq  0 \}$ and denote the common support of $(F_t,\delta)$  by $ \mathcal{S}_{t,\delta} = \mathcal{S}_{t} \cap \mathcal{S}_{\delta} $. For the mixture distribution in (\ref{mud}), denote by $R = R(\delta)$  a subset of the index set $T$ for which

 \begin{align}  & \underset{X  \stackrel{}{\sim}  F_{t} }{\E} \big[  \delta^2(X)  \big] = 0 \; \; \; \; \forall \; t \in T \setminus R  \; , \label{zero}     \\ & z \in T \setminus R \; , s_z < \min_{v \in R} s_v \implies  F_{z}(x-h \iota , x +h \iota) = 0  \; \; \; \; \; \; \text{for every} \; \; x \in \bigcup_{t \in R} \mathcal{S}_{t,\delta} \label{interzero} \end{align}  
holds for all sufficiently small $h > 0 $. Condition (\ref{zero}) restricts the support of $\delta(.)$ to lie outside of $ \bigcup_{t \in T \setminus R} \mathcal{S}_t $. Condition (\ref{interzero}) essentially states that $ \bigcup_{t \in R} \mathcal{S}_{t, \delta}    $ is well-separated from $\mathcal{S}_z$ for every $z \in T \setminus R$ that has a sharper singularity than that of $R$. In particular, if there exists $\epsilon > 0 $ such that $  \inf_{x \in \mathcal{S}_{t , \delta} \,  ,  \, y \in \mathcal{S}_z } \|  x-y \|_{\infty} \geq \epsilon  $ for every $t \in R$ and $z \in T \setminus R$, then (\ref{interzero}) holds for every $h <  \epsilon$ (for a one dimensional illustration, see Example \ref{supportr}).

Denote the restricted $\delta$-singularity coefficient of $F_{X}$ by $s_{\delta} = \min_{t \in R} s_{t}$   and the  probability measure induced from the restriction by $F_{R} = (\sum_{t \in R  } \alpha_t)^{-1} \sum_{t \in  R } \alpha_t F_t    $.   The following theorem illustrates that for directions $\delta(.)$ that satisfy (\ref{zero}-\ref{interzero}) with no power at rate $\gamma_n^*$, there may be power when the alternatives approach at a slower rate  $\gamma_n > \gamma_n^*$, provided that the support of $\delta(.)$ intersects nontrivially with some subset of the mixture components support.

\begin{theorem} \label{pow2} Suppose the alternative hypothesis $H_1$ and Assumptions (\ref{data}, \ref{error}, \ref{kernel}, \ref{greg})  hold.  Let  $F_{X} $ be as in (\ref{mud}) and $R(\delta) \subset T$ be such that (\ref{zero}-\ref{interzero}) holds.   Suppose $\delta(.)$ is uniformly continuous on the support of $F_{R}$ and either $(i)$ $\delta \in L^{\infty}(X)$ or $(ii)$ $\delta \in L^2(X)$ and Assumption \ref{greg}(iii) holds with $M(X) \in L^{4 + \epsilon}$ for some $\epsilon > 0 $.

\begin{enumerate}

\item[(i)]  If  $\gamma_n = o(   n^{-1/2} \{ \E[ F_X(X - h_n \iota , X + h_n \iota) ]   \}^{1/4} h_n^{-s_{\delta}q/2})$, then $$ \frac{n \hat{I}_{n}}{\sqrt{\hat{\sigma}_{n}^{2}}} \overset{d}{=} N(0,1) + o_{\mathbb{P}}(1).  $$

\item[(ii)] If $\gamma_n \asymp  n^{-1/2} \{ \E[ F_X(X - h_n \iota , X + h_n \iota) ]   \}^{1/4} h_n^{-s_{\delta}q/2} $  and \begin{align}
  \liminf_{h \downarrow 0}  \frac{\underset{X  \stackrel{}{\sim} F_{R} }{\E} \big[  \delta^2(X) F_R(X-h \iota , X+ h \iota)      \big]   }{\underset{X  \stackrel{}{\sim} F_{R} }{\E}  \big[  F_R(X - h \iota , X + h \iota)      \big]}                > 0 \; ,  \label{liminf2}
\end{align}

then there exists a deterministic sequence of constants $(L_n)_{n=1}^{\infty}$ such that
 \begin{align*}
 & \frac{n \hat{I}_{n}}{\sqrt{\hat{\sigma}_{n}^{2}}} \overset{d}{=} N(0,1) + L_n + o_{\mathbb{P}}(1)  \; \;, \; \; \liminf_{n \rightarrow \infty} L_n > 0.
\end{align*}
\end{enumerate}
\end{theorem}
In the special case where $  R = T  $, we have $s_{\delta} 	  = \min_{t \in T} s_t = s $. By Lemma \ref{d-lemma}, $F_{X} \in \mathcal{D}(s)$ and $   \E[ F_X(X - h_n \iota , X + h_n \iota) ]   \asymp h_n^{sq}$. Therefore, the hypothesis and conclusion of Theorem \ref{pow2} is identical to  Theorem  \ref{pow1} in this setting. Moreover,  when  $F_{X} \in \mathcal{D}(s)$, boundedness of $\delta(.)$ can be replaced with $L^2(X)$ integrability and the existence of sufficient moments for the envelope function $M(.)$ appearing in Assumption \ref{greg}. Other variations on these assumptions are possible as well.

When $ R \subsetneqq T $ and $s_{\delta} > s $, Condition (\ref{liminf}) of Theorem \ref{pow1} will fail as the ratio has rate $  h^{q(s_{\delta}-s)} $. Condition (\ref{liminf2}) corrects for this by using  only the mixture components where $\delta(.)$ has support and the appropriate rate $ \E_{X  \sim F_R}  \big[  F_R(X - h \iota , X + h \iota)      \big] \asymp h_n^{s_{\delta}q}$ for  the denominator.  If (\ref{liminf2}) holds, there will be power for alternatives that approach in the $\delta(.)$ direction at   rate  $$ \gamma_n = n^{-1/2} \{ \E[ F_X(X - h_n \iota , X + h_n \iota) ]   \}^{1/4} h_n^{-s_{\delta}q/2} \asymp  \gamma_n^* h_n^{-(s_{\delta}-s)q/2} > \gamma_n^* .$$
As an application, let $\delta(.)$ be as in Example \ref{supportr}. We take $R$  to  be the subcomponent that contains only the $N(0,1)$ distribution. Then (\ref{liminf2}) holds whenever $  \E _{X \sim N(0,1)}[   \delta^2(X)      ]   > 0 $. In this case, Theorem \ref{pow2} implies that power exists  at the rate $$ \gamma_n = n^{-1/2} h_n^{- \log_{3}(2) /4}   h_n^{-( 1  - \log_{3}(2) )/2}  \approx n^{-1/2} h_n^{-0.3423} .$$
Thus, under the null, the  goodness-of-fit statistic $\hat{\tau}_n$ in (\ref{selfn}) converges weakly to a standard Gaussian for a large class of continuous measures. However, the usual local power analysis is complicated by possible singularities and their interplay with the direction of approach to the null model.

\section{Simulations}
\label{sec4}
In this section, we provide simulation evidence on the asymptotic normality of the goodness-of-fit test statistic  under the null hypothesis. Additionally, we illustrate the sensitivity of the test to the distribution of the conditioning variables. We use the  Epanechnikov kernel and the number of replications  in all cases is $5000$ (population quantities defined through $\mathbb{P}$ use the empirical analog from the simulation draws).

In the interest of working with a distribution that admits a non-trivial singular component, we make use of the fact that there exists a well developed theory (see e.g. \cite{ergodic})  for approximating a random sample from the normalized Hausdorff measure of any self-similar fractal. We focus on a two dimensional regressor  with mixture distribution \begin{equation} \label{mixing}
F_{X} = \alpha_{1} F_{U} + \alpha_{2} F_{S} \; \; \; \; \;, \; \; \; \alpha_1+\alpha_2 = 1 \; \; \; \; , \; \; \; \alpha_1,\alpha_2 \geq 0.
\end{equation}
where $F_{U}$ is the uniform distribution on $[-2,2]^2$ and $F_{S}$ is the $\log_{2}(3) \approx 1.585 $ normalized Hausdorff measure on $T=$ Sierpi\'{n}ski's Triangle (with vertices at $\{  (-2,0),(2,0),(0,2)    \}$).

A random sample from $F_{S}$ can be approximated using  the Markov chain generated by the chaos game (see e.g. \cite[Chapter 2.4]{ergodic})  on the iterated function system (the family of contraction maps in Example \ref{ex2}) associated with $T$.  In all cases, the initial $100$ draws are discarded as a Markov chain  burn-in  and the remaining draws are used as the observed sample of size $n$.

\begin{figure}[H]\centering
\subfloat{\label{2a}\includegraphics[height=3cm,width=.45\linewidth]{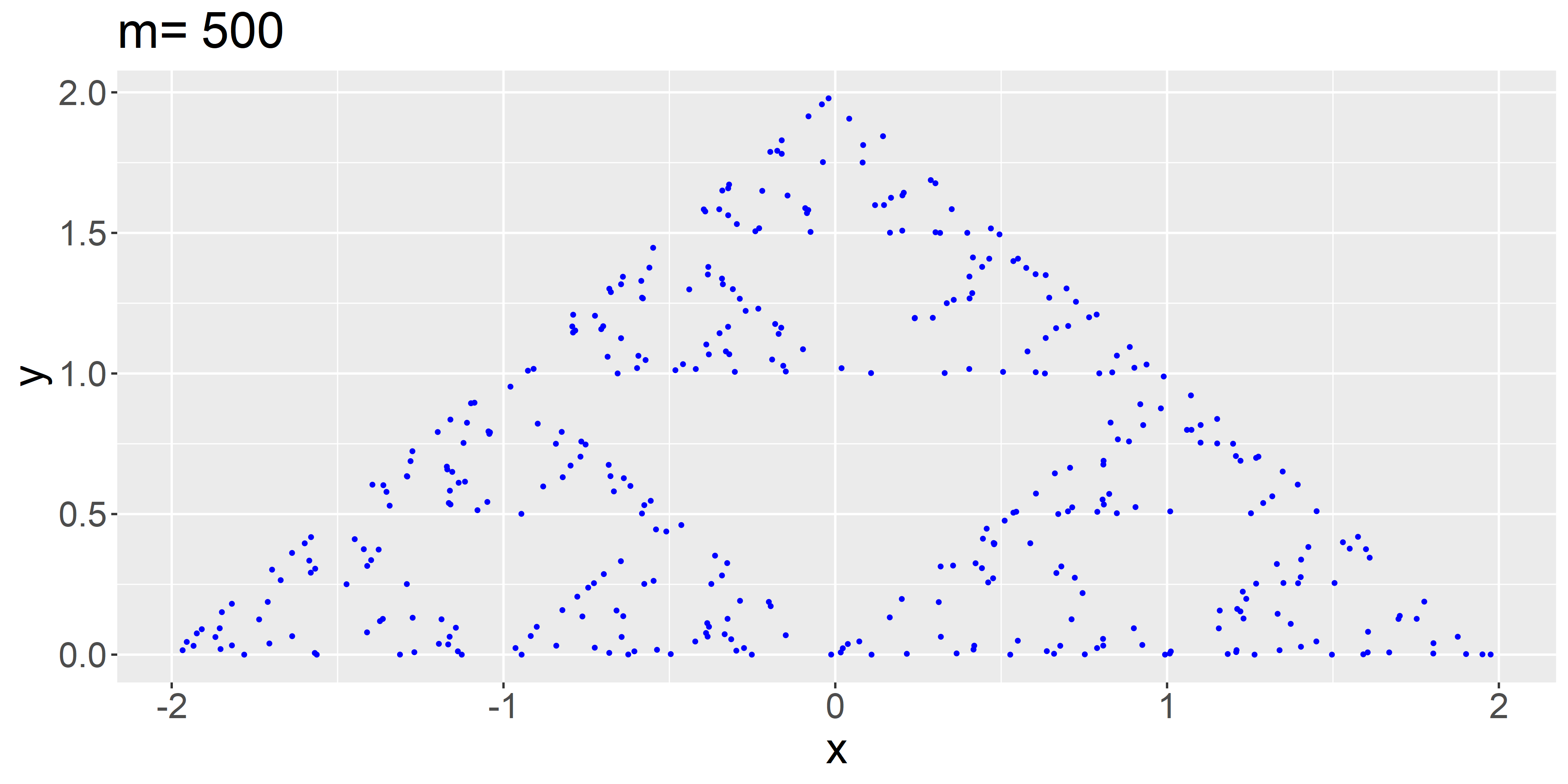}}\hfill
\subfloat{\label{2b}\includegraphics[height=3cm,width=.45\linewidth]{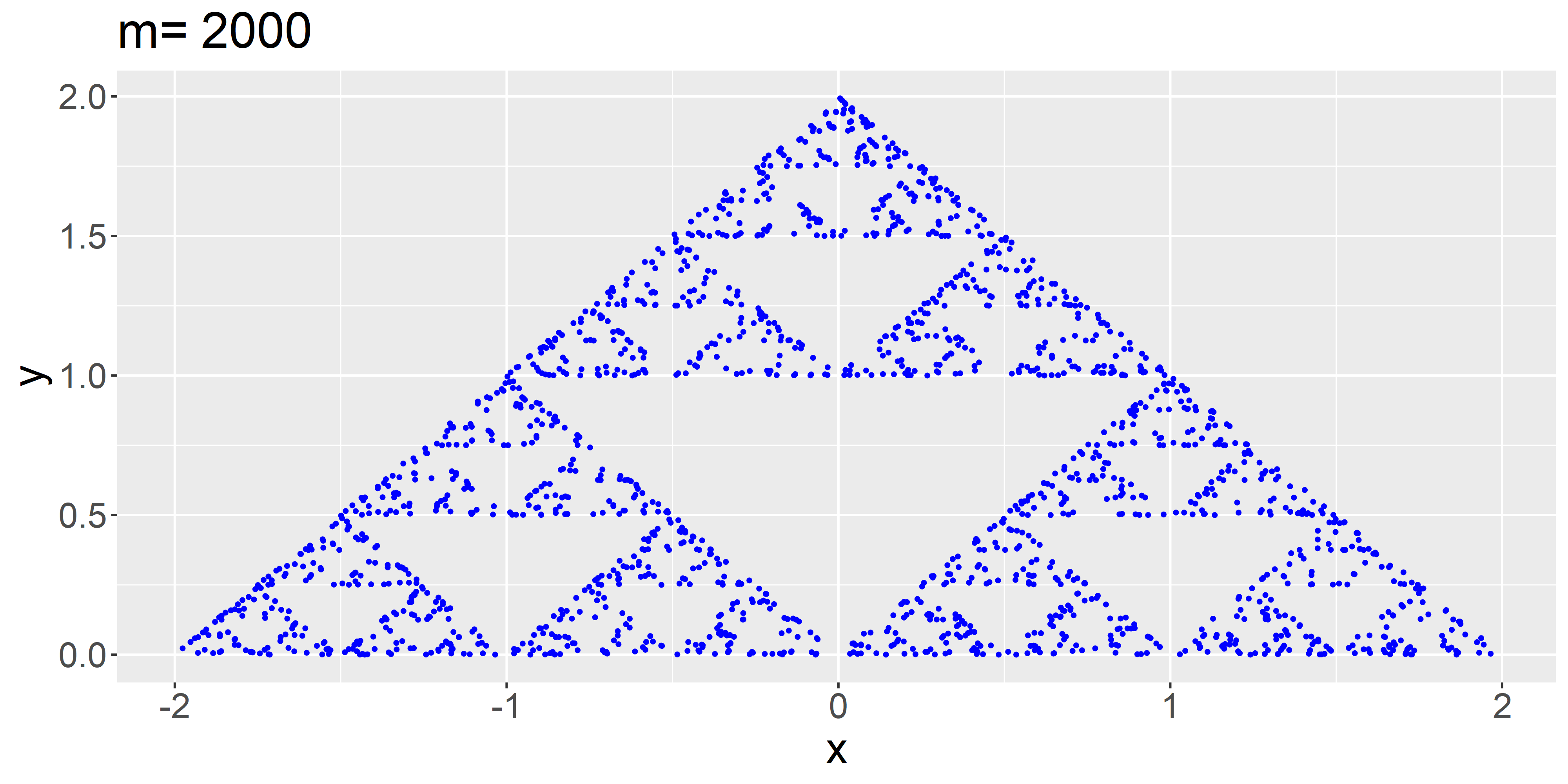}}\par 
\subfloat{\label{2c}\includegraphics[height=3cm,width=.45\linewidth]{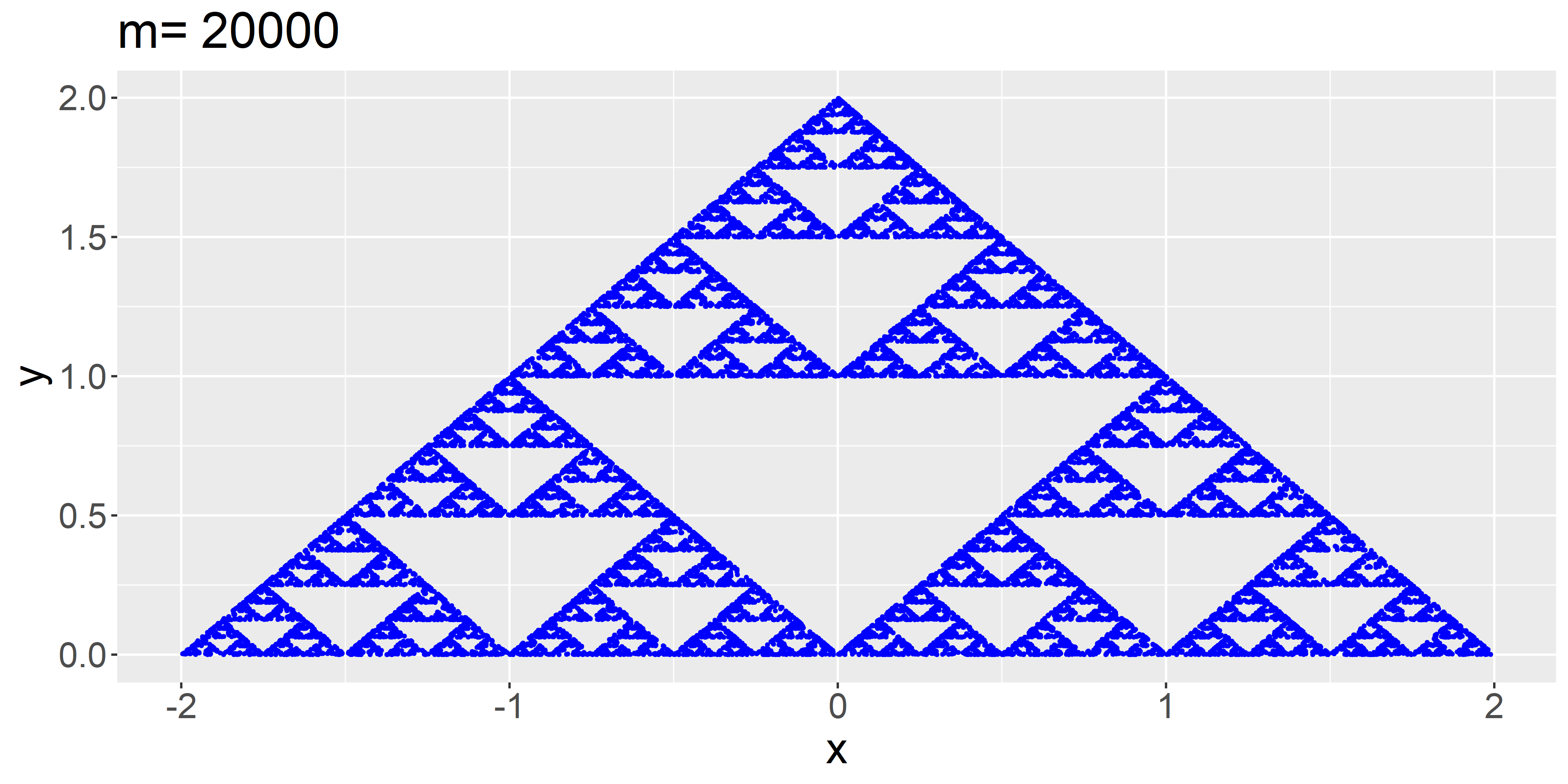}}
\caption{Stages of the chaos game for Sierpi\'{n}ski's Triangle. $m$ denotes the Markov chain length.  }
\label{fig2}
\end{figure}

Denote the coordinates of $X$ by $X=(V_1,V_2)$. We consider the null model \begin{equation}
H_0 \; \; : \; \; Y = \beta_0 +  \beta_1 V_1 +  \beta_2 V_2 + u \; \; \; \; \; \; , \; \; \; u \sim N(0,1).   \label{nullmod}
\end{equation}

The data is generated with $(\beta_0,\beta_1,\beta_2) = (1,1,1)$. We use the feasible goodness-of-fit statistic $ \hat{\tau}_n = n \hat{I}_n / \sqrt{ \hat{\sigma}_n^2}  $ where $  \hat{I}_n  $ and $ \hat{\sigma}_n^2$ are as in (\ref{I(n)}) and (\ref{sigma-hat}), respectively. The estimated residuals $\hat{u}_i$ are computed using ordinary least squares on the regression in (\ref{nullmod}).

\begin{figure}[H]
  \begin{subfigure}{0.32\textwidth}
    \centering
    \includegraphics[height=3cm,width=\linewidth]{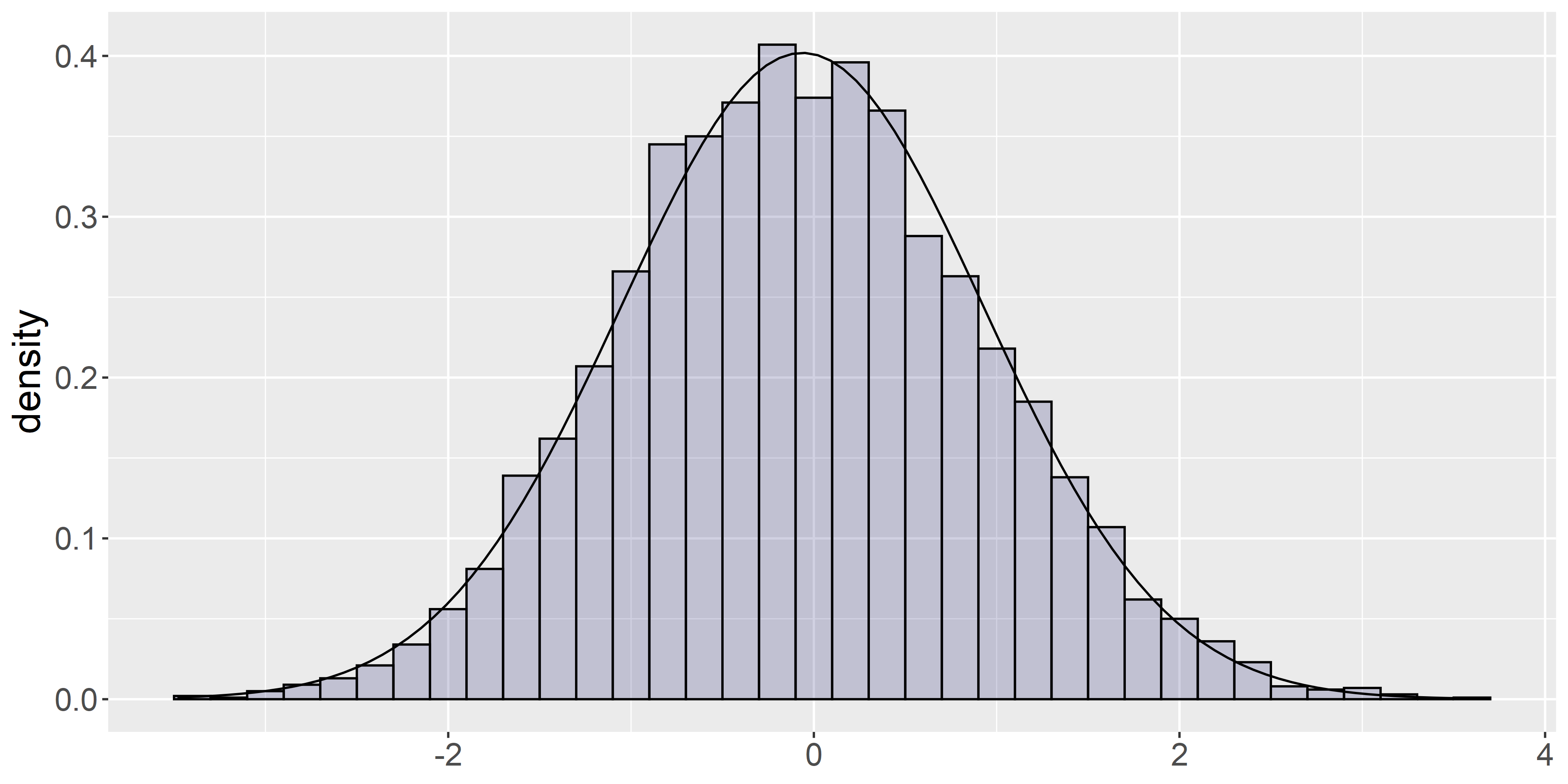}
  \end{subfigure}%
  \begin{subfigure}{0.32\textwidth}
    \centering
    \includegraphics[height=3cm,width=\linewidth]{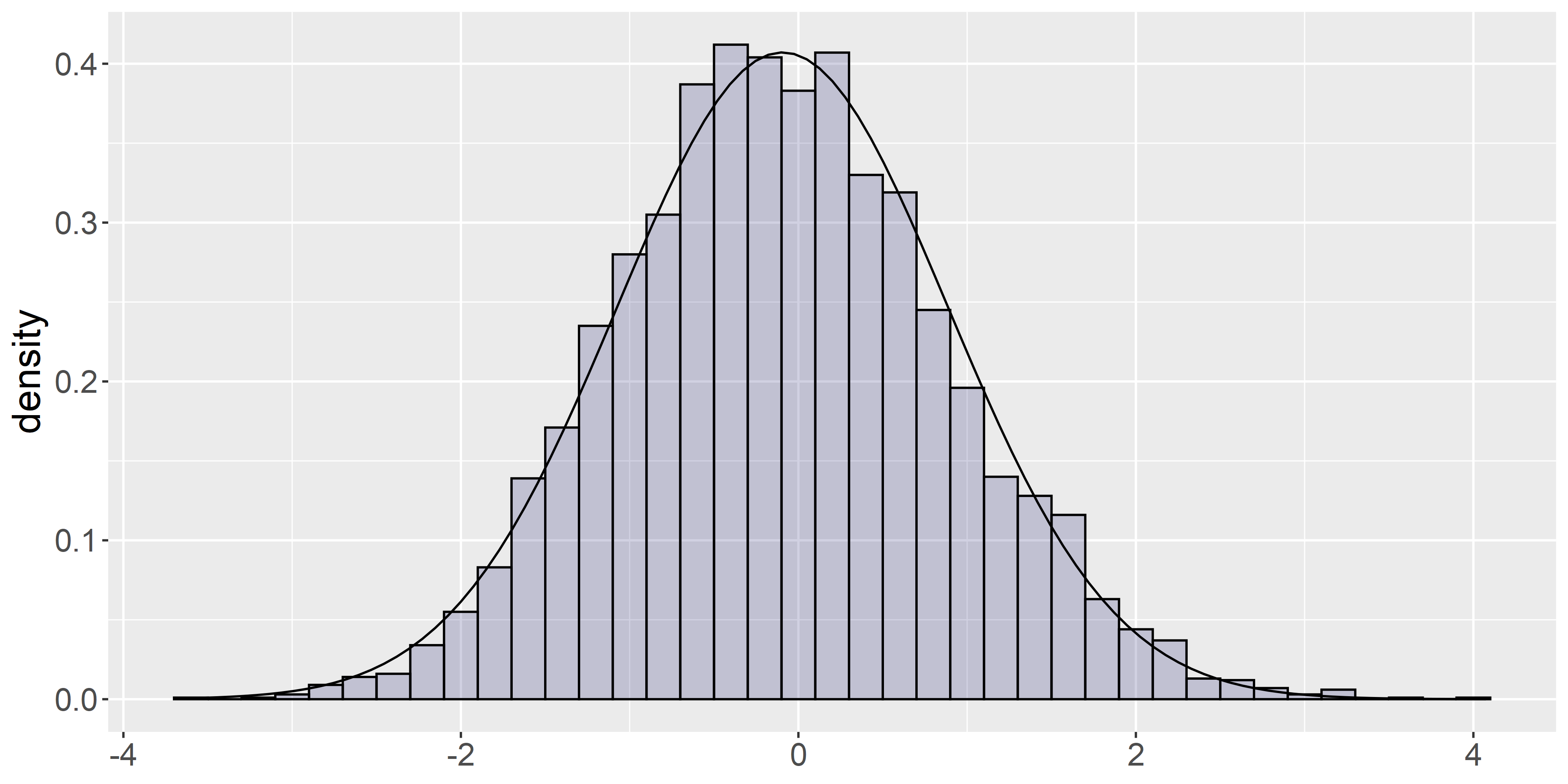}
  \end{subfigure}
  \begin{subfigure}{0.32\textwidth}\quad
    \centering
    \includegraphics[height=3cm,width=\linewidth]{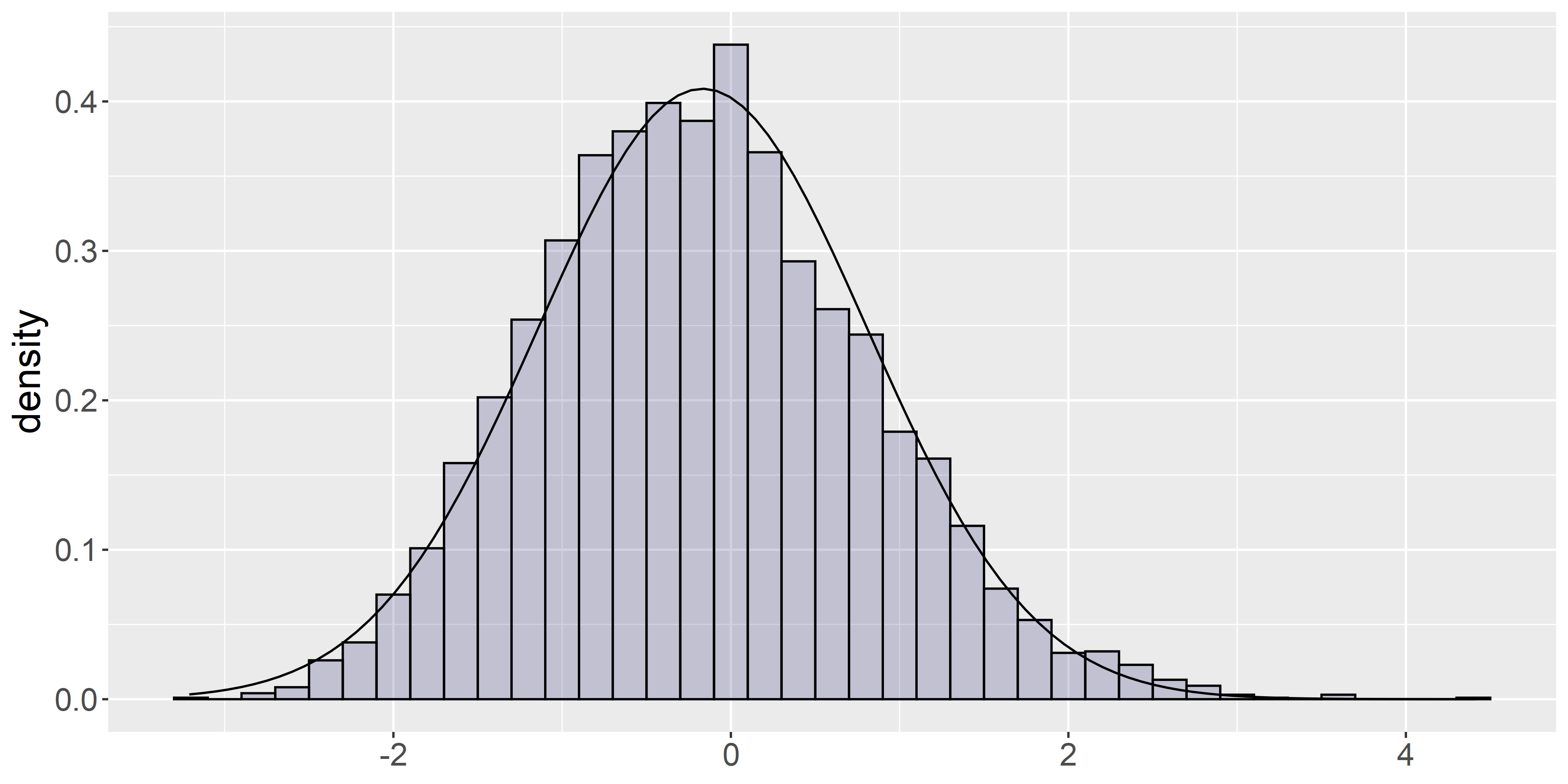}
  \end{subfigure}
  \medskip

  \begin{subfigure}{0.32\textwidth}
    \centering
    \includegraphics[height=3cm,width=\linewidth]{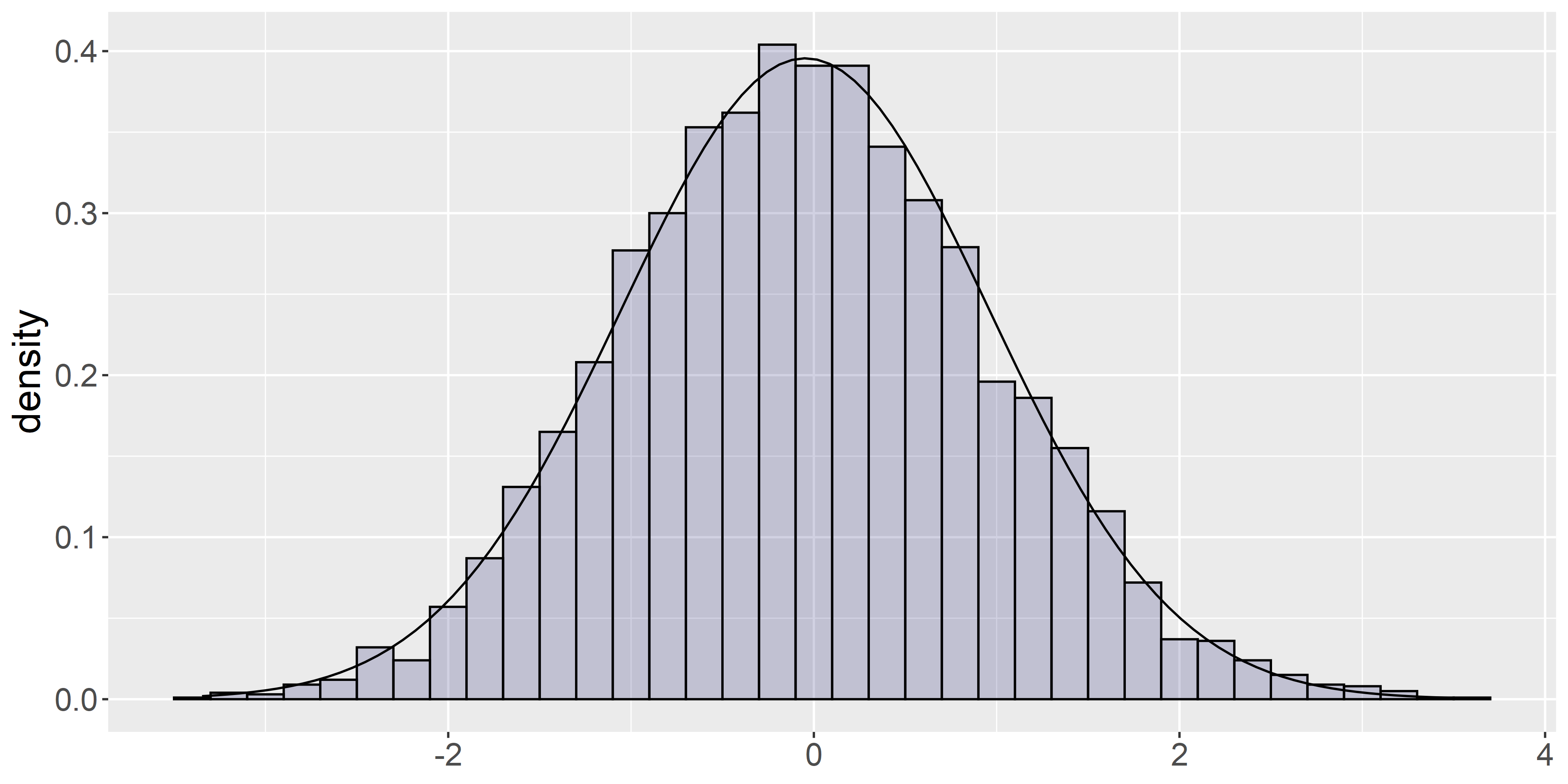}
    \caption{$h_n= n^{-1/2.5}$}
    \label{3a}
  \end{subfigure}
  \begin{subfigure}{0.32\textwidth}
    \centering
    \includegraphics[height=3cm,width=\linewidth]{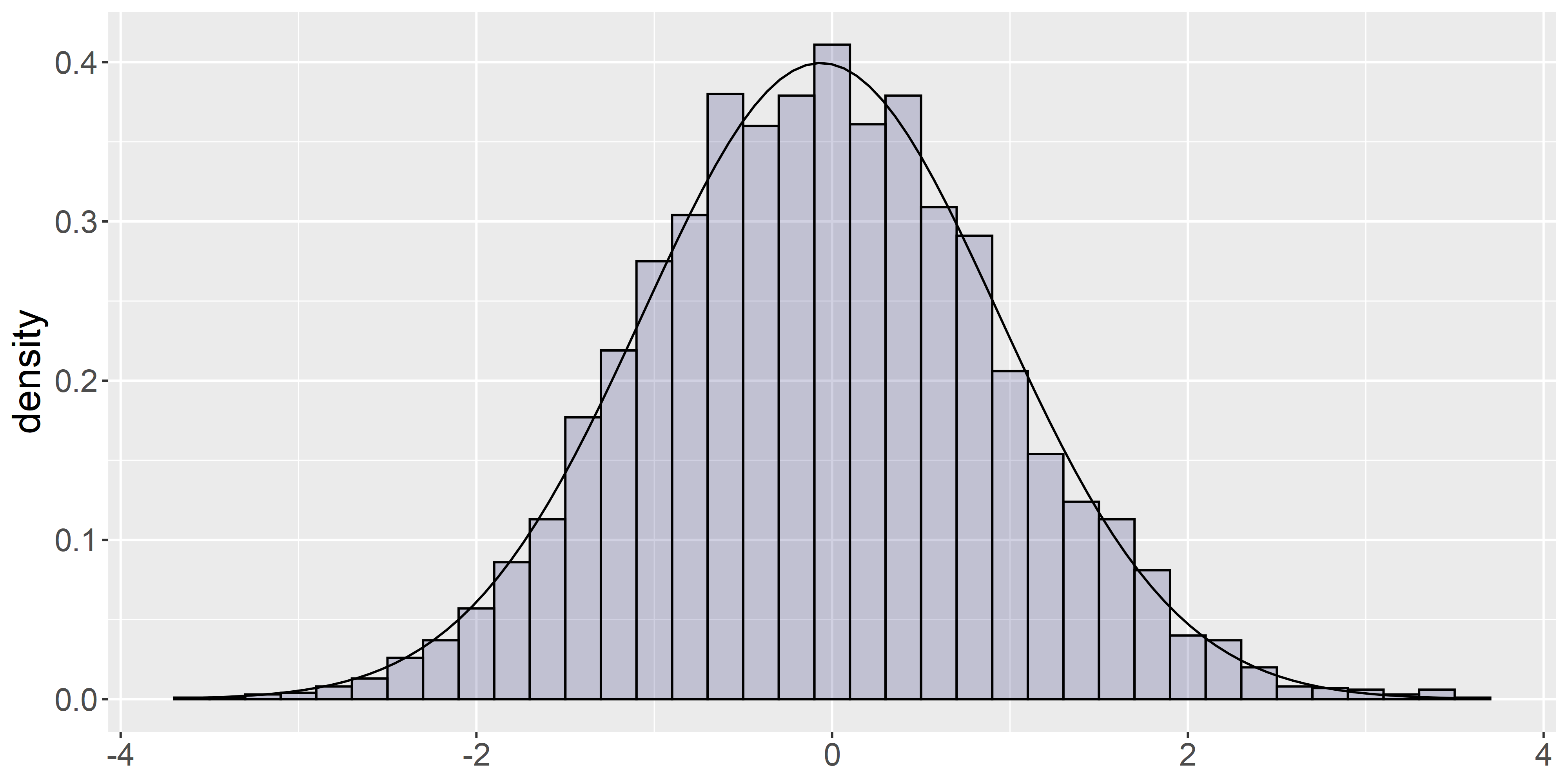}
    \caption{$h_n= n^{-1/3}$}
    \label{3b}
  \end{subfigure}
  \begin{subfigure}{0.32\textwidth}
    \centering
    \includegraphics[height=3cm,width=\linewidth]{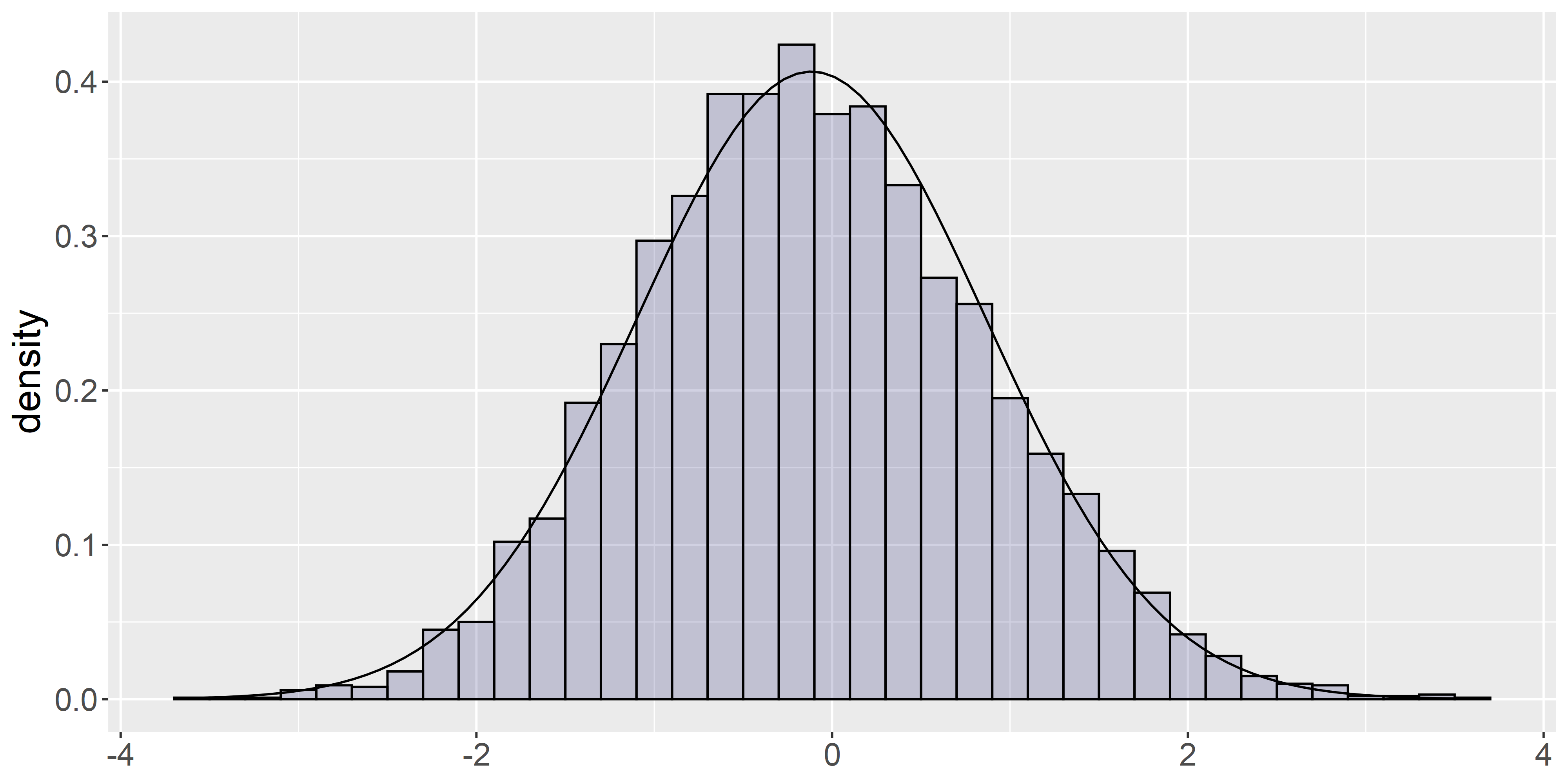}
   \caption{$h_n= n^{-1/4}$}
    \label{3c}
  \end{subfigure}
  \caption{Histograms of $\hat{\tau}_n$  when $F_X = 0.5 F_U + 0.5 F_S$ and $N( \E(\hat{\tau}_n), \text{Var}(\hat{\tau}_n)  )$ density  superimposed. Top to bottom: $n=1500, n=10000$.}
  \label{fig3}
\end{figure}
Figure \ref{fig3} illustrates the approximate normality of the test statistic under the null. As noted in the literature (e.g. \cite{mammen1}), even with absolutely continuous regressors, the finite sample distribution of the test statistic typically places more mass on the negative axis. In particular, a test computed with the usual one-tailed Gaussian critical value may be slightly undersized in finite samples. 
\begin{figure}[H]\centering
\subfloat{\label{4a}\includegraphics[width=0.9\linewidth,height=5cm]{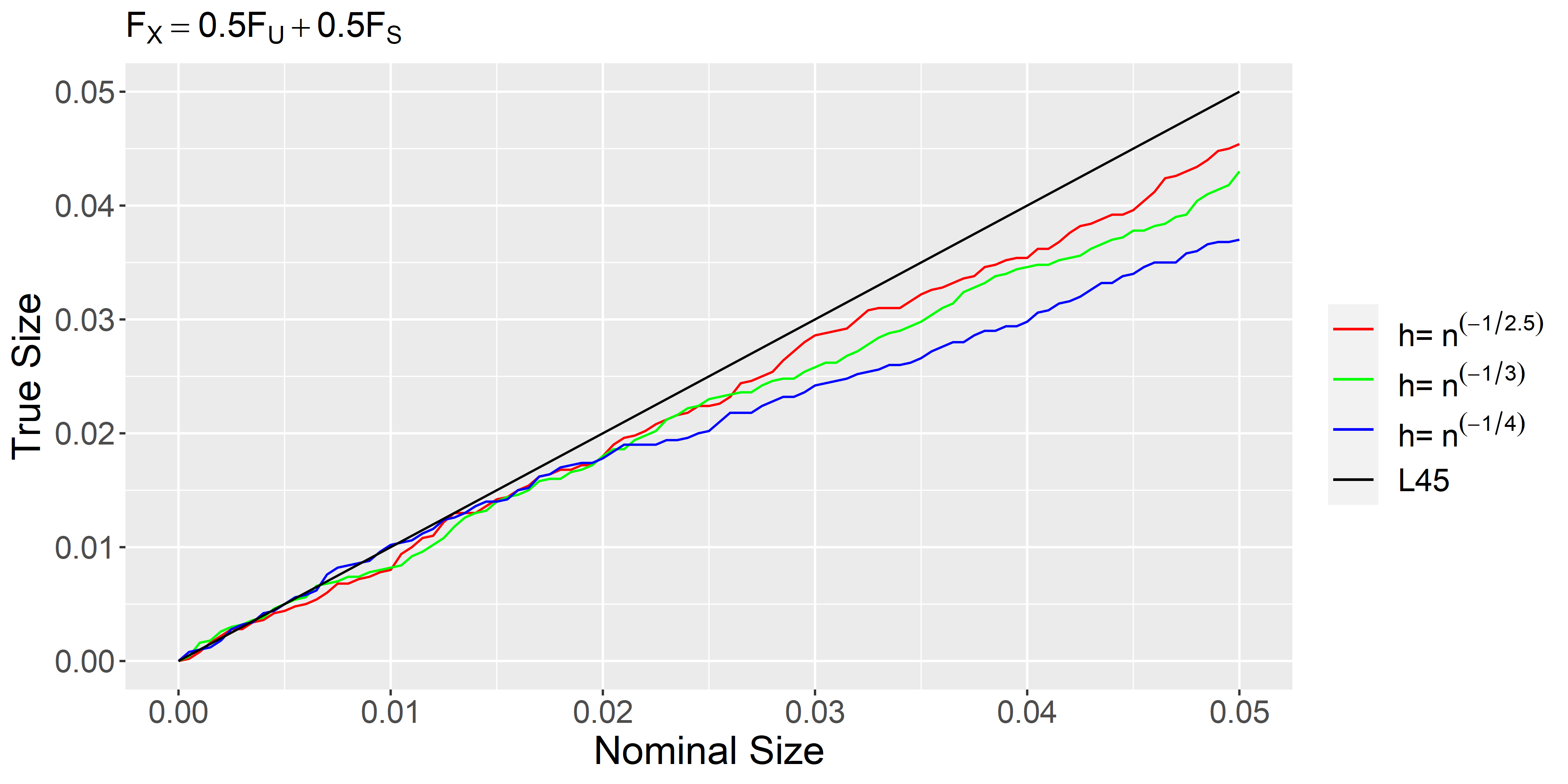}}
\caption{True/Nominal size at $n=1500$. The L45 line represents the points where equality occurs.}
\label{fig4}
\end{figure}

Next, we examine the sensitivity of the test statistic’s power as the level of singular contamination
increases. This is incorporated into the data generating process through the mixture coefficient $\alpha_2$ that
appears in (\ref{mixing}). We consider alternative models \begin{align*} 
& (i)  \; \; \; \; H_1 \; \; : \; \;   Y = 1 +  V_1 + V_2 + 0.4 \sin(2 \pi V_1) \sin(2 \pi V_2) + u \; \; ,   \\ & (ii)   \; \; \; H_1 \; \; : \; \; Y = 1 +   V_1 +  V_2 + 0.4 \cos(2 \pi V_1) \cos(2 \pi V_2) + u  \; \; .  
\end{align*}

\begin{table}[htp]

\centering
\caption{Power of the kernel test under alternative DGP $(i)$, $n=1500$}\label{tab1}

\begin{tabular}{@{}lcccccc@{}}
\hline 
& \multicolumn{3}{c}{\textbf{Nominal Size 1\%}} &\multicolumn{3}{c}{\textbf{Nominal Size 5\%}} \\
    \cmidrule(lr){2-4}\cmidrule(lr){5-7}
$\alpha_2$ & $  h_n = n^{-1/4}  $ & $h_n = n^{-1/3} $  &
$h_n = n^{-1/2.5} $ & $  h_n = n^{-1/4}  $ & $h_n = n^{-1/3} $  &
$h_n = n^{-1/2.5} $ \\
\hline
{$0.2$} & 0.504 & 0.230  & 0.103      & 0.714  & 0.448 & 0.259 \\
   {$0.5$}       & 0.728  & 0.431   & 0.208     & 0.876    & 0.657 & 0.422  \\
  {$0.8$}        & 0.905  & 0.673    & 0.388    & 0.963    & 0.847 & 0.619   \\[1pt]  
\hline
\end{tabular}

\bigskip

\caption{Power of the kernel test under alternative DGP $(ii)$, $n=1500$}\label{tab2}

\begin{tabular}{@{}lcccccc@{}}
\hline 
& \multicolumn{3}{c}{\textbf{Nominal Size 1\%}} &\multicolumn{3}{c}{\textbf{Nominal Size 5\%}} \\
    \cmidrule(lr){2-4}\cmidrule(lr){5-7}
$\alpha_2$ & $  h_n = n^{-1/4}  $ & $h_n = n^{-1/3} $  &
$h_n = n^{-1/2.5} $ & $  h_n = n^{-1/4}  $ & $h_n = n^{-1/3} $  &
$h_n = n^{-1/2.5} $ \\
\hline
{$0.2$} & 0.474& 0.222  & 0.104      & 0.687  & 0.441 & 0.265 \\
   {$0.5$}       & 0.678  & 0.411   & 0.204   & 0.839    & 0.643 & 0.423 \\
  {$0.8$}        & 0.859  & 0.658    & 0.390	    & 0.945    & 0.841 & 0.623   \\[1pt]  
\hline
\end{tabular}

\end{table}

\newpage As Tables \ref{tab1} and \ref{tab2} illustrate, the power exhibits a relevant dependence on the bandwidth. This is true even in the fully absolutely continuous case (see e.g. Table 2-4 in \cite{Zhang}).  Denote by $\delta(.)$, the trigonometric functions appearing in the alternative models $(i)$ and $(ii)$ above. From expanding $\hat{I}_n$ (as in the proof of Theorem \ref{pow1}), the term that provides the positive bias under the alternative hypothesis is given by \begin{align*}
    B_n =  \frac{1}{n(n-1)} \sum_{i=1}^n \sum_{j \neq i} K \bigg(  \frac{X_i - X_j}{h_n}    \bigg) \delta(X_i) \delta(X_j).
\end{align*}
Intuitively, as $h_n \rightarrow 0$, the terms in this sum are nonzero only when $\| X_i - X_j  \|_{\infty} $ is small and by uniform continuity this implies $\delta(X_i) \delta(X_j) \approx \delta^2(X_i)$.

At all bandwidth levels, the test exhibits higher power as the weight on the singular component increases. The expected small ball probability increases as the weight on the singular component increases. The interpretation of this in finite samples is that, given any observation $X_i$, there is a higher frequency of $X_j$ for which $\|  X_i - X_j  \|_{\infty} < h_n $ (and hence $K(h_n^{-1} [ X_i - X_j ]   ) \neq  0 $) as the weight on the singular component increases. A larger bandwidth has a similar effect on the small ball probability, although it also influences the statistic through other factors, such as its interaction with the uniform continuity of $\delta(.)$.

Finally, we note that the situation in small to moderate samples also depends on the choice of $\delta(.)$ used to construct the alternative. The trigonometric functions have comparable magnitude everywhere, thereby allowing us to focus more closely on the effects of singular contamination. By contrast, using a drift function $\delta(.)$ that is large in magnitude away from the support of the singular component could result in a power loss as contamination increases (although, by Remark \ref{singcom}, this effect is expected to vanish in large samples, provided that the support of $\delta(.)$ intersects the support of the singular component).
\section{Discussion}
\label{sec5}
This paper develops the limit theory for a class of kernel-weighted statistics when the underlying distribution of the conditioning variables admits a non-trivial Lebesgue decomposition.  The limit theory for these statistics centers around the behavior of the expected small ball probability. Under the null, the usual kernel smoothed goodness-of-fit statistic converges weakly to a standard Gaussian for a large class of continuous measures. However, in contrast to the absolutely continuous case, the usual local power analysis of these statistics depends non-trivially on the direction of approach to the null model. We expect that our analysis has similar implications for the more complicated setups that make use of kernel smoothed statistics.

The results could be extended in future work in a number
of directions. In this paper, we focus primarily on the goodness-of-fit statistic proposed in \cite{Zhang}. The results could be extended to other applications (e.g. \cite{phillips,gao,ShaikhVytlacil,lin,zheng1998consistent}) that make use of an identical form of the statistic. As discussed in \cite{gonz2013}, this statistic is motivated by  the moment condition $ \E \big( u \E(u|X) f_X(X)            \big) = 0 $. By contrast, the kernel-based goodness-of-fit statistics in \cite{mammen1} and \cite{dette1999consistent} are motivated by the moment conditions $\E \big(   \E^2[ u |X]     \big) = 0$ and $ \E \big[   u^2  - \big(  u - \E[u|X]  \big)^2          \big] = 0$, respectively. In all cases, the test statistic is asymptotically equivalent to a quadratic form, and so an appropriately debiased version of the statistic has similar asymptotics to that of a degenerate U-statistic. We expect that our analysis could be extended to these cases. A more ambitious avenue would be to expand the investigation to accommodate the more complicated setups beyond goodness-of-fit testing that make use of kernel smoothed statistics.

\section*{Acknowledgments}
The authors are grateful to Donald Andrews, Xiaohong Chen, Yuichi Kitamura, Renaud Raqu\'epas, Michael R Sullivan and Edward Vytlacil for their suggestions and constructive comments.

\begin{supplement}
\stitle{Supplement to ``Kernel-weighted specification testing under general distributions''.}
\sdescription{This supplemental
file contains  additional proofs and technical results  omitted in the main text.}
\end{supplement}



\nocite{fanLi}
\nocite{gao2008}
\nocite{meilan2020goodness}
\nocite{santana}
\nocite{sen2014}
\nocite{shah2018}
\nocite{verzelen2010}
\nocite{dette1999consistent}
\nocite{Bierens2}

\bibliographystyle{imsart_number} 
\bibliography{main.bib}       


\pagebreak

\begin{center}
\textbf{ \large Supplement to ``Kernel-weighted specification testing under general distributions''}
\end{center}

\setcounter{equation}{0}
\setcounter{figure}{0}
\setcounter{table}{0}
\setcounter{page}{1}
\makeatletter
\renewcommand{\theequation}{S\arabic{equation}}
\renewcommand{\thefigure}{S\arabic{figure}}
\renewcommand{\bibnumfmt}[1]{[S#1]}
\renewcommand{\citenumfont}[1]{S#1}

This supplemental file contains  additional proofs and technical details omitted in
the main text. For  convenience, we first list some of the notation that
was introduced in the main text and is frequently encountered in the proofs. Section \ref{sec1sup} details a number of technical auxiliary lemmas that are used
for the proofs of the main results. Section \ref{sec2sup} contains the proofs of the
statements appearing in the main text.

Given $x,s,t \in \R_{}^q$, we define the cube centered at $x$ with directions $(s,t)$ to be
\[
B\left(  x-s,x+t\right)  =\left\{  y \in \R^q  :x_{i}-s_{i}\leq y_{i}\leq x_{i}+t_{i} \; \: \forall  \; i=1,...,q\right\}  .
\]
Given any $f(X) \in L^1(X)$, define
\begin{align}
  \Omega_{f}(x-s,x+t)=\int\limits_{B\left(  x-s,x+t\right)  } f
(y)d F_{X}(y) . \label{omega} \end{align}
If $\mu_{2}(X) $ and $ \mu_{4}(X)$ are defined as in Assumption $2(ii)$, we denote the special case of     $f(X) = \mu_{2}(X), \mu_{4}(X)$   by $\Omega_{2}(.)$ and $\Omega_{4}(.)$, respectively. Additionally, we denote the case $f(X)=1$ by

\begin{align}
 F_{X}(x-s,x+t)=\int\limits_{B\left(  x-s,x+t\right)  } d F_{X}\;.     \label{Fx}
\end{align}

Given a sufficiently differentiable function $g$ (in the sense of Definition $3.1$), we define $$ \partial_{x} g(x) = \frac{\partial^q g(x)}{ \partial x_{1} \dots \partial x_{q}} .$$

The sequence of local alternatives used in the local power analysis is denoted by \begin{equation}
H_{1}:Y_{}=g_{}\left(  X_{},\beta_{0}\right)  +\gamma_{n}\delta\left(
X_{}\right)  +u_{}.  \label{reg}
\end{equation}
The support of $\delta(.)$ is denoted by $\mathcal{S}_{\delta} = \{ x \in \R^q : \delta(x) \neq 0  \}$. When $F_{X}$ can be expressed as a mixture that includes component $F_{t}$, the notation $\underset{X  \stackrel{}{\sim} F_{t} }{\mathbb{P}} $ and $\underset{X  \stackrel{}{\sim} F_{t} }{\mathbb{E}}$ is used to indicate that the operators are defined with respect to $X\sim F_{t}$.

\section{Auxiliary Lemmas}
\label{sec1sup}

\begin{lemma}
\label{app1}

Suppose $\phi \in C^1(\R)$ with $\phi(0) = 0$. Suppose $g  \in C(\R)$  has support contained in $ \Delta = [-1,1]$ and is twice continuously differentiable on $\Delta$.  Then the function \begin{align}
\label{appeq1} T(u) =  \int \limits_{[-1,1]}  \phi \big(  \partial_{v} [ g(v) g(u+v)         ]    \big)  \mathbbm{1} \big \{   -1 \leq u+v \leq 1     \big \}   dv 
\end{align}
is continuous on $\R$, has support contained in $ \Sigma = [-2,2]  $ and is continuously differentiable on $\Sigma$.

\end{lemma}

\begin{proof}[Proof of Lemma \ref{app1}]

Since $g$ has support contained in $\Delta$ and $\phi(0) = 0$, it is straightforward to deduce that $ T(u) =   0   $ for $ u \notin \Sigma$ and  that $\lim_{u \uparrow 2} T(u) = \lim _{u \downarrow -2} T(u) = 0 $.  Let $f(u,v) =  \phi \big(  \partial_{v} [ g(v) g(u+v)         ]  \big) $ and denote by $f'(u,v)$, the derivative of $f(u,v)$ with respect to $u$. From the hypothesis on $g$, the functions $f(u,v)$ and $f'(u,v)$ are uniformly continuous on the compact set $ \mathcal{E} = \{ (u,v) \in \Sigma \times \Delta : -1 \leq u+v \leq 1   \}  $. In the interior of $\Sigma$, a straightforward application of the Leibniz rule yields

 \begin{align}
\label{Tderiv} T'(u) =  \begin{cases}  \int \limits_{[-1,1-u]} f'(u,v) dv - f(u,1-u) & u \in (0,2)  \; , \\  \int_{[-1,1]} f'(u,v) dv & u=0 \; ,       \\  \int \limits_{[-1+|u|,1]} f'(u,v) dv + f(u,-1+|u|)  &   u \in (-2,0).     \end{cases}
\end{align}
Since $f(0,1) = f(0,-1) = \phi(0) = 0 $, it follows that $T'(u)$ is continuous on $(-2,2)$ and it admits a continuous extension to $\Sigma=[-2,2]$ with $ T'(2) = T'(-2) = \phi(0) = 0.   $

\end{proof}

\begin{lemma}
\label{aux-new}
Fix any $M(X),G(X) \in L^1 (X)$ and suppose Assumptions (1,3) hold. Then for every $l \in \mathbb{N}$, there exists a universal constant $C > 0 $ such that \begin{align}
& \E \bigg[ K^{l} \bigg(   \frac{X_1 - X_2}{h}     \bigg)  M(X_2) G(X_1)          \bigg] \nonumber \\ & \leq C \min  \big \{ \E \big[  \left| G(X)   \right| \Omega_{\left|M\right|}(X- h \iota , X + h \iota)\big] \, , \, \E \big[  \left|M(X) \right| \Omega_{|G|}(X- h \iota , X + h \iota)\big]                            \big \}   \label{kbound1}
\end{align}
holds for every $h > 0 $. In particular, with $G = 1,$ we obtain that \begin{align} \label{kbound2}
\E \bigg[ K^{l} \bigg(   \frac{X_1 - X_2}{h}     \bigg)  M(X_2)          \bigg] \leq C \E \big[  M(X) F_X(X - h \iota , X + h \iota)         \big].
\end{align}

\end{lemma}

\begin{proof}
[Proof of Lemma \ref{aux-new}] 
Without loss of generality, we take $M$ and $G$ to be non-negative. \begin{align*}
\E \bigg[ K^{l} \bigg(   \frac{X_1 - X_2}{h}     \bigg)  M(X_2) G(X_1)          \bigg]  =\mathbb{E}\bigg[   M(X_2) \int_{\mathbb{R}^{q}} G(x)  K^{l}%
\bigg(\frac{x-X_{2}}{h_{}}\bigg)d F_{X} (x)\bigg].
\end{align*}
Define
\[
f(x)= G(x),\;g(x)=K^{l}\bigg(\frac{x-X_{2}}{h_{}}\bigg).
\]
Let $X_2^i$ denote the $i^{th}$ coordinate of $X_2$. Conditional on $X_{2}$, $(f,g)$ satisfy the hypothesis of Lemma 3.2 with  $ \mathcal{O}  =  (X_{2}^{1}%
-h_{},X_{2}^{1}+h_{}) \times \dots \times (X_{2}^{q}%
-h_{},X_{2}^{q}+h_{}) $.  Applying Lemma 3.2 yields
\begin{align*}
\int_{\mathbb{R}^{q}} G(x)  K^{l}%
\bigg(\frac{x-X_{2}}{h_{}}\bigg)d F_{X} (x)
&  = (-1)^q \int_{\mathbb{R}^{q}} \mathbbm{1} \{ t \in \mathcal{O}  \}  \Omega_{G}(X_{2}-h_{}\iota,t)\partial_{t}%
K^{l}\bigg(\frac{t-X_{2}}{h_{}}\bigg)dt\\
&  = (-1)^q \int_{\left[  -1,1\right]  ^{q}}\Omega_{G}(X_{2}-h_{}\iota,X_{2}%
+h_{}v)\partial_{v}K^{l}(v)dv\;,
\end{align*}
where the last equality follows from the change of variables $t\rightarrow
X_{2}+h_{}v$. It follows that
\begin{align*}
& \E \bigg[ K^{l} \bigg(   \frac{X_1 - X_2}{h}     \bigg)  M(X_2) G(X_1)          \bigg] \\ & =\mathbb{E}\bigg( M(X)(-1)^{q}%
\int_{\left[  -1,1\right]  ^{q}}\Omega_{G}(X-h_{}\iota,X+h_{}v)\partial
_{v}K^{l}(v)dv\bigg) \\ & \leq  C \E \big[  M(X) \Omega_{G}(X- h \iota , X + h \iota)\big]                  
\end{align*}
where $C =  \int_{[-1,1]^q} \left| \partial_{v} K^l(v)      \right| dv    $. The claim follows from repeating the argument with the roles of $M$ and $G$ reversed.
\end{proof}

\begin{lemma}
\label{aux-estimates}
Fix any non-negative functions $M(X), G(X) \in L^2(X)$ and suppose Assumption 4(ii) holds. Then as $h_n \downarrow 0$, we have that
\begin{align*}
& (i) \; \; \;   \E[M(X) F_X(X- h_n \iota , X+ h_n \iota)] = o \big( \{ \E[F_X(X - h_n \iota , X + h_n \iota)]    \}^{3/4}  \big) \; ,  \\ & (ii) \; \; \;
 \E [        \Omega_{M}^2 (X - h_n \iota , X + h_n \iota )       ]  = o \big( \E[F_X(X - h_n \iota , X + h_n \iota)]       \big) \; ,  \\ & (iii) \; \; \; 
 \E \big[  G(X) \Omega_{M} (X - h_n \iota , X + h_n \iota )              \big]   = o \big( \{\E[F_X(X - h_n \iota , X + h_n \iota)] \}^{1/2}  \big)    \: .        
\end{align*}
\end{lemma}

\begin{proof}
[Proof of Lemma \ref{aux-estimates}] 

\begin{enumerate}

\item[(i)] Let $\alpha_n > 0 $ denote  a fixed sequence of constants.  \begin{align*}
& \E[M(X) F_X(X- h_n \iota , X+ h_n \iota)] \\ & = \E[M(X) \mathbbm{1} \{M \leq \alpha_n \} F_X(X- h_n \iota , X+ h_n \iota)] \\ & \; \; \; \;  +  \E[M(X) \mathbbm{1} \{M > \alpha_n \} F_X(X- h_n \iota , X+ h_n \iota)] \\ & = T_1 + T_2.
\end{align*}
Clearly $T_1 \leq \alpha_n \E[ F_X(X- h_n \iota , X+ h_n \iota)] $. For $T_2$, H\"{o}lder's inequality and Assumption 4(ii) yield \begin{align}
T_2 & \leq  \big \{ \E \big( M^{3/2} \mathbbm{1} \{M > \alpha_n \}           \big)            \big \}^{2/3} \big \{    \E[ \{ F_X(X- h_n \iota , X+ h_n \iota) \}^3     ] \big \}^{1/3}  \nonumber  \\ & \leq    \{ \E \big( M^{3/2} \mathbbm{1} \{M > \alpha_n \}           \big)            \big \}^{2/3}   \{ \E[F_X(X - h_n \iota , X + h_n \iota)]    \}^{2/3}  \zeta_{n}^{1/3}  \label{ass4bound} 
\end{align}
where  \begin{equation}   \zeta_{n} =  \frac{\mathbb{E} \left[ \big[ F_{X}(X-h_n\iota,X+h_n\iota)
 \big]^3  \right] }{\left(  \mathbb{E}\left[  F_{X}(X-h_n\iota,X+h_n
\iota)\right]  \right)  ^{2}}     \downarrow 0.    \label{zeta} \end{equation}

For the first term on the right side of (\ref{ass4bound}) we further obtain that $$ \E \big( M^{3/2} \mathbbm{1} \{M > \alpha_n \}     \big) \leq  \frac{\E[M^{2} \mathbbm{1} \{M > \alpha_n \} ]}{\alpha_n^{1/2}}  . $$
Combining the bounds yield \begin{align*} & \E[M(X) F_X(X- h_n \iota , X+ h_n \iota)] \\ & \leq  \alpha_n \E[ F_X(X- h_n \iota , X+ h_n \iota)] +  \zeta_{n}^{1/3}  \frac{\{ \E[F_X(X - h_n \iota , X + h_n \iota)]    \}^{2/3}}{\alpha_n^{ 1/3   }}  \{  \E (M^2)  \}^{2/3}    . 
\end{align*}
The result follows with  $ \alpha_n \asymp \sqrt{ \zeta_{n}}  \{\E[ F_X(X- h_n \iota , X+ h_n \iota)] \}^{-1/4} $.  \\

\item[(ii)] Let $\gamma_n > 0$ denote a fixed sequence of constants. Define \begin{align*}&  \Omega_{M > \gamma_n}(X - h_n \iota ,X + h_n  \iota )=\int \limits_{ t \in \R^q : \|t - X  \|_{\infty} \leq h_n  } M(t) \mathbbm{1}\{M(t) > \gamma_n \} dF_{X}(t) \; , \\ &   \Omega_{M \leq \gamma_n}(X - h_n \iota ,X + h_n  \iota )=\int \limits_{ t \in \R^q : \|t - X  \|_{\infty} \leq h_n  } M(t) \mathbbm{1} \{M(t) \leq \gamma_n \} dF_{X}(t) \:.
\end{align*}
It follows that
\begin{align*}
 &  \sqrt{ \E \big[  \Omega_{M}^2 (X - h_n \iota , X + h_n \iota )              \big] }          \\ & \leq \sqrt{ \E \big[  \Omega_{M \leq \gamma_n}^2 (X - h_n \iota , X + h_n \iota )              \big] } +  \sqrt{ \E \big[  \Omega_{M > \gamma_n}^2 (X - h_n \iota , X + h_n \iota )              \big] }  \\ & = T_1 + T_2.  
\end{align*}
For $T_1$, monotonicity of the $L^p(X)$ norm and Assumption 4(ii) yields \begin{align*} 
T_1^2  \leq  \gamma_n^2 \E \big[ \big( F_X(X- h_n \iota , X+ h_n \iota) \big)^2  \big] &  \leq  \gamma_n^2  \{ \E \big[ \{ F_X(X- h_n \iota , X+ h_n \iota) \}^3  \big] \}^{2/3} \\ &   \leq  \gamma_n^2  \{ \E[ F_X(X- h_n \iota , X+ h_n \iota)] \}^{4/3}  \zeta_n^{2/3}         
\end{align*}
 where  \begin{equation*}   \zeta_{n} =  \frac{\mathbb{E} \left[ \big[ F_{X}(X-h_n\iota,X+h_n\iota)
 \big]^3  \right] }{\left(  \mathbb{E}\left[  F_{X}(X-h_n\iota,X+h_n
\iota)\right]  \right)  ^{2}}     \downarrow 0.     \end{equation*}
For $T_2$, Cauchy-Schwarz yields \begin{align*}
    \Omega_{M > \gamma_n}^2(X - h_n \iota ,X + h_n  \iota ) & \leq   \E(M^2 \mathbbm{1} \{M > \gamma_n \}) F_X(X- h_n \iota , X + h_n \iota).
\end{align*} 
From combining the bounds, we obtain that \begin{align*}   \sqrt{ \E [        \Omega_{M}^2 (X - h_n \iota , X + h_n \iota )       ] }  & \leq  \zeta_{n}^{1/3}   \gamma_n \{ \E[ F_X(X- h_n \iota , X+ h_n \iota)]  \}^{2/3}  \\ &   + \big \{ \E(M^2 \mathbbm{1} \{M > \gamma_n\}) \}^{1/2} \big  \{ \E[ F_X(X- h_n \iota , X+ h_n \iota)] \big\}^{1/2}.
\end{align*}
The result follows by letting $ \gamma_n \uparrow \infty $ sufficiently slowly, for e.g.  $$ \gamma_n \asymp  \{ \E[ F_X(X- h_n \iota , X+ h_n \iota)] \}^{-1/8}.  $$ 

\item[(iii)] The result follows immediately from part (ii) and an application of Cauchy-Schwarz.
\end{enumerate}
\end{proof}

\begin{lemma}
\label{varestpre}
 Define \begin{equation} \label{sn}
S_{n} =\frac{2}{n(n-1)}\sum_{i=1}^{n}\sum_{j\neq i}u_{i}^{2}u_{j}^{2}%
K^{2}\left(  \frac{X_{i}-X_{j}}{h_{n}}\right)  .%
\end{equation}
Suppose Assumptions 1-4 hold and $h_n \downarrow 0 \,,\, n h_n^{q} \uparrow \infty$. Then $$ \left| S_n -2 \E(H_n^2)  \right| = o_{\mathbb{P}} (\E (H_n^2)) . $$
\end{lemma}

\begin{proof}
[Proof of Lemma \ref{varestpre}]
By Markov's inequality it suffices to verify that
\begin{equation*}
\frac{\E\big[ \big(   S_n   -2 \E(H_n^2)  \big)^2      \big]       }{\{ \E(H_n^2)  \}^2} =  \frac{\E\big(  S_n^2      \big) -4\{ \E(H_n^2)  \}^2        }{\{ \E(H_n^2)  \}^2} = o(1) \; ,
\end{equation*}
where  \[
\mathbb{E}\big(S_{n}^{2}\big)=\frac{4}{n^{2}(n-1)^{2}}\sum_{i=1}^{n}\sum
_{t=1}^{n}\sum_{j\neq i}\sum_{l\neq t}\mathbb{E}[H_{n}^{2}(Z_{i},Z_{j}%
)H_{n}^{2}(Z_{t},Z_{l})\big].
\]
The expression in the sum is equal to $ [\E(H_n^2)]^2 =  [\mathbb{E}(H_{n}^{2}(Z_{1}%
,Z_{2}))]^{2}$ when the indices $(i,t,j,l)$ are all distinct. When the indices
are not all distinct, Cauchy-Schwarz  yields
\[
\mathbb{E}[H_{n}^{2}(Z_{i},Z_{j})H_{n}^{2}(Z_{t},Z_{l})\big]\leq
\sqrt{\mathbb{E}\big[H_{n}^{4}(Z_{i},Z_{j})\big]}\sqrt{\mathbb{E}%
\big [H_{n}^{4}(Z_{t},Z_{l})\big]}=\mathbb{E}[H_{n}^{4}(Z_{1},Z_{2})].
\]
In the sum, there are $n(n-1)(n-2)(n-3)$ terms that correspond to distinct indices and $O(n^3)$ terms otherwise. It follows that
\begin{align*}
\frac{\mathbb{E}(S_{n}^{2})- 4 \{\E(H_n^2) \}^2}{\{  \E(H_n^2)  \}^2}   &  \leq 4  \left(  \frac{n(n-1)(n-2)(n-3)}%
{n^{2}(n-1)^{2}}-1\right)   
  +  O(n^{-1}) \frac{\E(H_n^4)}{\{  \E(H_n^2) \}^2} \\ & = o(1) + O\bigg(  \frac{\E(H_n^4)}{n \{  \E(H_n^2) \}^2}  \bigg).
\end{align*}
By Lemma 3.5 and Assumption 4(i) we have that $\E(H_n^4) \lessapprox \E(H_n^2)$.  This and Corollary 3.4(ii) imply  \begin{align*}
 \frac{\E(H_n^4)}{n \{  \E(H_n^2) \}^2}  \lessapprox \frac{1}{n \E[H_n^2]} \lessapprox \frac{1}{n h_n^q} 
\end{align*}
Since $n h_n^{q} \uparrow \infty$,  we obtain that $$ \frac{\mathbb{E}(S_{n}^{2})- 4 \{\E(H_n^2) \}^2}{\{  \E(H_n^2)  \}^2}   = o(1). $$
\end{proof}

\begin{lemma}
\label{varest}
Let Assumptions 1-5 hold and $h_n \downarrow 0 \,,\, n h_n^{q} \uparrow \infty$. Suppose further that the alternative hypothesis $H_1$ in (\ref{reg}) holds with $\delta \in L^{\infty}(X)$ and  $ \gamma_n \lessapprox n^{-1/2} \{ \E[ F_{X}(X - h_n \iota , X + h_n \iota) ]   \}^{-1/4} $. Then

\begin{equation*}
\label{varequiv}
\left|  \hat{\sigma}_{n}^2 - 2 \E(H_n^2)     \right|  =o_{\mathbb{P}} \big(  \E(H_n^2)   \big) .
\end{equation*} 
\end{lemma}

\begin{proof}
[Proof of Lemma \ref{varest}]
Observe that
$$  \left|  \hat{\sigma}_{n}^2 - 2 \E(H_n^2)     \right|  \leq \left|   \hat{\sigma}_n^2 - S_n    \right| + \left| S_n - 2 \E(H_n^2)   \right|  = T_1 + T_2 $$
where $S_n$ is as in (\ref{sn}). By Lemma \ref{varestpre}, $T_2 = o_{\mathbb{P}} (\E(H_n^2))  $. Therefore, it suffices to verify that $T_1 = o_{\mathbb{P}} (\E(H_n^2))$. The mean value theorem implies that there exists a $\beta_*$ on the line segment connecting $\hat{\beta}$ and $\beta_0$ such that \begin{align}
 \hat{u}_i - u_i =  g(X_{i},\beta_0)- g(X_{i},\hat{\beta}) + \gamma_n \delta(X_i) = [\nabla_{\beta} g_{}(X_{i},\beta_*)]'( \beta_0 - \hat{\beta} ) + \gamma_n  \delta(X_i).  \label{mv2}
\end{align}
Since $\| \hat{\beta} - \beta_0 \|_{2} = O_\mathbb{P}(n^{-1/2})$, it suffices to work under the setting where $\hat{\beta}$ and $\beta_*$ lie in the neighborhood $\mathcal{N}$ of Assumption 5. Let $M(X)$ be as in Assumption 5.  Write the estimator as \begin{align*}
\hat{\sigma}_n^2 & =   \frac{2}{n(n-1)}  \sum_{i=1}^n \sum_{j \neq i} K^2 \bigg( \frac{X_i -  X_j}{h_n} \bigg)\hat{u}_i^2 \hat{u}_j^2         \\ &   = \frac{2}{n(n-1)}  \sum_{i=1}^n \sum_{j \neq i} K^2 \bigg( \frac{X_i -  X_j}{h_n} \bigg) \bigg \{  u_i^2 u_j^2 + (\hat{u}_{i}^2 - u_{i}^2 ) \hat{u}_{j}^2 + ( \hat{u}_{j}^2 - u_j^2  ) u_{i}^2  \bigg \} \\ & = S_n +  \frac{2}{n(n-1)}  \sum_{i=1}^n \sum_{j \neq i} K^2 \bigg( \frac{X_i -  X_j}{h_n} \bigg)  (\hat{u}_{i}^2 - u_{i}^2 ) \hat{u}_{j}^2 \\ & \; \; \; \; \; \;  \; \; \;   + \frac{2}{n(n-1)}  \sum_{i=1}^n \sum_{j \neq i} K^2 \bigg( \frac{X_i -  X_j}{h_n} \bigg)  ( \hat{u}_{j}^2 - u_j^2  ) u_{i}^2 \\ & = S_n + A_1 + A_2 .
\end{align*}
We will verify that $A_{1} = o_{\mathbb{P}}( \E(H_n^2) )$. The argument for $A_2$ is completely analogous and omitted. Observe that \begin{align}
\left| \hat{u}_i^2- u_i^2 \right| = \left| (\hat{u}_i - u_i)(\hat{u}_i  + u_i )   \right|  & = \left| ( \hat{u}_i - u_i) (\hat{u}_i - u_i + 2 u_i)        \right| \nonumber  \\ &      \leq (\hat{u}_i - u_i)^2 + 2 \left| \hat{u}_i - u_i   \right| \left| u_i \right| . \label{errorb}
\end{align}
From this bound, we obtain  \begin{align*}
\left|  \hat{u}_i^2 - u_i^2         \right| \hat{u}_j^2 & \leq \left|  \hat{u}_i^2 - u_i^2         \right| \left| \hat{u}_j^2  - u_j^2 \right| + \left|  \hat{u}_i^2 - u_i^2         \right| u_j^2 \\  & =  \left|  \hat{u}_i^2 - u_i^2         \right| \big( \left| \hat{u}_j^2  - u_j^2 \right| + u_j^2 \big)       \\ & \leq  \big[ (\hat{u}_i - u_i)^2 + 2 \left| \hat{u}_i - u_i   \right| \left| u_i \right|    \big] \big[   (\hat{u}_j - u_j)^2 + 2 \left| \hat{u}_j - u_j   \right| \left| u_j \right|   + u_j^2      \big] .
\end{align*}
From (\ref{mv2}) and Cauchy-Schwarz we have that \begin{align}
& \left|  \hat{u}_i - u_i  \right| \leq M(X_i) \| \hat{\beta} - \beta_0   \|_{2} + \gamma_n \left| \delta(X_i) \right|  \; , \\ & \left|  \hat{u}_i - u_i  \right| ^2 \leq 2 M^2(X_i) \| \hat{\beta} - \beta_0   \|_{2}^2 + 2 \gamma_n^2 \delta^2(X_i) .
\end{align}
From substituting this into the previous bound, we obtain that \begin{align*}  A_1 \leq  E_1 + E_2 + E_3 + E_4 + E_5 + E_6 
\end{align*}
where \begin{align*}
E_1 = \frac{2}{n(n-1)} \sum_{i=1}^n \sum_{j \neq i} K^2 \bigg(   \frac{X_i - X_j}{h_n}      \bigg)  &  \big[ 2 M^2(X_i) \| \hat{\beta} - \beta_0   \|_{2}^2 + 2 \gamma_n^2 \delta^2(X_i) \big ] \\ & \times   \big[ 2 M^2(X_j) \| \hat{\beta} - \beta_0   \|_{2}^2 + 2 \gamma_n^2 \delta^2(X_j) \big].
\end{align*}

\begin{align*}
E_2 = \frac{2}{n-1}\sum_{i=1}^n \sum_{j \neq i} K^2 \bigg(   \frac{X_i - X_j}{h_n}      \bigg) & \big[ 2 M^2(X_i) \| \hat{\beta} - \beta_0   \|_{2}^2 + 2 \gamma_n^2 \delta^2(X_i) \big ] \\ & \times 2 \big[ M(X_j) \| \hat{\beta} - \beta_0   \|_{2} + \gamma_n \left| \delta(X_j) \right| \big] |u_j| .
\end{align*}

\begin{align*}
E_3 = \frac{2}{n-1}\sum_{i=1}^n \sum_{j \neq i} K^2 \bigg(   \frac{X_i - X_j}{h_n}      \bigg) & \big[ 2 M^2(X_i) \| \hat{\beta} - \beta_0   \|_{2}^2 + 2 \gamma_n^2 \delta^2(X_i) \big ] u_j^2.
\end{align*}

\begin{align*}
E_4 =  \frac{2}{n-1}\sum_{i=1}^n \sum_{j \neq i} K^2 \bigg(   \frac{X_i - X_j}{h_n}      \bigg) & 2 |u_i| \big[ M(X_i) \| \hat{\beta} - \beta_0   \|_{2} + \gamma_n \left| \delta(X_i) \right|  \big] \\ & \times  \big[  2 M^2(X_j) \| \hat{\beta} - \beta_0   \|_{2}^2 + 2 \gamma_n^2 \delta^2(X_j) \big] .
\end{align*}

\begin{align*}
E_5 =  \frac{2}{n-1}\sum_{i=1}^n \sum_{j \neq i} K^2 \bigg(   \frac{X_i - X_j}{h_n}      \bigg) & 2 |u_i| \big[ M(X_i) \| \hat{\beta} - \beta_0   \|_{2} + \gamma_n \left| \delta(X_i) \right|  \big] \\ & \times  2 |u_j| \big[ M(X_j) \| \hat{\beta} - \beta_0   \|_{2} + \gamma_n \left| \delta(X_j) \right|  \big] .
\end{align*}

\begin{align*}
E_6 =  \frac{2}{n-1}\sum_{i=1}^n \sum_{j \neq i} K^2 \bigg(   \frac{X_i - X_j}{h_n}      \bigg) & 2 |u_i| \big[ M(X_i) \| \hat{\beta} - \beta_0   \|_{2} + \gamma_n \left| \delta(X_i) \right|  \big] u_j^2.
\end{align*}
We will verify that $E_i = o_{\mathbb{P}} ( \E(H_n^2)$ for $i=1,\dots,6$. Before proceeding with the bounds, we state a few preliminary facts. By Lemma 3.5 and Assumption 4(i) we have  $\E[F_X(X - h_n \iota , X + h_n \iota)]    \asymp \E[H_n^2]$. By Corollary 3.4(ii) we have $\E(H_n^2) \gtrapprox h_n^q$.

\begin{enumerate}
\item[(1)] Markov's Inequality and $\| \hat{\beta} - \beta_0  \|_{2} = O_{\mathbb{P}}(n^{-1/2}) $ imply \begin{align*}
E_1 = \; &  n^{-2} O_{\mathbb{P}} \bigg(  \E \bigg[  K^2 \bigg(  \frac{X_1 - X_2}{h_n}    \bigg) M^2(X_1) M^2(X_2)       \bigg]           \bigg) \\ & +  n^{-1} \gamma_n^2 O_{\mathbb{P}} \bigg(  \E \bigg[  K^2 \bigg(  \frac{X_1 - X_2}{h_n}    \bigg) M^2(X_1) \delta^2(X_2)       \bigg]           \bigg) \\ &  + n^{-1} \gamma_n^2 O_{\mathbb{P}} \bigg(  \E \bigg[  K^2 \bigg(  \frac{X_1 - X_2}{h_n}    \bigg) M^2(X_2) \delta^2(X_1)       \bigg]           \bigg) \\ & + \gamma_n^4 O_{\mathbb{P}} \bigg(  \E \bigg[  K^2 \bigg(  \frac{X_1 - X_2}{h_n}    \bigg) \delta^2(X_1) \delta^2(X_2)       \bigg]           \bigg).
\end{align*}
By assumption $M(X), \delta(X) \in L^2(X)$. It follows, by Lemma \ref{aux-new}, that all the expectations above are $O(1)$. Substituting $ \gamma_n \lessapprox n^{-1/2} \{ \E[ F_{X}(X - h_n \iota , X + h_n \iota) ]   \}^{-1/4}$ yields \begin{align*} \frac{E_1}{\E(H_n^2)} = O_{\mathbb{P}} \bigg( \frac{1}{n^2 \E(H_n^2) } + \frac{1}{n^2 \{ \E(H_n^2)  \}^{3/2} }  + \frac{1}{n^2 \{ \E(H_n^2) \}^2}              \bigg).
\end{align*}
Since $\E(H_n^2) \gtrapprox h_n^q$ and $n h_n^{q} \uparrow \infty$, it follows that $E_1 = o_{\mathbb{P}}(\E(H_n^2)) $.  \\

\item[(2)]
 Markov's Inequality, Assumption 2(ii) and $\| \hat{\beta} - \beta_0  \|_{2} = O_{\mathbb{P}}(n^{-1/2}) $ imply \begin{align*}
E_2 = \; &  n^{-3/2} O_{\mathbb{P}} \bigg(  \E \bigg[  K^2 \bigg(  \frac{X_1 - X_2}{h_n}    \bigg) M^2(X_1) M(X_2)       \bigg]           \bigg) \\ & +  n^{-1} \gamma_n  O_{\mathbb{P}} \bigg(  \E \bigg[  K^2 \bigg(  \frac{X_1 - X_2}{h_n}    \bigg) M^2(X_1) \left| \delta (X_2) \right|       \bigg]           \bigg) \\ &  + n^{-1/2} \gamma_n^2 O_{\mathbb{P}} \bigg(  \E \bigg[  K^2 \bigg(  \frac{X_1 - X_2}{h_n}    \bigg) \delta^2(X_1) M(X_2)       \bigg]           \bigg) \\ & + \gamma_n^3 O_{\mathbb{P}} \bigg(  \E \bigg[  K^2 \bigg(  \frac{X_1 - X_2}{h_n}    \bigg) \delta^2(X_1) \left| \delta(X_2) \right|       \bigg]           \bigg).
\end{align*}
By Lemma \ref{aux-new}, the first three expectations  are $O(1)$. For the last expectation, we use Lemma \ref{aux-new} and $\| \delta \|_{\infty} < \infty $ to obtain $$  \E \bigg[  K^2 \bigg(  \frac{X_1 - X_2}{h_n}    \bigg) \delta^2(X_1) \left| \delta(X_2) \right|       \bigg] \lessapprox \E[ F_{X}(X - h_n \iota , X + h_n \iota) ]   \lessapprox \E(H_n^2).   $$

Substituting $ \gamma_n \lessapprox n^{-1/2} \{ \E[ F_{X}(X - h_n \iota , X + h_n \iota) ]   \}^{-1/4}$ yields  $$  \frac{E_2}{\E(H_n^2)} = O_{\mathbb{P}} \bigg(  \frac{1}{n^{3/2} \E(H_n^2)} + \frac{1}{n^{3/2} \{ \E(H_n^2)  \}^{5/4}  } + \frac{1}{n^{3/2} \{ \E(H_n^2)   \}^{3/2}  } + \frac{1}{ n^{3/2} \{ \E(H_n^2)   \}^{3/4}   }             \bigg)    $$
Since $\E(H_n^2) \gtrapprox h_n^q$ and $n h_n^{q} \uparrow \infty$, it follows that $E_2 = o_{\mathbb{P}}(\E(H_n^2)) $. \\

\item[(3)]   Markov's Inequality, Assumption 2(ii) and $\| \hat{\beta} - \beta_0  \|_{2} = O_{\mathbb{P}}(n^{-1/2}) $ imply  \begin{align*}
E_3 & =  n^{-1} O_{\mathbb{P}} \bigg( \E \bigg[ K^2 \bigg( \frac{X_1 - X_2}{h_n}    \bigg)   M^2(X_1)     \bigg]    \bigg) \\ & \;  \;  +  \gamma_n^2 O_{\mathbb{P}} \bigg( \E \bigg[ K^2 \bigg( \frac{X_1 - X_2}{h_n}    \bigg)   \delta^2(X_1)     \bigg]    \bigg).
\end{align*}

By Lemma \ref{aux-new},  the first  expectation  is $O(1)$. For the second expectation, we use Lemma \ref{aux-new} and $\| \delta \|_{\infty} < \infty $ to obtain $$  \E \bigg[  K^2 \bigg(  \frac{X_1 - X_2}{h_n}    \bigg) \delta^2(X_1) \bigg] \lessapprox \E[ F_{X}(X - h_n \iota , X + h_n \iota) ]   \lessapprox \E(H_n^2).   $$
Substituting $ \gamma_n \lessapprox n^{-1/2} \{ \E[ F_{X}(X - h_n \iota , X + h_n \iota) ]   \}^{-1/4}$ yields  $$  \frac{E_3}{ \E(H_n^2)} = O_{\mathbb{P}} \bigg(  \frac{1}{n \E(H_n^2)} + \frac{1}{n \{ \E(H_n^2)  \}^{1/2} }   \bigg).  $$
Since $\E(H_n^2) \gtrapprox h_n^q$ and $n h_n^{q} \uparrow \infty$, it follows that $E_3 = o_{\mathbb{P}}(\E(H_n^2)) $. \\

\item[(4)] The argument to show $E_4 = o_{\mathbb{P}}(\E(H_n^2)) $ is completely analogous to the one used for $E_{2}$. \\

\item[(5)] Markov's Inequality, Assumption 2(ii) and $\| \hat{\beta} - \beta_0  \|_{2} = O_{\mathbb{P}}(n^{-1/2}) $ imply 

\begin{align*}
E_5 = \; &  n^{-1} O_{\mathbb{P}} \bigg(  \E \bigg[  K^2 \bigg(  \frac{X_1 - X_2}{h_n}    \bigg) M(X_1) M(X_2)       \bigg]           \bigg) \\ & +  n^{-1/2} \gamma_n  O_{\mathbb{P}} \bigg(  \E \bigg[  K^2 \bigg(  \frac{X_1 - X_2}{h_n}    \bigg) M(X_1) \left| \delta (X_2) \right|       \bigg]           \bigg) \\ &  + n^{-1/2} \gamma_n O_{\mathbb{P}} \bigg(  \E \bigg[  K^2 \bigg(  \frac{X_1 - X_2}{h_n}    \bigg) \left|\delta(X_1) \right| M(X_2)       \bigg]           \bigg) \\ & + \gamma_n^2 O_{\mathbb{P}} \bigg(  \E \bigg[  K^2 \bigg(  \frac{X_1 - X_2}{h_n}    \bigg) \left|\delta(X_1) \right| \left| \delta(X_2) \right|       \bigg]           \bigg).
\end{align*}
By Lemma \ref{aux-new},  the first  expectation  is $O(1)$. For the other expectations, we use Lemma \ref{aux-new} and Lemma \ref{aux-estimates}(iii) to obtain \begin{align*}
 & \E \bigg[  K^2 \bigg(  \frac{X_1 - X_2}{h_n}    \bigg) M(X_1) \left| \delta (X_2) \right|       \bigg]   +  \E \bigg[  K^2 \bigg(  \frac{X_1 - X_2}{h_n}    \bigg) \left| \delta(X_1) \right|  M(X_2)      \bigg] \\ & + \E \bigg[  K^2 \bigg(  \frac{X_1 - X_2}{h_n}    \bigg) \left|\delta(X_1) \right| \left| \delta(X_2) \right|       \bigg]      \\ & \lessapprox \E \big[  \left| \delta(X) \right| \Omega_{M}(X- h \iota , X + h \iota)\big] + \E \big[  \left| \delta(X) \right| \Omega_{\left| \delta \right|}(X- h \iota , X + h \iota)\big]    \\  & = o \big( \{\E[F_X(X - h_n \iota , X + h_n \iota)] \}^{1/2}  \big)  \\ & = o \big( \{  \E(H_n^2) \}^{1/2}   \big).
\end{align*}
Substituting $ \gamma_n \lessapprox n^{-1/2} \{ \E[ F_{X}(X - h_n \iota , X + h_n \iota) ]   \}^{-1/4}$ yields  $$  \frac{E_5}{ \E(H_n^2)} = O_{\mathbb{P}} \bigg( \frac{1}{n \E(H_n^2)}  + \frac{ \{ \E(H_n^2)   \}^{1/2}   }{n \{ \E(H_n^2)     \}^{5/4} }    +  \frac{ \{ \E(H_n^2)   \}^{1/2}   }{n \{  \E(H_n^2)    \}^{3/2}}       \bigg)   .     $$
Since $\E(H_n^2) \gtrapprox h_n^q$ and $n h_n^{q} \uparrow \infty$, it follows that $E_5 = o_{\mathbb{P}}(\E(H_n^2)) $. \\

\item[(6)] Markov's Inequality, Assumption 2(ii) and $\| \hat{\beta} - \beta_0  \|_{2} = O_{\mathbb{P}}(n^{-1/2}) $ imply  \begin{align*}
E_6 =  n^{-1/2} O_{\mathbb{P}} \bigg(  \E \bigg[  K^2 \bigg(  \frac{X_1 - X_2}{h_n}    \bigg) M(X_1)     \bigg]        \bigg) + \gamma_n O_{\mathbb{P}} \bigg(  \E \bigg[  K^2 \bigg(  \frac{X_1 - X_2}{h_n}    \bigg) \left| \delta(X_1) \right|     \bigg]        \bigg).
\end{align*}
Lemma \ref{aux-new} and \ref{aux-estimates}(i) imply \begin{align*}
 & \E \bigg[  K^2 \bigg(  \frac{X_1 - X_2}{h_n}    \bigg) M(X_1)     \bigg]   +   \E \bigg[  K^2 \bigg(  \frac{X_1 - X_2}{h_n}    \bigg) \left| \delta(X_1) \right|     \bigg]       \\ & \lessapprox  \E \big[  M(X) F_X(X - h \iota , X + h \iota)         \big] + \E \big[  \left| \delta(X) \right| F_X(X - h \iota , X + h \iota)         \big] \\ & = o \big( \{\E[F_X(X - h_n \iota , X + h_n \iota)] \}^{3/4}  \big)  \\ & = o \big( \{  \E(H_n^2) \}^{3/4}   \big).
\end{align*}

Substituting $ \gamma_n \lessapprox n^{-1/2} \{ \E[ F_{X}(X - h_n \iota , X + h_n \iota) ]   \}^{-1/4}$ yields   $$   \frac{E_6}{ \E(H_n^2)} = O_{\mathbb{P}} \bigg(  \frac{ \{ \E(H_n^2)   \}^{3/4} }{\sqrt{n} \E(H_n^2)  }   + \frac{\{ \E(H_n^2)   \}^{3/4}}{ \sqrt{n} \{ \E(H_n^2)   \}^{5/4}  }   \bigg) .    $$
Since $\E(H_n^2) \gtrapprox h_n^q$ and $n h_n^{q} \uparrow \infty$, it follows that $E_6 = o_{\mathbb{P}}(\E(H_n^2)) $.

\end{enumerate}

\end{proof}

\begin{lemma}
\label{varest2}
Suppose the conditions of Theorem $3.11$ hold and $h_n \downarrow 0 \,,\, n h_n^{q} \uparrow \infty$. Assume  the alternative hypothesis $H_1$ in (\ref{reg}) holds with $\delta \in L^{2}(X)$ and  $ \gamma_n \lessapprox n^{-1/2} \{ \E[ F_X(X - h_n \iota , X + h_n \iota) ]   \}^{1/4} h_n^{-s_{\delta}q/2} $. Then

\begin{equation*}
\label{varequiv2}
\left|  \hat{\sigma}_{n}^2 - 2 \E(H_n^2)     \right|  =o_{\mathbb{P}} \big(  \E(H_n^2)   \big) .
\end{equation*}
\end{lemma}

\begin{proof}
[Proof of Lemma \ref{varest2}] 
The proof is analogous to that of Lemma \ref{varest}.  The main difference being that we  allow for $\delta \in L^2(X)$ and that the upper bound on the rate $\gamma_n$ in this Lemma may be larger than allowed for in Lemma \ref{varest}. 

 Let $E_1, \dots, E_6$ be as defined in the proof of Lemma \ref{varest}. We will verify that $E_i = o_{\mathbb{P}} ( \E(H_n^2)$ for $i=1,\dots,6$.  Before proceeding with the bounds, we state a few preliminary facts.  As shown in the proof of Theorem $3.11$, there exists a  $ M < \infty$ such that  \begin{align}
\label{delres22} \mathbb{P} \bigg(  \mathbbm{1} \big \{ X \in \mathcal{S}_{\delta}       \big \}   \frac{F_X(X - h \iota , X + h \iota)}{(2h)^{s_{\delta}q}}  \leq M          \bigg) = 1
\end{align}
holds for all sufficiently small $h > 0 $.  By Lemma 3.7, we have $F_{X} \in \mathcal{D}(s)$ for some $s \in (0,1]$ and $\E[F_X(X - h_n \iota , X + h_n \iota)] \asymp h_n^{sq}$.  Assumptions (4,  6) are automatically satisfied.  Then, by Lemma 3.5 and Assumption 4(i), we have $ \E[H_n^2] \asymp h_n^{sq} $ as well.

\begin{enumerate}
\item[(1)] Markov's Inequality and $\| \hat{\beta} - \beta_0  \|_{2} = O_{\mathbb{P}}(n^{-1/2}) $ imply \begin{align*}
E_1 = \; &  n^{-2} O_{\mathbb{P}} \bigg(  \E \bigg[  K^2 \bigg(  \frac{X_1 - X_2}{h_n}    \bigg) M^2(X_1) M^2(X_2)       \bigg]           \bigg) \\ & +  n^{-1} \gamma_n^2 O_{\mathbb{P}} \bigg(  \E \bigg[  K^2 \bigg(  \frac{X_1 - X_2}{h_n}    \bigg) M^2(X_1) \delta^2(X_2)       \bigg]           \bigg) \\ &  + n^{-1} \gamma_n^2 O_{\mathbb{P}} \bigg(  \E \bigg[  K^2 \bigg(  \frac{X_1 - X_2}{h_n}    \bigg) M^2(X_2) \delta^2(X_1)       \bigg]           \bigg) \\ & + \gamma_n^4 O_{\mathbb{P}} \bigg(  \E \bigg[  K^2 \bigg(  \frac{X_1 - X_2}{h_n}    \bigg) \delta^2(X_1) \delta^2(X_2)       \bigg]           \bigg).
\end{align*}
By assumption $M(X), \delta(X) \in L^2(X)$. It follows, by Lemma \ref{aux-new}, that all the expectations above are $O(1)$.  Substituting $ \gamma_n \lessapprox n^{-1/2} \{ \E[ F_X(X - h_n \iota , X + h_n \iota) ]   \}^{1/4} h_n^{-s_{\delta}q/2}$ yields \begin{align*} \frac{E_1}{\E(H_n^2)} = O_{\mathbb{P}} \bigg( \frac{1}{n^2 h_n^{sq}  } + \frac{  h_n^{sq/2} }{n^2  h_n^{sq}  h_n^{s_{\delta}q}  }  + \frac{h_n^{sq}}{n^2 h_n^{sq}   h_n^{2 s_{\delta}q }}              \bigg).
\end{align*}
Since $s \leq s_{\delta}\leq 1 $ and $n h_n^{q} \uparrow \infty$, it follows that $E_1 = o_{\mathbb{P}}(\E(H_n^2)) $.  \\

\item[(2)]
 Markov's Inequality, Assumption 2(ii) and $\| \hat{\beta} - \beta_0  \|_{2} = O_{\mathbb{P}}(n^{-1/2}) $ imply \begin{align*}
E_2 = \; &  n^{-3/2} O_{\mathbb{P}} \bigg(  \E \bigg[  K^2 \bigg(  \frac{X_1 - X_2}{h_n}    \bigg) M^2(X_1) M(X_2)       \bigg]           \bigg) \\ & +  n^{-1} \gamma_n  O_{\mathbb{P}} \bigg(  \E \bigg[  K^2 \bigg(  \frac{X_1 - X_2}{h_n}    \bigg) M^2(X_1) \left| \delta (X_2) \right|       \bigg]           \bigg) \\ &  + n^{-1/2} \gamma_n^2 O_{\mathbb{P}} \bigg(  \E \bigg[  K^2 \bigg(  \frac{X_1 - X_2}{h_n}    \bigg) \delta^2(X_1) M(X_2)       \bigg]           \bigg) \\ & + \gamma_n^3 O_{\mathbb{P}} \bigg(  \E \bigg[  K^2 \bigg(  \frac{X_1 - X_2}{h_n}    \bigg) \delta^2(X_1) \left| \delta(X_2) \right|       \bigg]           \bigg).
\end{align*}
By Lemma \ref{aux-new},  the first three expectations  are $O(1)$ and \begin{align*}
 \E \bigg[  K^2 \bigg(  \frac{X_1 - X_2}{h_n}    \bigg) \delta^2(X_1) \left| \delta(X_2) \right|       \bigg]     \lessapprox \E \big[  \delta^2(X) \Omega_{\left|\delta \right|}(X- h_n \iota , X + h_n \iota)\big]   .
\end{align*}
By Cauchy-Schwarz, we obtain
\begin{align*}    \Omega_{|\delta|}(X -  h_n \iota , X +  h_n \iota)  \leq  \sqrt{ \E[ \delta^2(X)  ]  } \sqrt{F_X(X  -  h_n \iota , X +  h_n \iota)}  .
\end{align*}
From this bound and  (\ref{delres22}) it follows that \begin{align*}
\E \bigg[  K^2 \bigg(  \frac{X_1 - X_2}{h_n}    \bigg) \delta^2(X_1) \left| \delta(X_2) \right|       \bigg] & \lessapprox \E  \big[  \mathbbm{1} \{ X \in \mathcal{S}_{\delta}  \}   \delta^2(X) \sqrt{F_X(X  -  h_n \iota , X +  h_n \iota)}     \; \big  ] \\ & \lessapprox h_n^{s_{\delta}q/2}.
\end{align*}

Substituting $ \gamma_n \lessapprox n^{-1/2} \{ \E[ F_X(X - h_n \iota , X + h_n \iota) ]   \}^{1/4} h_n^{-s_{\delta}q/2}$ yields $$  \frac{E_2}{\E(H_n^2)} = O_{\mathbb{P}} \bigg(  \frac{1}{n^{3/2}  h_n^{sq}  } + \frac{h_n^{sq/4}}{n^{3/2} h_n^{s_{\delta}q/2} h_n^{sq}    } + \frac{h_n^{sq/2}}{n^{3/2} h_n^{s_{\delta}q} h_n^{sq}  } + \frac{h_n^{3sq/4} h_n^{s_{\delta}q/2}  }{ n^{3/2} h_n^{1.5 s_{\delta}q}   h_n^{sq}    }             \bigg)    $$
Since $s \leq s_{\delta}\leq 1 $ and $n h_n^{q} \uparrow \infty$, it follows that $E_2 = o_{\mathbb{P}}(\E(H_n^2)) $. \\

\item[(3)]

 Markov's Inequality, Assumption 2(ii) and $\| \hat{\beta} - \beta_0  \|_{2} = O_{\mathbb{P}}(n^{-1/2}) $ imply  \begin{align*}
E_3 & =  n^{-1} O_{\mathbb{P}} \bigg( \E \bigg[ K^2 \bigg( \frac{X_1 - X_2}{h_n}    \bigg)   M^2(X_1)     \bigg]    \bigg) \\ & \;  \;  +  \gamma_n^2 O_{\mathbb{P}} \bigg( \E \bigg[ K^2 \bigg( \frac{X_1 - X_2}{h_n}    \bigg)   \delta^2(X_1)     \bigg]    \bigg).
\end{align*}

By Lemma \ref{aux-new},  the first  expectation  is $O(1)$ and \begin{align*}
 \E \bigg[  K^2 \bigg(  \frac{X_1 - X_2}{h_n}    \bigg) \delta^2(X_1)       \bigg]    & \lessapprox \E \big[ \mathbbm{1} \{ X \in \mathcal{S}_{\delta}   \} \delta^2(X) F_{X} (X- h_n \iota , X + h_n \iota)\big]  \\ & \lessapprox h_n^{ s_{\delta}q} \; ,
\end{align*}
where the last expression follows from (\ref{delres22}). Substituting $ \gamma_n \lessapprox n^{-1/2} \{ \E[ F_X(X - h_n \iota , X + h_n \iota) ]   \}^{1/4} h_n^{-s_{\delta}q/2}$ yields $$  \frac{E_3}{ \E(H_n^2)} = O_{\mathbb{P}} \bigg(  \frac{1}{n h_n^{sq}   } + \frac{ h_n^{sq/2} h_n^{s_{\delta}q}  }{n h_n^{s_{\delta}q} h_n^{sq}   }   \bigg).  $$
Since $s \leq s_{\delta}\leq 1 $ and $n h_n^{q} \uparrow \infty$, it follows that $E_3 = o_{\mathbb{P}}(\E(H_n^2)) $. \\

\item[(4)]

 The argument to show $E_4 = o_{\mathbb{P}}(\E(H_n^2)) $ is completely analogous to the one used for $E_{2}$. \\
 
 \item[(5)] Markov's Inequality, Assumption 2(ii) and $\| \hat{\beta} - \beta_0  \|_{2} = O_{\mathbb{P}}(n^{-1/2}) $ imply 

\begin{align*}
E_5 = \; &  n^{-1} O_{\mathbb{P}} \bigg(  \E \bigg[  K^2 \bigg(  \frac{X_1 - X_2}{h_n}    \bigg) M(X_1) M(X_2)       \bigg]           \bigg) \\ & +  n^{-1/2} \gamma_n  O_{\mathbb{P}} \bigg(  \E \bigg[  K^2 \bigg(  \frac{X_1 - X_2}{h_n}    \bigg) M(X_1) \left| \delta (X_2) \right|       \bigg]           \bigg) \\ &  + n^{-1/2} \gamma_n O_{\mathbb{P}} \bigg(  \E \bigg[  K^2 \bigg(  \frac{X_1 - X_2}{h_n}    \bigg) \left|\delta(X_1) \right| M(X_2)       \bigg]           \bigg) \\ & + \gamma_n^2 O_{\mathbb{P}} \bigg(  \E \bigg[  K^2 \bigg(  \frac{X_1 - X_2}{h_n}    \bigg) \left|\delta(X_1) \right| \left| \delta(X_2) \right|       \bigg]           \bigg).
\end{align*}
By Lemma \ref{aux-new},  the first  expectation  is $O(1)$ and \begin{align*}
 & \E \bigg[  K^2 \bigg(  \frac{X_1 - X_2}{h_n}    \bigg) M(X_1) \left| \delta (X_2) \right|       \bigg]   +  \E \bigg[  K^2 \bigg(  \frac{X_1 - X_2}{h_n}    \bigg) \left| \delta(X_1) \right|  M(X_2)      \bigg] \\ & + \E \bigg[  K^2 \bigg(  \frac{X_1 - X_2}{h_n}    \bigg) \left|\delta(X_1) \right| \left| \delta(X_2) \right|       \bigg]      \\ & \lessapprox \E \big[  \left| \delta(X) \right| \Omega_{M}(X- h \iota , X + h \iota)\big] + \E \big[  \left| \delta(X) \right| \Omega_{\left| \delta \right|}(X- h \iota , X + h \iota)\big]    
\end{align*}
Repeating the argument from the bounds for $E_2$ yields \begin{align*}
  &  \E \big[  \left| \delta(X) \right| \Omega_{M}(X- h \iota , X + h \iota)\big] + \E \big[  \left| \delta(X) \right| \Omega_{\left| \delta \right|}(X- h \iota , X + h \iota)\big]    \\ & \lessapprox h_n^{s_{\delta}q/2}.
\end{align*}

Substituting $ \gamma_n \lessapprox n^{-1/2} \{ \E[ F_X(X - h_n \iota , X + h_n \iota) ]   \}^{1/4} h_n^{-s_{\delta}q/2}$ yields  $$  \frac{E_5}{ \E(H_n^2)} = O_{\mathbb{P}} \bigg( \frac{1}{n h_n^{sq}  }  + \frac{ h_n^{sq/4} h_n^{s_{\delta}q/2}       }{n  h_n^{s_{\delta}q/2} h_n^{sq}   }    +  \frac{ h_n^{sq/2} h_n^{s_{\delta}q/2}   }{n h_n^{s_{\delta}q} h_n^{sq}   }       \bigg)   .     $$
Since $s \leq s_{\delta}\leq 1 $ and $n h_n^{q} \uparrow \infty$, it follows that $E_5 = o_{\mathbb{P}}(\E(H_n^2)) $. \\

\item[(6)]

Markov's Inequality, Assumption 2(ii) and $\| \hat{\beta} - \beta_0  \|_{2} = O_{\mathbb{P}}(n^{-1/2}) $ imply  \begin{align*}
E_6 =  n^{-1/2} O_{\mathbb{P}} \bigg(  \E \bigg[  K^2 \bigg(  \frac{X_1 - X_2}{h_n}    \bigg) M(X_1)     \bigg]        \bigg) + \gamma_n O_{\mathbb{P}} \bigg(  \E \bigg[  K^2 \bigg(  \frac{X_1 - X_2}{h_n}    \bigg) \left| \delta(X_1) \right|     \bigg]        \bigg).
\end{align*}

By Lemma \ref{aux-new}, $F_{X} \in \mathcal{D}(s)$ and (\ref{delres}) we obtain \begin{align*}
 & \E \bigg[  K^2 \bigg(  \frac{X_1 - X_2}{h_n}    \bigg) M(X_1)     \bigg]   +   \E \bigg[  K^2 \bigg(  \frac{X_1 - X_2}{h_n}    \bigg) \left| \delta(X_1) \right|     \bigg]       \\ & \lessapprox  \E \big[  M(X) F_X(X - h \iota , X + h \iota)         \big] + \E \big[  \left| \delta(X) \right| F_X(X - h \iota , X + h \iota)         \big] \\ & \lessapprox h_n^{sq} + h_n^{s_{\delta}q} \\ & \lessapprox h_n^{s q}.
\end{align*}

Substituting $ \gamma_n \lessapprox n^{-1/2} \{ \E[ F_X(X - h_n \iota , X + h_n \iota) ]   \}^{1/4} h_n^{-s_{\delta}q/2}$ yields   $$   \frac{E_6}{ \E(H_n^2)} = O_{\mathbb{P}} \bigg(  \frac{ h_n^{sq} }{\sqrt{n}  h_n^{sq}  }   + \frac{h_n^{sq} h_n^{sq/4}  }{ \sqrt{n} h_n^{s_{\delta}q/2} h_n^{sq}    }   \bigg) .    $$
Since $s \leq s_{\delta}\leq 1 $ and $n h_n^{q} \uparrow \infty$, it follows that $E_6 = o_{\mathbb{P}}(\E(H_n^2)) $.

\end{enumerate}

\end{proof}

\section{Proofs}
\label{sec2sup}

\subsection{Lemma 3.2}

\begin{proof}
[Proof of Lemma 3.2] The function $A\rightarrow
\int_{A}f(x)dF_{X}(x) $ is a finite Borel measure on
$\mathbb{R}^{q}$, which we will refer to as $\Omega_{f}$. Since $g$ has support contained in the closure of $\mathcal{O}$, the expectation can be expressed as
\[
\int_{\mathbb{R}^{q}}f(x)g(x)dF_{X}(x)=\int_{ \R^q}   g(x) \prod_{i=1}^q  \mathbbm{1} \{l_i \leq x_i \leq u_i \} d\Omega
_{f}(x) .
\]

At the boundary of $\mathcal{O}$, we define the mixed partials of $g$ through continuous extension. Fix any  $x = (x_1, \dots , x_q) \in  \mathcal{O}  $. The assumed hypothesis on $g$ and its mixed partials imply that 
\[
g(x)=(-1)^{q}\int_{x_{q}}^{u_{q}}\dots\int_{x_{1}}^{u_{1}}\partial
_{t}g(t)dt_{1}\dots dt_{q}.
\]
Substitution of this identity and an application of Fubini's Theorem yields
\begin{align*}
\int_{\mathbb{R}^{q}}f(x)g(x)d F_{X}(x)  &  = (-1)^{q} \int_{ \R^q }
\int_{\mathbb{R}^{q}} \prod_{i=1}^{q}  \mathbbm{1} \{l_i \leq x_i \leq u_i \}    \mathbbm{1} \{x_{i} \leq  t_{i} \leq u_{i}\} \partial_{t} g(t) dt d
\Omega_{f}(x)\\
&  = (-1)^{q} \int_{\mathbb{R}^{q}} \int_{\mathbb{R}^{q}} \prod_{i=1}^{q} 
\mathbbm{1} \{l_{i} \leq x_{i} \leq u_{i} \} \mathbbm{1}\{x_{i} \leq t_{i} \leq
u_{i} \} \partial_{t} g(t) d \Omega_{f}(x) dt\\ & = (-1)^{q} \int_{\mathbb{R}^{q}} \int_{\mathbb{R}^{q}} \prod_{i=1}^{q} 
\mathbbm{1} \{l_{i} \leq x_{i} \leq t_{i} \} \mathbbm{1} \{l_i \leq  t_i \leq u_i  \}  \partial_{t} g(t) d \Omega_{f}(x) dt \\
&  = (-1)^{q}\int_{\R^q}  \prod_{i=1}^q  \mathbbm{1} \big \{ l_i \leq t_i \leq u_i   \big \}   \Omega_{f}(l,t)\partial_{t}g(t)dt .
\end{align*}

\end{proof}

\subsection{Lemma 3.3}
\begin{proof}
[Proof of Lemma 3.3]
The proof of part $(i)$ is provided in the main text. We focus on $(ii)$ and $(iii)$ here.
\begin{itemize} 
\item[(ii)] The proof follows the same steps as in part $(i)$. \begin{align*}
\mathbb{E}[H_{n}^{4}(Z_{1},Z_{2})] &  =\mathbb{E}\bigg[\mu_{4}(X_{1})\mu
_{4}(X_{2})K^{4}\bigg(\frac{X_{1}-X_{2}}{h_{n}}\bigg)\bigg]\\
&  =\mathbb{E}\bigg[\mu_{4}(X_{2})\int_{\mathbb{R}^{q}}\mu_{4}(x)K^{4}%
\bigg(\frac{x-X_{2}}{h_{n}}\bigg)d F_{X} (x)\bigg].
\end{align*}
Define
\[
f(x)=\mu_{4}(x),\;g(x)=K^{4}\bigg(\frac{x-X_{2}}{h_{n}}\bigg).
\]
Let $X_2^i$ denote the $i^{th}$ coordinate of $X_2$. Conditional on $X_{2}$, $(f,g)$ satisfy the hypothesis of Lemma 3.2 with  $ \mathcal{O}  =  (X_{2}^{1}%
-h_{n},X_{2}^{1}+h_{n}) \times \dots \times (X_{2}^{q}%
-h_{n},X_{2}^{q}+h_{n}) $.  Applying Lemma 3.2 yields
\begin{align*}
& \int_{\mathbb{R}^{q}}\mu_{4}(x)K^{4}\bigg(\frac{x-X_{2}}{h_{n}}\bigg)dF_{X}(x)
 \\ &  = (-1)^q \int_{\mathbb{R}^{q}} \mathbbm{1} \{ t \in \mathcal{O}  \}   \Omega_{4}(X_{2}-h_{n}\iota,t)\partial_{t}%
K^{4}\bigg(\frac{t-X_{2}}{h_{n}}\bigg)dt\\
&  = (-1)^q \int_{\left[  -1,1\right]  ^{q}}\Omega_{4}(X_{2}-h_{n}\iota,X_{2}%
+h_{n}v)\partial_{v}K^{4}(v)dv\;,
\end{align*}
where the last equality follows from the change of variables $t\rightarrow
X_{2}+h_{n}v$. It follows that
\[
\mathbb{E}[H_{n}^{4}(Z_{1},Z_{2})]=\mathbb{E}\bigg(\mu_{4}(X)(-1)^{q}%
\int_{\left[  -1,1\right]  ^{q}}\Omega_{4}(X-h_{n}\iota,X+h_{n}v)\partial
_{v}K^{4}(v)dv\bigg).
\]
Define $(v_{1},v_{-1})$ to be the partitioned
vector $(v_{1},\dots,v_{q})$ with $v_{-1}=\left(  v_{2},...,v_{q}\right)  $.
For any fixed choice of $v_{-1}\in\mathbb{R}^{q-1}$, we have
\begin{align*}
&  \int_{[-1,1]}\Omega_{4}(X-h_{n}\iota,X+h_{n}(v_{1},v_{-1}))\partial_{v_{1}}k^{4}(v_{1})dv_{1}\\
& \qquad\qquad =     \int_{[0,1]} \Omega_{4}(X-h_{n}\iota,X+h_{n}(v_{1},v_{-1}))\partial_{v_{1}}k^{4}(v_{1})dv_{1}  \\ 
  & \qquad\qquad\qquad\qquad
  +  \int_{[-1,0]}  \Omega_{4}(X-h_{n}\iota,X+h_{n}(v_{1},v_{-1}))\partial_{v_{1}}k^{4}(v_{1})dv_{1}  \\
& \qquad\qquad = \int_{[0,1]} \Omega_{4}(X-h_{n}\iota,X+h_{n}(v_{1},v_{-1}))\partial_{v_{1}}k^{4}(v_{1})dv_{1} \\
  & \qquad\qquad\qquad\qquad
  -\int_{[0,1]}\Omega_{4}(X-h_{n}\iota,X+h_{n}(-v_{1},v_{-1}))\partial_{v_{1}}k^{4}(v_{1})dv_{1} \\
& \qquad\qquad = \int_{[0,1]} \Omega_{4}(X-h_{n}(v_{1},\iota),X+h_{n}(v_{1},v_{-1}))\partial_{v_{1}}k^{4}(v_{1})dv_{1}\;,
\end{align*}
where the second  equality follows from the change of variables $v_1 \rightarrow -v_1$ and $  \partial_{v_1}  k^{4}%
(-v_{1})=- \partial_{v_1} k^{4}(v_{1})$ ($k$ is a symmetric function). Iterating this procedure from $v_{1}$ to $v_{q}$ yields
\[
\int_{\lbrack-1,1]^{q}}\Omega_{4}(X-h_{n}\iota,X+h_{n}v)\partial_{v}%
K^{4}(v)dv=\int_{[0,1]^{q}}\Omega_{4}\big(X-h_{n}v,X+h_{n}v\big)\partial
_{v}K^{4}(v)dv.
\]
The expression for $\E(H_n^4)$ follows from substituting  $(-1)^q \partial_{v} K^4(v) = \partial_{v} K^4(-v)$.
\item[(iii)] \begin{align*}
&  \mathbb{E}\big[  G_{n}^2(Z_{1},Z_{2})\big] \\
&  =\mathbb{E}\bigg(\mu_{2}(X_{1})\mu_{2}(X_{2})\bigg[\int_{\mathbb{R}^{q}}%
\mu_{2}(x)K\bigg(\frac{x-X_{1}}{h_{n}}\bigg)K\bigg(\frac{x-X_{2}}{h_{n}%
}\bigg)dF_{X}(x)\bigg]^{2}\bigg).
\end{align*}
Define
\[
f(x)=\mu_{2}(x),\;g(x)=K\bigg(\frac{x-X_{1}}{h_{n}}\bigg)K\bigg(\frac{x-X_{2}%
}{h_{n}}\bigg).
\]
Conditional on $(X_{1}, X_{2})$,  $(f,g)$ satisfy the hypothesis of Lemma 3.2 with  $ \mathcal{O} = \big \{ x \in \R^q :   \max( \| x - X_1 \|_{\infty}, \|  x-  X_2 \|_{\infty}   ) < h_n   \big \}  $. Applying Lemma 3.2 yields \begin{align*}
&  \int_{\mathbb{R}^{q}}   \mu_{2}(x)K\left(  \frac{x-X_{1}}{h_{n}}\right)
K\left(  \frac{x-X_{2}}{h_{n}}\right)  dF_{X}(x)\\
&  =(-1)^{q}\int_{\mathbb{R}^{q}}  \mathbbm{1} \big \{ t \in \mathcal{O}     \big \}   \Omega_{2}\big( \max \big \{ X_{1}, X_2 \big \} -h_{n}\iota,t \big) \partial
_{t}\bigg[K\bigg(\frac{t-X_{1}}{h_{n}}\bigg)K\bigg(\frac{t-X_{2}}{h_{n}%
}\bigg)\bigg]    dt \\
&  =(-1)^{q}\int_{\left[  -1,1\right]  ^{q}}  \mathbbm{1} \big \{ v + h_n^{-1}(X_1-X_2) \in [-1,1]^q       \big \}   \Omega_{2}( \max \big \{ X_{1}, X_2 \big \} -h_{n}\iota
,X_{1}+h_{n}v) \\ &  \; \; \; \; \; \; \;  \; \; \; \; \; \; \;  \; \; \;  \; \; \; \; \; \; \;  \; \times \partial_{v}\bigg[K(v)K\bigg(\frac{X_{1}-X_{2}}{h_{n}%
}+v\bigg)\bigg]dv .
\end{align*}
where the last equality follows from the change of variables $t\rightarrow
X_{1}+h_{n}v$. By using H\"{o}lder's inequality on the final expression above we obtain \begin{align*} &
 \bigg[ \int_{\mathbb{R}^{q}}   \mu_{2}(x)K\left(  \frac{x-X_{1}}{h_{n}}\right)
K\left(  \frac{x-X_{2}}{h_{n}}\right)  dF_{X}(x) \bigg]^2 \\ & \leq 2^{q} \int_{\left[  -1,1\right]  ^{q}}  \mathbbm{1} \big \{ v + h_n^{-1}(X_1-X_2) \in [-1,1]^q       \big \}   \Omega_{2}^2( \max \big \{ X_{1}, X_2 \big \} -h_{n}\iota
,X_{1}+h_{n}v) \\ &  \; \; \; \; \; \; \;  \; \; \; \; \; \; \;  \; \; \;  \; \; \; \; \; \;  \times \bigg( \partial_{v}\bigg[K(v)K\bigg(\frac{X_{1}-X_{2}}{h_{n}%
}+v\bigg)\bigg] \bigg)^2   dv .
\end{align*}
Note that $ \sup \limits_{v \in [-1,1]^q} \Omega_{2}^2( \max \big \{ X_{1}, X_2 \big \} -h_{n}\iota
,X_{1}+h_{n}v) \leq \Omega_2^2 (X_1 - h_n \iota , X_1 + h_n \iota) $ from which we obtain \begin{align*}
&  \mathbb{E}\big[  G_{n}^2(Z_{1},Z_{2})\big] \\ &  \leq 2^{q} \E \bigg( \mu_{2}(X_1)      \Omega_2^2 (X_1 - h_n \iota , X_1 + h_n \iota) \mu_2(X_2)  \\ & \; \; \; \; \; \; \; \; \; \; \; \times   \int_{\left[  -1,1\right]  ^{q}}  \mathbbm{1} \big \{ v + h_n^{-1}(X_1-X_2) \in [-1,1]^q  \big \}  \bigg( \partial_{v}\bigg[K(v)K\bigg(\frac{X_{1}-X_{2}}{h_{n}%
}+v\bigg)\bigg] \bigg)^2   dv       \bigg) \\ & = 2^q  \E \bigg( \mu_{2}(X_1)      \Omega_2^2 (X_1 - h_n \iota , X_1 + h_n \iota) \int_{\R^q} \mu_2(x)  \\ & \; \; \; \; \; \; \; \; \; \; \; \times   \int_{\left[  -1,1\right]  ^{q}}  \mathbbm{1} \big \{ v + h_n^{-1}(X_1-x) \in [-1,1]^q  \big \}  \bigg( \partial_{v}\bigg[K(v)K\bigg(\frac{X_{1}-x}{h_{n}%
}+v\bigg)\bigg] \bigg)^2   dv \, d F_X(x)      \bigg).
\end{align*}
Define
\begin{align*}
f(x)  &  =\mu_{2}(x),\;\;\;\\
g(x)  &  =\int_{\left[  -1,1\right]  ^{q}}  \mathbbm{1} \big \{ v + h_n^{-1}(X_1-x) \in [-1,1]^q  \big \}  \bigg( \partial_{v}\bigg[K(v)K\bigg(\frac{X_{1}-x}{h_{n}%
}+v\bigg)\bigg] \bigg)^2   dv .
\end{align*}
Let $X_1^i$ denote the $i^{th}$ coordinate of $X_1$. Since $K$ is a product kernel, we have that $$ \bigg( \partial_{v}\bigg[K(v)K\bigg(\frac{X_{1}-x}{h_{n}%
}+v\bigg)\bigg] \bigg)^2 = \prod_{i=1}^q \big \{ \partial_{v_i} \big[ k(v_i) k( h_n^{-1} [X_1^i - x_i] + v_i  )     \big]         \big \}^2   . $$
By Lemma \ref{app1} with $\phi(t) = t^2$, $u_i = h_n^{-1}(X_1^i- x_i)$, the function $g(x)$ is continuously differentiable on $u_i \in [-2,2]$. It follows that  $(f,g)$ satisfy the hypothesis of Lemma 3.2 with  $ \mathcal{O}  =  (X_{1}^{1}%
-2 h_{n},X_{1}^{1}+2h_{n}) \times \dots \times (X_{1}^{q}%
-2h_{n},X_{1}^{q}+2h_{n}) $. Applying Lemma 3.2  yields \begin{align*}
& \int_{\R^q} \mu_2(x) \int_{\left[  -1,1\right]  ^{q}}  \mathbbm{1} \big \{ v + h_n^{-1}(X_1-x) \in [-1,1]^q  \big \}  \bigg( \partial_{v}\bigg[K(v)K\bigg(\frac{X_{1}-x}{h_{n}%
}+v\bigg)\bigg] \bigg)^2   dv \, d F_X(x) \\ &  = (-1)^q \int_{\R^q} \mathbbm{1} \{ t \in \mathcal{O}  \} \Omega_2(X_1 - h_n \iota , t) \\ & \; \; \; \; \; \; \;  \; \; \; \; \;  \; \; \; \;  \times  \partial_{t} \bigg[ \int_{\left[  -1,1\right]  ^{q}}  \mathbbm{1} \big \{ v + h_n^{-1}(X_1-t) \in [-1,1]^q  \big \}  \bigg( \partial_{v}\bigg[K(v)K\bigg(\frac{X_{1}-t}{h_{n}%
}+v\bigg)\bigg] \bigg)^2   dv \bigg] dt \\ & = \int_{[-2,2]^q} \Omega_2(X_1 - h_n \iota , X_1 + h_n u)\partial_{u} \bigg[ \int_{\left[  -1,1\right]  ^{q}}  \mathbbm{1} \big \{ v -u \in [-1,1]^q  \big \}  \big( \partial_{v}\big[K(v)K(v-u)  \big] \big)^2   dv \bigg] du \, ,
\end{align*}
where the last equality follows from the change of variables $t\rightarrow
uh_{n}+X_{1}$. 
\end{itemize}

\end{proof}

\subsection{Corollary 3.4}

\begin{proof}
[Proof of Corollary 3.4] \begin{enumerate}
\item[(i)] The support of $F_X$ can be represented as $\mathcal{S}_{F_X} = \{ x \in \R^q : F_X(x- r \iota , x + r \iota) > 0 \; \; \forall \; r > 0  \} $. By Lemma 3.3$(i)$ we obtain  \begin{align*}
h_n^{-q}  \E[H_n^2] = \mathbb{E}\bigg( \mathbbm{1} \{ X \in \mathcal{S}_{F_X} \} \: \mu_{2}(X)\int_{\left(  0,1\right]  ^{q}} h_n^{-q}   \Omega
_{2}(X-h_{n}v,X+h_{n}v)  \partial_{v}K^{2}(-v)dv\bigg).
\end{align*}
Denote the maximal function associated to $\mu_{2}$ by \begin{align}
\label{maximal}  (M \mu_2)(X) = \sup_{h > 0 } \frac{ \Omega
_{2}(X-h \iota ,X+h \iota)}{F
_{X}(X-h \iota ,X+h \iota)]} \mathbbm{1} \big \{ X \in \mathcal{S}_{F_X} \big \}.
\end{align}    The integrand of the expectation can be dominated (up to a constant) by \begin{align*}
 &  \mathbbm{1} \{ X \in \mathcal{S}_{F_X} \} \mu_2(X) \big [h_n^{-q}  F
_{X}(X-h_{n} \iota ,X+h_{n} \iota) \big]  \frac{ \Omega
_{2}(X-h_n \iota ,X+h_n \iota)}{F
_{X}(X-h_n \iota ,X+h_n \iota)]}  \\ & \leq \mathbbm{1} \{ X \in \mathcal{S}_{F_X} \} \mu_2(X) \big [h_n^{-q}  F
_{X}(X-h_{n} \iota ,X+h_{n} \iota) \big] \sup_{h > 0}  \frac{ \Omega
_{2}(X-h \iota ,X+h \iota)}{F
_{X}(X-h \iota ,X+h \iota)]}      \\ & \leq  \mu_2(X)  2^q \| f_X \|_{L^{\infty}} (M \mu_2)(X)   \\ & = T(X).
\end{align*}

It is well known (see for e.g. \cite[]{calderon2}) that the maximal operator $g(X) \rightarrow (Mg)(X)  $ is strong type $(2,2)$ bounded so that $ \mu_{2}(X) \in L^2(X) \implies (M \mu_2)(X)   \in L^2(X) $. By Cauchy-Schwarz, $T(X) \in L^1(X)$. Dominated convergence and Lebesgue's differentiation theorem \cite[Theorem 7.10]{rudin2}    yield \begin{align*}
h_{n}^{-q} \mathbb{E}\left[ H_n^2(Z_1,Z_2) \right]  \xrightarrow[n \rightarrow \infty]{} &  \mathbb{E}\bigg(\mu_{2}%
^{2}(X)f_{X}(X)\int_{[0,1]^{q}}\big[\partial_{v}K^{2}(-v)\big]\prod_{i=1}%
^{q}\left(  2v_{i}\right)  dv\bigg)\\
&  =\mathbb{E}\bigg(\mu_{2}^{2}(X)f_{X}(X)\prod_{i=1}^{q}\int_{0}^{1}\left(
-2v_{i}\right)  \partial_{v_i}  k^{2}(v_{i})dv_{i}\bigg)\;\\
&  =\mathbb{E}\bigg(\mu_{2}^{2}(X)f_{X}(X)\prod_{i=1}^{q}\int_{-1}^{1}\left(
-v_{i}\right)  \partial_{v_i}k^{2}(v_{i})dv_{i}\bigg)\;\\
&  =\mathbb{E}\bigg(\mu_{2}^{2}(X)f_{X}(X)\int_{[-1,1]^{q}}K^{2}(v)dv\bigg) \; ,
\end{align*}
where the second-last equality follows noting that $v_{i}%
\rightarrow-v_{i} \, \partial_{v_i} k^{2}(v_{i})$ is an even function and the
last equality follows from univariate integration by parts. 

\item[(ii)] Note  that $\partial_{v} K^2(-v) \geq 0$ for every $v \in [0,1]^q$ so that the integrand defining $\E(H_n^2)$ in Lemma 3.3 is non-negative. The case where $F_{X}$  is absolutely continuous (but not necessarily with a bounded density) follows from Fatou's Lemma and expressing the limit as in part (i).

If $F_{X}$  is not absolutely continuous, then the singular measure   $   (\rho_{d} F_{X}^d +\rho_{s} F_{X}^s )   $ in the Lebesgue decomposition of $F_X$ is non-trivial. Define the probability measure $F_X^{d+s} = [\rho_{d} + \rho_{s}]^{-1} (\rho_{d} F_{X}^d +\rho_{s} F_{X}^s )  $ and let $F_{X}^{d+s}(x-s,x+t)$ be defined the
same way as  in (\ref{Fx}) but with  $d F_X^{d+s}$ replacing $dF_{X}$. By Lemma 3.3 and Assumption 2(iii), we obtain \begin{align*}
  \mathbb{E}\left[ H_n^2(Z_1,Z_2) \right]
&  = \underset{X  \stackrel{}{\sim} F_X^{}}{\E} \bigg(\mu_{2}(X)\int_{\left[  0,1\right]  ^{q}}\Omega
_{2}(X-h_{n}v,X+h_{n}v)\partial_{v}K^{2}(-v)dv\bigg)        \\ &  \gtrapprox \underset{X  \stackrel{}{\sim} F_X^{}}{\E}     \bigg[ \int_{[0,1]^{q}} F_{X}^{}\big(X-h_{n}v,X+h_{n}v\big) \partial_{v}K^{2}(-v)dv\bigg]      \\
& \gtrapprox  \underset{X  \stackrel{}{\sim} F_X^{d+s}}{\E}     \bigg[ \int_{[0,1]^{q}} F_{X}^{d+s}\big(X-h_{n}v,X+h_{n}v\big) \partial_{v}K^{2}(-v)dv\bigg]  .
\end{align*}

Fix any $v \in (0,1)^q$. Then Assumption 3 implies  $\partial_{v} K^2(-v) > 0 $ and a straightforward application of \cite[Theorem 7.15]{rudin2} yields $h_{n}^{-q}  F_{X}^{d+s}\big(X-h_n v    ,X+ h_n v \big)   \uparrow  \infty $ almost everywhere with respect to $ F_X^{d+s}$. As this measure is non-trivial and the integrand is non-negative, Fatou's lemma yields \begin{align*}
&  \liminf_{n\rightarrow\infty}  h_n^{-q}  \underset{X  \stackrel{}{\sim} F_X^{d+s}}{\E}  \bigg[ \int_{[0,1]^{q}} F_{X}^{d+s}\big(X-h_{n}v,X+h_{n}v\big) \partial_{v}K^{2}(-v)dv\bigg]  \\
&  \geq \underset{X  \stackrel{}{\sim} F_X^{d+s}}{\E} \bigg[ \int_{[0,1]^{q}} \liminf_{n \rightarrow \infty}  h_n^{-q}  F_{X}^{d+s}\big(X-h_{n}v,X+h_{n}v\big) \partial_{v}K^{2}(-v)dv\bigg] \\
&  = \infty.
\end{align*}
\end{enumerate}
\end{proof}

\subsection{Lemma 3.5}

\begin{proof}
[Proof of Lemma 3.5] 

\begin{enumerate}

\item[(i)]
The lower bound was provided in the main text. Here, we bound the moment from above.
Note that $\partial_{v} K^2(-v)\geq 0 $ for every $v \in [0,1]^q$.  From the expression defining $\E[H_n^2]$ in Lemma 3.3 we obtain that
\begin{align*}
\mathbb{E}\left[ H_n^2(Z_1,Z_2) \right]  &
=\mathbb{E}\bigg(\mu_{2}(X)\int_{[0,1]^{q}}\Omega_{2}(X-h_{n}v,X+h_{n}%
v)  \partial_{v}K^2(-v) dv\bigg)\\
&  \leq M_{2}\mathbb{E}\big[\mu_{2}(X) \Omega_{2}(X- h_n \iota,X+h_n \iota )    \big]\\
&  \leq B^{2}M_{2} \E \big[  F_X \big(X -   h_n \iota , X + h_n \iota     \big)       \big]
\end{align*}
where $B$ is as in Assumption 2 and $M_{2}=\int_{[0,1]^{q}}  \partial_{v}K^2(-v) dv $.  \\

\item[(ii)] The bound follows from the same argument as in part (i) by replacing $(\mu_2,\Omega_2)$ with $(\mu_4,\Omega_4)$. \\

\item[(iii)] By Lemma 3.3(iii), we have that \begin{align*}
\mathbb{E}[G_{n}^{2}(Z_{1},Z_{2})]  &
\leq C \, \mathbb{E}\bigg(  \mu_2(X)  \Omega_2^2(X - h_n \iota , X +  h_n \iota)     \int_{[-2,2]^q} \Omega_{2} ( X - 2 h_n \iota , X + h_n u     )     \\
&  \;\;\;\;\; \; \; \; \; \; \;  \times \left|   \partial_{u} \bigg[ \int_{  [-1,1]^q
}   \big(    \partial_{v} 
 \big[ K(v)K\big(v-u\big) \big]   \big)^2 \:  \mathbbm{1} \big \{  v-u \in [-1,1]^q      \big \}     dv \bigg]   du\bigg) \right|.
\end{align*}
for some universal constant $C > 0 $. Since the outer integral is over $u \in [-2,2]^q$, we further obtain that \begin{align*} \E[G_n^2] & \leq D \E \big[ \{ \Omega_{2}(X - h_n \iota , X + h_n \iota) \}^2 \Omega_2(X - 2 h_n  \iota , X + 2 h_n \iota)            \big]   \\ & \leq D \E \big[ \{ \Omega_2(X - 2 h_n  \iota , X + 2 h_n \iota)     \}^3       \big]  
\end{align*} 
for some universal constant $D > 0 $. The claim follows from using Assumption 2(ii) to bound the expression on the right.

\end{enumerate}

\end{proof}

\subsection{Lemma 3.7}

\begin{proof}
[Proof of Lemma 3.7] 
\begin{enumerate}
\item[(i)] Fix any $F_X \in \mathcal{D}(s)$.  From the definition of $\mathcal{D}(s)$ and an application of Fatou's Lemma, we obtain that as $h \downarrow 0$ \begin{align*}
  \E[ F_X(X - h \iota , X + h \iota) ]  \asymp   h^{sq}   \; \; \; ,  \; \; \;  \mathbb{E} \big[ \big\{ F_{X}(X-h_{}\iota,X+h_{}\iota) \big \}^3  \big]  \lessapprox h^{3sq}.
\end{align*}
Assumption 4 follows immediately.
\item[(ii)] 
 The second condition of Definition 3.6 follows from \begin{align*}
 \E\big[  \underline{F_{X}}(X,s_{}) \big]  &  =   \underset{X  \stackrel{}{\sim} F_{X} }{\E}  \bigg(  \liminf_{h \downarrow 0}   \sum_{t \in T}     \alpha_{t} \frac{F_{t}(X - h \iota , X + h \iota)}{(2h)^{s_{}q}}   \bigg)    \\  & \geq  \underset{X  \stackrel{}{\sim} F_{X} }{\E}  \bigg(  \liminf_{h \downarrow 0}   \sum_{t \in T  : s_t  = s}     \alpha_{t} \frac{F_{t}(X - h \iota , X + h \iota)}{(2h)^{s_{}q}}   \bigg)   \\ &  \geq  \underset{X  \stackrel{}{\sim} F_{X} }{\E}  \bigg(    \sum_{t \in T  : s_t  = s}     \alpha_{t}  \liminf_{h \downarrow  0}  \frac{F_{t}(X - h \iota , X + h \iota)}{(2h)^{s_{}q}}   \bigg)         \\ &  =  \underset{X  \stackrel{}{\sim} F_{X} }{\E}  \bigg(    \sum_{t \in T  : s_t  = s}     \alpha_{t} \, \underline{F_t}(X,s)     \bigg)                \\ &  \geq       \sum_{t \in T  : s_t  = s}     \alpha_{t}^2  \underset{X  \stackrel{}{\sim} F_{t} }{\E}    \big[ \, \underline{F_t}(X,s)      \big]                \\ & > 0.
\end{align*}
From Definition 3.6, there exist constants  $(M_t)_{t \in T}$ such that \begin{equation}
\label{fracmix}  \underset{X  \stackrel{}{\sim} F_{t} }{\mathbb{P}} \bigg(   \frac{F_{t}(X - h \iota , X + h \iota)}{(2h)^{s_{t}q}}  \leq M_t     \bigg)   = 1 \; \; \; \; \; \forall \;  t \in T.
\end{equation}
We claim that the mixture measure $ F_{X} =  \sum_{t \in T} \alpha_t F_t $ satisfies  the first condition of Definition 3.6 with    $ M^* =  4^q \sum_{t \in T} \alpha_t M_t  $. 
Let $S_t \subseteq \R^q $  denote the set where (\ref{fracmix}) holds under the measure $\underset{X  \stackrel{}{\sim} F_{t} }{\mathbb{P}}$. In particular, we have that $\underset{X  \stackrel{}{\sim} F_{X} }{\mathbb{P}} \big( X \in  \bigcup_{t  \in T} S_t     \big ) = 1$. Therefore, it suffices to verify  $$(2h)^{-s_{} q  } F_X(x -h \iota , x + h \iota) \leq M^* \; \; \; \; \forall \; x \in  \bigcup_{t \in T} S_t  . $$ 
Fix any $t \in T$. If $x \in S_{t}  $, we have that \begin{align*}
(2h)^{-s_{} q  } F_X(x- h \iota , x + h \iota)  & = (2h)^{-s_{} q  }  \bigg [ \alpha_{t} F_{t}(x- h \iota , x + h \iota)    +  \sum_{i \in T : i \neq t}  \alpha_{i} F_{i}(x- h \iota , x + h \iota)  \bigg] \\ & \leq \alpha_{t} M_t + \sum_{i \in T : i \neq t} \alpha_{i} (2h)^{-s_{} q  }  F_{i}(x- h \iota , x + h \iota) ].
\end{align*}
Fix any $i \in T$ such that $i \neq t$. Then either the set  $ A^h(x,S_i) =  \{ y_{x} \in S_i : \|  x - y_x \|_{\infty} \leq h  \}$ is non-empty or $F_{i}(x-h \iota , x +h \iota) = 0$ (because $\underset{X  \stackrel{}{\sim} F_{i} }{\mathbb{P}}(S_i) = 1$). In the former case, we have for any fixed $y_x \in A^h(x,S_i)$ the inclusion $\{ t \in \R^q : \| t-x \|_{\infty} \leq h  \}  \subseteq \{ t \in \R^q : \| t-y_{x} \|_{\infty} \leq 2h  \}   $ and the bound \begin{align*} (2h)^{- s_{} q}
F_{i}(x- h \iota , x + h \iota) & \leq (2h)^{- s_{} q}F_{i}(y_{x}- 2h \iota , y_{x} + 2h \iota) \\ & \leq   (2h)^{- s_{} q} (4h)^{s_{t_i} q}  M_i \\ & \leq 2^{(2 s_{i} - s_{}) q } M_i  \\ & \leq 4^q M_i.
\end{align*}
Since this holds for every $ i \neq t$, we obtain that \begin{align*}
(2h)^{-s_{} q  } F_X(x- h \iota , x + h \iota)  \leq  \alpha_t M_t +  4^q \sum_{ i \in T : i \neq t} \alpha_i  M_i \; \; \; \; \; \forall  \; x \in S_t. 
\end{align*}
Since this holds for every $t \in T$, we obtain that \begin{align*}
(2h)^{-s_{} q  } F_X(x- h \iota , x + h \iota)  \leq   4^q \sum_{ i \in T} \alpha_i  M_i \; \; \; \; \; \forall  \; x \in \bigcup_{t \in T} S_t. 
\end{align*}
\end{enumerate}
\end{proof}

\subsection{Theorem 3.9}

\begin{proof}
[Proof of Theorem 3.9]

The null hypothesis is a special case of the alternative  $H_1$ with the choices $ \gamma_n = 0 $ and $\delta(X) = 1 $. Hence, by an application of Lemma \ref{varest}, we have that $ \hat{\sigma}_n^2 = 2 \E(H_n^2)[1+ o_{\mathbb{P}}(1)]$. Therefore, it suffices to verify  \begin{equation}
\label{Inth} \frac{n \hat{I}_n}{\sqrt{2 \E(H_n^2)}} = N(0,1) + o_{\mathbb{P}}(1).
\end{equation}
 The mean value theorem implies that there exists a $\beta_*$ on the line segment connecting $\hat{\beta}$ and $\beta_0$ such that \begin{align}
 \hat{u}_i - u_i =  g(X_{i},\beta_0)- g(X_{i},\hat{\beta})  = [\nabla_{\beta} g_{}(X_{i},\beta_*)]'( \beta_0 - \hat{\beta} ) .  \label{mv2}
\end{align} Since $\| \hat{\beta} - \beta_0 \|_{2} = O_\mathbb{P}(n^{-1/2})$, it suffices to work under the setting where $\hat{\beta}$ and $\beta_*$ lie in the fixed neighborhood $\mathcal{N}$ of Assumption 5.

The statistic can be expressed as
\begin{align*}
\hat{I}_n   = U_n  & +  \frac{1}{n(n-1)} \sum_{i=1}^n \sum_{j \neq i} K \bigg( \frac{X_i - X_j}{h_n}   \bigg) ( \hat{u}_i - u_i  )  ( \hat{u}_j - u_j  )  \\ &   + \frac{2}{n(n-1)} \sum_{i=1}^n \sum_{j \neq i} K \bigg( \frac{X_i - X_j}{h_n}   \bigg) ( \hat{u}_i - u_i )u_{j} \\ & = U_n + A_1 + A_2.
\end{align*}
The result follows from Theorem 3.8 if $A_{i} = o_{\mathbb{P}}( n^{-1} \sqrt{ \E(H_n^2) } )$ for $i=1,2$. 

\begin{enumerate}
\item[(i)] We verify that $A_1 =  o_{\mathbb{P}}( n^{-1} \sqrt{ \E(H_n^2) } )$. From (\ref{mv2}) we obtain \begin{align*}
& A_1  =  (\hat{\beta} - \beta_0)'  B_1 (\hat{\beta} - \beta_0) \\ & \text{where} \; \; \;  B_1 =  \frac{1}{n(n-1) } \sum_{i=1}^n \sum_{j \neq i} K \bigg( \frac{X_i - X_j}{h_n}      \bigg)  [\nabla_{\beta} g_{}(X_{i},\beta_*)] [\nabla_{\beta} g_{}(X_{j},\beta_*)]' 
\end{align*}
 Observe that $  A_1 \leq \|  \hat{\beta} - \beta_0  \|_{2}^2 \| B_1 \|_{op}$ and
\begin{align*}
\| B_1 \|_{op} & =  \bigg \| \frac{1}{n(n-1) } \sum_{i=1}^n \sum_{j \neq i} K \bigg( \frac{X_i - X_j}{h_n}      \bigg)  [\nabla_{\beta} g_{}(X_{i},\beta_*)] [\nabla_{\beta} g_{}(X_{j},\beta_*)]'           \bigg \|_{op} \\ & \leq  \frac{1}{n (n-1)} \sum_{i=1}^n \sum_{j \neq i} K \bigg( \frac{X_i - X_j}{h_n}      \bigg) M(X_i) M(X_j).
\end{align*}
Since $\| \hat{\beta} - \beta_0 \|_{2} = O_{\mathbb{P}}(n^{-1/2})$, the desired result follows if we can verify that \begin{equation*} \E \bigg[  K \bigg( \frac{X_1 - X_2}{h_n}      \bigg) M(X_1) M(X_2) \bigg] = o \big( [\E(H_n^2)]^{1/2}      \big). \end{equation*}
From Lemma \ref{aux-new} and \ref{aux-estimates} we obtain

\begin{align*}
 \E \bigg[  M(X_1) K \bigg(  \frac{X_1 - X_2}{h_n} \bigg) M(X_2)                    \bigg]
& \lessapprox  \E \big[  M(X) \Omega_{M} (X - h_n \iota , X + h_n \iota )              \big]                    \\ & = o \big( \{ \E \big[  F_X (X - h_n \iota , X + h_n \iota )              \big]   \}^{1/2}      \big) \; ,
\end{align*}
By Lemma 3.5 and Assumption 4(i), we have $\E[H_n^2] \asymp  \E \big[  F_X (X - h_n \iota , X + h_n \iota )              \big] $ and the claim follows. \\

\item[(ii)] We verify that $A_2 =  o_{\mathbb{P}}( n^{-1} \sqrt{ \E(H_n^2) } )$. A second order Taylor expansion of $g(X,\beta_0)$ around $g(X,\hat{\beta})$ yields
\begin{align*}
A_{2} &  =(\hat{\beta}-\beta_{0})^{\prime}\bigg[\frac{2}{n(n-1)}\sum
_{i=1}^n \sum_{j\neq i}K\bigg(\frac{X_{i}-X_{j}}{h_{n}}\bigg)\nabla_{\beta} g_{}%
(X_{i},\beta_{0})u_{j}\bigg]\\
&  +(\hat{\beta}-\beta_{0})^{\prime}\bigg[\frac{1}{n(n-1)}\sum_{i=1}^n %
\sum_{j\neq i}K\bigg(\frac{X_{i}-X_{j}}{h_{n}}\bigg)\nabla_{\beta}^{2}g_{}%
(X_{i},\beta_{\ast})u_{j}\bigg](\hat{\beta}-\beta_{0})\\
&  =(\hat{\beta}-\beta_{0})^{\prime} B_2 +(\hat{\beta}-\beta
_{0})^{\prime} C_2 (\hat{\beta}-\beta_{0}) 
\end{align*}
where $\beta_*$ is on the line segment connecting $\hat{\beta}$ and $\beta_0$. Since $\| \hat{\beta} - \beta_0  \|_{2} = O_{\mathbb{P}}(n^{-1/2})$, it suffices to verify that \begin{equation} \label{b2c2}  \| B_2 \|_{2}^2 = o_{\mathbb{P}} \big(n^{-1}  \E(H_n^2) \big  )  \; \; , \; \;  \| C_2 \|_{op}  = o_{\mathbb{P}} \big( \sqrt{\E(H_n^2)}  \big)   .
\end{equation}
For $C_2$,  the triangle inequality and Assumptions (2, 5) yield $$ \E ( \| C_2 \|_{op} ) \lessapprox   \E \bigg[  K \bigg(  \frac{X_1 - X_2}{h_n}   \bigg) G(X_1)    \bigg] . $$
By Lemma \ref{aux-new}, we obtain \begin{align*}
 \E \bigg[  K \bigg(  \frac{X_1 - X_2}{h_n}   \bigg) G(X_1)    \bigg]   \lessapprox  \E \big[  G(X) F_X (X - h_n \iota , X + h_n \iota )              \big]             .    
\end{align*}
From H\"{o}lder's inequality  and Assumption 4(ii), it follows that \begin{align*}
\E \big[  G(X) F_X (X - h_n \iota , X + h_n \iota )              \big]   &  \leq  \{ \E[ G^{4/3}] \}^{3/4}  \{ \E \big[ \big( F_X (X - h_n \iota , X + h_n \iota )   \big)^4           \big]  \}^{1/4} \\ & \leq \{ \E[ G^{4/3}] \}^{3/4}  \{ \E \big[ \big( F_X (X - h_n \iota , X + h_n \iota )   \big)^3           \big]  \}^{1/4} \\ & \leq  \{ \E[ G^{4/3}] \}^{3/4}  \{ \E \big[  F_X (X - h_n \iota , X + h_n \iota )              \big]  \}^{1/2} \zeta_{n}^{1/4} 
\end{align*}
where  \begin{equation}   \zeta_{n} =  \frac{\mathbb{E} \left[ \big[ F_{X}(X-h_n\iota,X+h_n\iota)
 \big]^3  \right] }{\left(  \mathbb{E}\left[  F_{X}(X-h_n\iota,X+h_n
\iota)\right]  \right)  ^{2}}     \downarrow 0.    \label{zeta} \end{equation}The bound for $\| C_2  \|_{op} $ in (\ref{b2c2}) now follows from substituting $\E[H_n^2] \asymp  \E \big[  F_X (X - h_n \iota , X + h_n \iota )              \big] $ . 
 
For $B_2$, let $\nabla_{\beta} g_{t}(X_{i})$ and
$B_{2,t}$ denote respectively, the $t^{th}$ coordinate of the vectors $\nabla_{\beta}
g_{}(X_{i},\beta_{0})$ and  $B_2$.  The bound for $B_2$ in (\ref{b2c2}) follows if we can show that $\E(B_{2,t}^2) = o_{\mathbb{P}}(n^{-1} \E(H_n^2))$ for every coordinate $t$. From the definition of $B_2$, we obtain \\
 \begin{align*} & \E (  B_{2,t}^2   ) \\ &  = \frac{4}{n^2(n-1)^2}  \E  \bigg( \sum_{i=1}^n \sum_{k=1}^n  \sum_{ \substack{j \neq i \\ j \neq k } }  K \bigg( \frac{X_i - X_j}{h_n}  \bigg) K \bigg( \frac{X_k - X_j}{h_n}  \bigg) \nabla_{\beta} g_{t}(X_{i}) \nabla_{\beta} g_{t}(X_{k}) u_j^2  \bigg)  \\     & \lessapprox n^{-1}  \mathbb{E}\bigg[K\bigg(\frac{X_{1}-X_{2}}{h_{n}}\bigg)K\bigg(\frac
{X_{3}-X_{2}}{h_{n}}\bigg) M(X_1) M(X_3)  \bigg] \\ & \; \; \; + n^{-2} \:  \mathbb{E}\bigg[K^{2}\bigg(\frac
{X_{1}-X_{2}}{h_n}\bigg) M^2(X_1)  \bigg] \\ & = n^{-1} T_1 + n^{-2} T_2   .  \end{align*}
For the second term, since $K(.)$ is bounded and $M \in L^2(X)$,  we have that $T_2 = O(1)$. By Corollary 3.4, $\E(H_n^2) \gtrapprox h_n^q$ and we obtain $$  \frac{ n^{-2} T_2}{n^{-1}   \E(H_n^2) } \lessapprox \frac{1}{n h_n^q}  = o(1).  $$
For the first term,  since $X_1,X_2,X_3$ are i.i.d,  we obtain that   \begin{align}
T_1 & =    \mathbb{E}\bigg[K\bigg(\frac{X_{1}-X_{2}}{h_{n}}\bigg)K\bigg(\frac
{X_{3}-X_{2}}{h_{n}}\bigg) M(X_1) M(X_3)  \bigg] \nonumber \\ &  = \mathbb{E}\bigg[\bigg(\int_{\mathbb{R}^{q}}K\bigg(\frac{t-X_{2}%
}{h_{n}}\bigg) M(t)  d F_{X}(t)\bigg)\bigg(\int_{\mathbb{R}^{q}}K\bigg(\frac{t-X_{2}}{h_{n}}%
\bigg) M(t)  d F_{X}%
(t)\bigg)\bigg]  \nonumber \\ & = \mathbb{E}\bigg[\bigg(\int_{\mathbb{R}^{q}}K\bigg(\frac{t-X_{2}%
}{h_{n}}\bigg) M(t)  d F_{X}(t)\bigg)^2 \bigg]. \label{boundsq0}
\end{align}
 
Define
\[
f(x)= M(x),\;g(x)=K\bigg(\frac{x-X_{2}}{h_{n}}\bigg).
\]
Let $X_2^i$ denote the $i^{th}$ coordinate of $X_2$. Conditional on $X_{2}$, $(f,g)$ satisfy the hypothesis of Lemma 3.2 with  $ \mathcal{O}  = (X_{2}^{1}%
-h_{n},X_{2}^{1}+h_{n}) \times \dots \times (X_{2}^{q}%
-h_{n},X_{2}^{q}+h_{n})  $.  Applying Lemma 3.2 yields
\begin{align}
 T_1 & = \mathbb{E}\bigg[\bigg(\int_{\mathbb{R}^{q}}K\bigg(\frac{x-X_{2}}{h_{n}%
}\bigg) M(x)  dF_{X}(x)\bigg)^{2}%
\bigg]  \nonumber
\\ &  = \E \bigg[  \bigg( \int_{\mathbb{R}^{q}} \Omega_{M}(X_{2}-h_{n}\iota,t)\partial_{t}%
K^{}\bigg(\frac{t-X_{2}}{h_{n}}\bigg)dt \bigg)^2 \bigg] \nonumber \\
&  = \E \bigg[  \bigg( \int_{\left[  -1,1\right]  ^{q}}  \Omega_{M} (X_{}-h_{n}\iota,X_{}%
+h_{n}v)\partial_{v}K^{}(v)dv \bigg)^2 \bigg] \nonumber \\ &  \lessapprox     \E \big[  \Omega_{M}^2 (X - h_n \iota , X + h_n \iota )              \big]              \nonumber       \\ & = o \big(   \E \big[  F_X (X - h_n \iota , X + h_n \iota )              \big]     \big)           \label{boundsq} \; ,
\end{align}
 where the third equality  follows from the change of variables $t\rightarrow
X_{2}+h_{}v$ and the support of $k(\,\cdot\,)$. The last equality follows from Lemma \ref{aux-estimates}. From substituting $\E[H_n^2] \asymp  \E \big[  F_X (X - h_n \iota , X + h_n \iota )              \big] $, we obtain that $ n^{-1} T_1 = o(  n^{-1} \E[H_n^2]   )   $.
\end{enumerate}
\end{proof}

\subsection{Theorem 3.10}

\begin{proof}
[Proof of Theorem 3.10] 

 By Lemma \ref{varest}, we have that $ \hat{\sigma}_n^2 / 2 = \E(H_n^2)[1+ o_{\mathbb{P}}(1)]$. Therefore, it suffices to study the limiting behavior of  $ n \hat{I}_n / \sqrt{2 \E(H_n^2)} $. The numerator of the statistic can be expressed as 
 
 \begin{align*}
 \hat{I}_n =  A_1 + A_2 + A_3 + A_4 + A_5 + A_6
 \end{align*}
 where \begin{align*}
 & A_1 = \frac{1}{n(n-1)} \sum_{i=1}^{n}\sum_{j\neq i}K\bigg(\frac{X_{i}-X_{j}%
}{h_{n}}\bigg) [g_{}(X_{i},\beta_0)- g(X_{i},\hat{\beta})][g_{}(X_{j},\beta_0)- 
{g}(X_{j},\hat{\beta})] \; , \\ & A_2 =  \frac{\gamma_n^2}{n(n-1)} \sum_{i=1}^{n}\sum_{j\neq i}K\bigg(\frac{X_{i}-X_{j}%
}{h_{n}}\bigg) \delta(X_i) \delta(X_j)\; , \\ & A_3 =  \frac{1}{n(n-1)} \sum_{i=1}^{n}\sum_{j\neq i}K\bigg(\frac{X_{i}-X_{j}%
}{h_{n}}\bigg) u_i u_j \; , \\ & A_4 = \frac{2 \gamma_n}{n(n-1)} \sum_{i=1}^{n}\sum_{j\neq i}K\bigg(\frac{X_{i}-X_{j}
}{h_{n}}\bigg) [g_{}%
(X_{i},\beta_0)- {g}(X_{i},\hat{\beta})] \delta_{}(X_{j}) \; , \\ & A_5 = \frac{2 }{n(n-1)}\sum_{i=1}^{n}\sum_{j\neq i}K\bigg(\frac{X_{i}-X_{j}%
}{h_{n}}\bigg) [g_{}(X_{i},\beta_0)- {g}(X_{i},\hat{\beta})]u_{j} \; , \\ & A_6 = \frac{2 \gamma_n}{n(n-1)}\sum_{i=1}^{n}\sum_{j\neq i}K\bigg(\frac{X_{i}-X_{j}%
}{h_{n}}\bigg) \delta(X_i) u_j .
 \end{align*}
From the proof of Theorem  3.8 and 3.9, we have that \begin{align} A_{i} = o_{\mathbb{P}}( n^{-1} \sqrt{ \E(H_n^2) } ) \; \; \; \; \; \; \; \; \;  i=1,5 \; \; \; \; \; \; \; \; , \; \; \; \;  \frac{ A_3}{\sqrt{2 \E(H_n^2)}}  = N(0,1) + o_{\mathbb{P}}(1)     .
\end{align}

It remains to derive the asymptotics for  $A_2 , A_4 $ and $ A_6$. The mean value theorem implies that there exists a $\beta_*$ on the line segment connecting $\hat{\beta}$ and $\beta_0$ such that \begin{align}
  g(X_{i},\beta_0)- g(X_{i},\hat{\beta})= [\nabla_{\beta} g_{}(X_{i},\beta_*)]'( \beta_0 - \hat{\beta} ).  \label{mv22}
\end{align}
Since $\| \hat{\beta} - \beta_0 \|_{2} = O_\mathbb{P}(n^{-1/2})$, it suffices to work under the setting where $\hat{\beta}$ and $\beta_*$ lie in the neighborhood $\mathcal{N}$ of Assumption 5.

\begin{enumerate}
\item[(a)]
We verify that $A_4 = o_{\mathbb{P}} \big( n^{-1} \sqrt{\E(H_n^2)}        \big) $ under both parts (i) and (ii) of Theorem 3.10. The mean value expansion (\ref{mv22}) yields \begin{align*}  A_4   =  (  \hat{\beta} - \beta_0  )' B_4.     \end{align*}
where \begin{align*}
\E \| B_4 \|_{2} & =    \gamma_n \E \bigg \| \frac{1}{n(n-1) } \sum_{i=1}^n \sum_{j \neq i} K \bigg( \frac{X_i - X_j}{h_n}      \bigg)  [\nabla_{\beta} g_{0}(X_{i},\beta_*)]          \delta(X_j) \bigg \|_{2} \\ &  \leq  \gamma_n \E \bigg[ K \bigg(  \frac{X_1 - X_2}{h_n}   \bigg) M(X_1) \left| \delta(X_2) \right|       \bigg]\\ & \lessapprox \gamma_n \E \bigg[ K \bigg(  \frac{X_2 - X_1}{h_n}   \bigg) M(X_1)       \bigg].
\end{align*}
By Lemma \ref{aux-new} and \ref{aux-estimates} we obtain \begin{align*}
 \E \bigg[ K \bigg(  \frac{X_2 - X_1}{h_n} \bigg) M(X_1)                    \bigg] & \lessapprox      \E \big[  M(X) F_X(X - h_n \iota , X + h_n \iota )              \big]                    \\ & = o \big( \{ \E[F_X(X - h_n \iota , X + h_n \iota)]    \}^{3/4}  \big) .
\end{align*}
By Markov's inequality, it follows that \begin{align*}
A_4 = \gamma_n n^{-1/2} o_{\mathbb{P}} \big(     \{ \E[F_X(X - h_n \iota , X + h_n \iota)]    \}^{3/4}           \big) . 
\end{align*}
By Lemma 3.5 and Assumption 4(i), we have $\E[H_n^2] \asymp  \E \big[  F_X (X - h_n \iota , X + h_n \iota )              \big] $  and it follows that \begin{align*}
\frac{n A_4}{ \sqrt{\E(H_n^2)}} =  o_{\mathbb{P}} \big( \sqrt{n} \gamma_n  \{ \E[F_X(X - h_n \iota , X + h_n \iota)]    \}^{1/4}  \big)
\end{align*}
This is  $o_{\mathbb{P}}(1)$ because  $\gamma_n \lessapprox n^{-1/2} \{ \E[ F_{X}(X - h_n \iota , X + h_n \iota) ]   \}^{-1/4}$. \\

\item[(b)] We verify that $A_6 = o_{\mathbb{P}} \big( n^{-1} \sqrt{\E(H_n^2)}        \big) $ under both parts (i) and (ii) of Theorem 3.10. From Assumption 2 we obtain  \begin{align*}  \E(|A_6|^2)  & =  \frac{\gamma_n^2}{n^2(n-1)^2 } \sum_{i=1}^n \sum_{k=1}^n \sum_{ \substack{j \neq i \\ j \neq k } }   \E \bigg[  K \bigg(  \frac{X_i-X_j}{h_n}   \bigg) K \bigg(  \frac{X_{k}-X_j}{h_n}   \bigg) \delta(X_{i}) \delta(X_{k}) u_j^2        \bigg]   \\ & =  \frac{\gamma_n^2  }{n^2(n-1)^2 } O(n^2)   \E \bigg[ K^2 \bigg(  \frac{X_1-X_2}{h_n}    \bigg) \delta^2(X_1)      \bigg]           \\ & + \frac{\gamma_n^2 }{n^2(n-1)^2 } O(n^3)  \E \bigg[   K \bigg(  \frac{X_1-X_2}{h_n}   \bigg) K \bigg(   \frac{X_3-X_2}{h_n}  \bigg) \delta(X_1) \delta(X_3)          \bigg]   \\ & =    \gamma_n^2  O(n^{-2}) T_1 + \gamma_n^2 O(n^{-1}) T_2 .
\end{align*}
From $\| \delta \|_{\infty} < \infty$  and  Lemma \ref{aux-new} we obtain \begin{align*}
T_1 \lessapprox  \E \bigg[   K^2 \bigg(  \frac{X_1-X_2}{h_n}    \bigg)      \bigg]  \lessapprox  \E \big[  F_{X}(X - h_n \iota , X + h_n \iota)               \big] .
\end{align*}
By an analogous argument to (\ref{boundsq}) we obtain $$ T_2 \lessapprox  \E \big [ \Omega_{\left| \delta \right|}^2 (X - h_n \iota , X + h_n \iota )     \big ] \lessapprox \E \big[   \big \{  F_{X}(X -  h_n \iota , X +  h_n \iota)      \big \}^2 \big].  $$
By Markov's inequality and $\gamma_n \lessapprox n^{-1/2} \{ \E[ F_{X}(X - h_n \iota , X + h_n \iota) ]   \}^{-1/4}$ we obtain \begin{align*}
(A_6)^2 = O_{\mathbb{P}} \bigg(  \frac{ \{ \E[ F_{X}(X - h_n \iota , X + h_n \iota) ]   \}^{1/2}}{n^3}  + \frac{\E \big[   \big(  F_{X}(X - h_n \iota , X +  h_n \iota)      \big)^2              \big]}{n^2  \{ \E[ F_{X}(X - h_n \iota , X + h_n \iota) ]   \}^{1/2}}                 \bigg).
\end{align*}
From substituting $\E[H_n^2] \asymp  \E \big[  F_X (X - h_n \iota , X + h_n \iota )              \big] $ it follows that \begin{align*}
\frac{n^2 (A_6)^2}{\E(H_n^2)}  = O_{\mathbb{P}} \bigg(  \frac{1}{n \sqrt{\E(H_n^2)}  } +   \frac{\E \big[  \big \{ F_X(X - h_n \iota , X +  h_n \iota ) \big \}^2 \big]}{\{  \E[  F_X(X - h_n \iota , X + h_n \iota)  ]  \}^{3/2} }          \bigg).
\end{align*}
By Corollary 3.4(ii), $\E(H_n^2) \gtrapprox h_n^q$. As $n h_n^q \uparrow \infty$  the first term on the right is $o(1)$. The second term is $o(1)$ by Assumption 6. \\

\item[(c)] We aim to show that \begin{align}
\E(A_2)   =     \gamma_n^2 \bigg \{ &  \E \bigg[  \delta^2(X) \int_{[0,1]^q} F_X(X - h_n v , X + h_n v) \partial_{v} K(-v) dv               \bigg] \nonumber \\ & + o(1)  \E \big[  F_{X}(X - h_n \iota , X + h_n \iota)  \big]         \bigg \}.
\label{biasa2}
\end{align}

\begin{align}
 \text{Var}(A_2) \lessapprox   \gamma_n^4 \bigg \{ & n^{-2}    \E \big[  F_{X}(X- h_n \iota , X + h_n \iota)            \big] \nonumber \\ & +   n^{-1} \E \big[ \big \{ F_X(X -  h_n \iota , X +  h_n \iota)  \big \}^2 \big] \bigg \} .
\label{vara2}
\end{align}
An analogous argument to the one used for Lemma 3.3 yields \begin{align*}
\E(A_2) & =  \gamma_n^2 \E \bigg[ K \bigg( \frac{X_1 - X_2}{h_n} \bigg) \delta(X_1) \delta(X_2)             \bigg] \\ & =  \gamma_n^2 \E \bigg(   \delta(X) \int 
 _{[0,1]^q} \Omega_{\delta} (X - h_nv , X + h_nv) \partial_{v} K(-v) dv          \bigg) \\ &  =  \gamma_n^2  R_2.
\end{align*}

The support of $F_{X}$ can be expressed as
$$ \mathcal{S}_{F_X} =  \{x \in \R^q :  F_{X}(x-r \iota , x +r \iota) > 0 \; \; \text{for every} \; r > 0     \} . $$

Define $$ Z_{n,v}(x) =  \begin{cases}   \frac{\Omega_{\delta}(x - h_n v , x + h_n v)}{F_X(x - h_n v , x + h_n v)} & F_X(x - h_n v , x + h_n v) > 0 \; , \\ 0 & \text{else}      . \end{cases}  $$ It follows that \begin{align*}
  & \left| R_2 - \E \bigg[  \delta^2(X) \int_{[0,1]^q} F_X(X - h_n v , X + h_n v) \partial_{v} K(-v) dv               \bigg]  \right| \\ & \leq   \E \bigg( \left| \delta(X)        \right|   \int_{[0,1]^q} \left| Z_{n,v}(X)  - \delta(X)          \right| F_X(X - h_n v , X + h_n v)   \partial_{v} K(-v) dv   \bigg) \\ & \lessapprox \E \bigg(   \int_{[0,1]^q} \left| Z_{n,v}(X)  - \delta(X)          \right| F_X(X - h_n v , X + h_n v)   \partial_{v} K(-v) dv   \bigg) \\ & =  \E \bigg(  \mathbbm{1} \big \{ X \in \mathcal{S}_{F_X}   \big \} \int_{[0,1]^q} \left| Z_{n,v}(X)  - \delta(X)          \right| F_X(X - h_n v , X + h_n v)   \partial_{v} K(-v) dv   \bigg) .
\end{align*}
For every $x \in \mathcal{S}_{F_X}$ we have that \begin{align*} 
\left|  Z_{n,v}(x)  - \delta(x)     \right|  & \leq   \frac{1}{F_X(x - h_n v , x+ h_n v)} \int \limits_{t \in \mathcal{S}_{F_X} : \left| t_i - x_i \right| \leq h_n v_i \; \forall i } \left|  \delta(t) - \delta(x)   \right| d F_X(t) \\ & \leq  \sup_{s,t \in \mathcal{S}_{F_X} : \| s-t \|_{\infty} \leq h_n} \left| \delta(t) - \delta(s)   \right| 
\end{align*}
uniformly over $v \in (0,1]^q$. Hence \begin{align*}
&  \left| R_2 - \E \bigg[  \delta^2(X) \int_{[0,1]^q} F_X(X - h_n v , X + h_n v) \partial_{v} K(-v) dv               \bigg]  \right|  \\ & \lessapprox \sup_{s,t \in \mathcal{S}_{F_X} : \| s-t \|_{\infty} \leq h_n} \left| \delta(t) - \delta(s)   \right|  \E \bigg(  \int_{[0,1]^q}       F_X(X - h_n v , X + h_n v)   \partial_{v} K(-v) dv               \bigg) \\ & \lessapprox \sup_{s,t \in \mathcal{S}_{F_X} : \| s-t \|_{\infty} \leq h_n} \left| \delta(t) - \delta(s)   \right| \E \big[  F_{X}(X - h_n \iota , X + h_n \iota)  \big]    . 
\end{align*}

Since $\delta(.)$ is uniformly continuous on $\mathcal{S}_{F_X}$ and $h_n \downarrow 0$, the preceding bound implies (\ref{biasa2}).

 It remains to show (\ref{vara2}). The variance is given by  \begin{align*}
\text{Var}(A_2) = \frac{\gamma_n^4}{n^2(n-1)^2} \sum_{i=1}^n \sum_{k=1}^n \sum_{j \neq i} \sum_{l \neq k} \text{Cov}  \bigg[ &  K \bigg(   \frac{X_i - X_j}{h_n}    \bigg) \delta(X_i) \delta(X_j), \\ &  K \bigg( \frac{X_{k} - X_l}{h_n}    \bigg) \delta(X_k) \delta(X_l)          \bigg].
\end{align*}
The covariance is $0$ if the indices $(i,j,k,l)$ are all distinct. There are $2n(n-1)(2n-3)$ terms for which the indices are not all distinct. Within this, $2n(n-1)$ correspond to the case where $(i,j) = (j,l)$ or $(i,j) = (l,j)$ and $2n(n-1)(2n-3) - 2n(n-1) = 4(n^3  - 3 n^2  + 2n)$ for the case where $(l,j) \neq (i,j) \neq (j,l)$. It follows that \begin{align*}
\text{Var}(A_2) = \frac{\gamma_n^4}{n^2(n-1)^2} \bigg[  O(n^2) T_1 + O(n^3) T_2          \bigg].
\end{align*}
where
\begin{align*} &  T_{1} =  \text{Var} \bigg[ K \bigg(  \frac{X_1 - X_2}{h_n}      \bigg) \delta(X_1) \delta(X_2)  \bigg] \; , \\ & T_2 =   \text{Cov} \bigg[    K \bigg(  \frac{X_1 - X_2}{h_n}      \bigg) \delta(X_1) \delta(X_2) ,  K \bigg(  \frac{X_3 - X_2}{h_n}      \bigg) \delta(X_3) \delta(X_2)                   \bigg]            . \end{align*}
By Lemma \ref{aux-new} we obtain   \begin{align*}
T_1   \leq    \E \bigg[ K^2 \bigg(  \frac{X_1 - X_2}{h_n} \bigg)   \delta^2(X_1) \delta^2(X_2)                      \bigg]  &  \lessapprox  \E \bigg[   K^2 \bigg(  \frac{X_1 - X_2}{h_n} \bigg)          \bigg] \\ &   \lessapprox \E \big[  F_{X}(X- h_n \iota , X + h_n \iota)            \big] .
\end{align*}

By an analogous argument to (\ref{boundsq}) we obtain
\begin{align*}
T_2  & \leq \E \bigg[    K \bigg( \frac{X_1 - X_2}{h_n}        \bigg) K \bigg( \frac{X_3 - X_2}{h_n}   \bigg) \delta(X_1) \delta(X_3) \delta^2(X_2)                \bigg] \\ & \lessapprox \E \big[ \big \{  F_{X}(X -  h_n \iota , X +  h_n \iota)       \big \}^2             \big]   .
\end{align*}
 The bound (\ref{vara2}) follows from combining the bounds for $T_1,T_2$.  From (\ref{vara2}) and  $\gamma_n \lessapprox n^{-1/2} \{ \E[ F_{X}(X - h_n \iota , X + h_n \iota) ]   \}^{-1/4}$ we obtain \begin{align*}
\text{Var} \bigg(    \frac{n A_2}{\sqrt{\E(H_n^2)}}   \bigg) \lessapprox  \frac{1}{n^2  \E(H_n^2)  } + \frac{\E \big[ \big \{ F_X(X - h_n \iota , X +  h_n \iota ) \big \}^2 \big]                 }{n  \E(H_n^2)  \E \big[  F_{X}(X- h_n \iota , X + h_n \iota)            \big]  }.
\end{align*}
The trivial bound $F_X^2(.) \leq F_X (.) $ reduces this to
\begin{align*}
\text{Var} \bigg(    \frac{n A_2}{\sqrt{\E(H_n^2)}}   \bigg) \lessapprox  \frac{1}{n^2  \E(H_n^2)  } + \frac{1}{n \E(H_n^2)} .
\end{align*}
By Corollary 3.4(ii), $\E(H_n^2) \gtrapprox h_n^q$. Since $n h_n^q \uparrow \infty$, the term on the right is $o(1)$. \\ 

\item[(d)]

It remains to prove the statement of the theorem. Combining the bounds derived in $(a-c)$ shows that $$  n \frac{\hat{I}_n}{\sqrt{\hat{\sigma}_n^2}} = N(0,1) +  L_n +             o_{\mathbb{P}}(1)     $$
where  \begin{align*}
L_n = \frac{n \E(A_2)}{\sqrt{\E(H_n^2)   }} =   \frac{n \gamma_n^2}{\sqrt{\E(H_n^2)}} \bigg \{  & \E \bigg[  \delta^2(X) \int_{[0,1]^q} F_X(X - h_n v , X + h_n v) \partial_{v} K(-v) dv               \bigg] \\ & +  o( \E[ F_X(X - h_n \iota , X + h_n \iota) ]  ) \bigg \}.
\end{align*}
Define $$ \alpha_n = n^{-1/2} \{ \E[ F_{X}(X - h_n \iota , X + h_n \iota) ]   \}^{-1/4} .  $$

Note that \begin{align*}
&  \E \big[  \delta^2(X) \int_{[0,1]^q} F_X(X - h_n v , X + h_n v) \partial_{v} K(-v) dv               \big] \\ & \lessapprox   \E \big[  \delta^2(X) F_X(X - h_n \iota , X +h_n \iota)               \big] .
\end{align*}

From substituting $\E[H_n^2] \asymp  \E \big[  F_X (X - h_n \iota , X + h_n \iota )              \big] $, it follows that \begin{align*}
 L_n  \asymp \begin{cases}  o(1) & \gamma_n =o( \alpha_n) \\  \frac{ \E \big[  \delta^2(X) \int_{[0,1]^q} F_X(X - h_n v , X + h_n v) \partial_{v} K(-v) dv               \big]}{\E[ F_X(X - h_n \iota , X + h_n \iota) ]} + o(1)  & \gamma_n \asymp \alpha_n .   \end{cases}
 \end{align*}

In particular $n \hat{I}_n / \sqrt{\hat{\sigma}_n^2} = N(0,1) + o_{\mathbb{P}}(1) $, provided that \begin{align*} \limsup_{h \downarrow 0} \frac{ \E \big[  \delta^2(X) F_X(X - h \iota , X +h \iota)               \big]}{\E[ F_X(X -h \iota , X + h \iota) ]} = 0 .   \end{align*} Additionally, if $\varepsilon$ is as in Assumption 4, we have that \begin{align*}
 & \frac{ \E \big[  \delta^2(X) \int_{[0,1]^q} F_X(X - h_n v , X + h_n v) \partial_{v} K(-v) dv               \big]}{\E[ F_X(X - h_n \iota , X + h_n \iota) ]} \\ & \gtrapprox  \frac{ \E \big[  \delta^2(X) \int_{[\varepsilon,1]^q} F_X(X - h_n v , X + h_n v) \partial_{v} K(-v) dv               \big]}{\E[ F_X(X -h  \iota , X + h  \iota) ]} \\ & \gtrapprox \frac{ \E \big[  \delta^2(X) F_X(X -h_n \varepsilon  \iota , X + h_n \varepsilon  \iota) ] }{\E[ F_X(X -h_n  \iota , X + h_n  \iota) ]} \\ & \gtrapprox \frac{ \E \big[  \delta^2(X) F_X(X -h_n \varepsilon  \iota , X + h_n \varepsilon  \iota) ] }{\E[ F_X(X -h_n \varepsilon  \iota , X + h_n \varepsilon  \iota) ]}. 
\end{align*}

In particular $\liminf \limits_{n \rightarrow \infty} L_n > 0 $, provided that $$ \liminf_{h \downarrow 0} \frac{ \E \big[  \delta^2(X) F_X(X - h \iota , X +h \iota)               \big]}{\E[ F_X(X -h \iota , X + h \iota) ]} > 0  . $$
\end{enumerate}

\end{proof}

\subsection{Theorem 3.11}

\begin{proof}
[Proof of Theorem 3.11] 
First, we claim that there exists a $M < \infty$ such that \begin{align}
\label{delres} \mathbb{P} \bigg(  \mathbbm{1} \big \{ X \in \mathcal{S}_{\delta}       \big \}   \frac{F_X(X - h \iota , X + h \iota)}{(2h)^{s_{\delta}q}}  \leq M          \bigg) = 1
\end{align}
holds for all sufficiently small $h > 0 $. Since $F_{X} = \sum_{t \in T} \alpha_{t} F_t$ and $\underset{X  \stackrel{}{\sim} F_{t} }{\E} [ \delta^2(X)  ] = 0 $ for every $t \in T \setminus R$, (\ref{delres}) follows if we can verify that there exists a $ M < \infty$ such that \begin{align}
\label{delres2}  \underset{X  \stackrel{}{\sim} F_{z} }{\mathbb{P}} \bigg(  \mathbbm{1} \big \{ X \in \mathcal{S}_{\delta}       \big \}   \frac{F_t(X - h \iota , X + h \iota)}{(2h)^{s_{\delta}q}}  \leq M          \bigg) = 1 \; \; \; \; \forall \; z \in R \; \; \; \forall \, t \in T.
\end{align}
Fix any $z \in R$. If $t \in T \setminus R$ and $s_t < s_{\delta}$, this follows from Condition (27) in the main text. For any $ t \in T $ with $s_{t} \geq s_{\delta}$, let $S_{t} \subseteq \R^q $ denote the set where the first condition of Definition 3.6 holds (with constant $M_{t}$). If $X \in S_t$, then $F_t(X - h \iota , X + h \iota) \leq M_{t} h^{s_t q } \leq M_{t} h^{s_{\delta} q} $ because $s_t \geq s_{\delta}$. If $X \notin S_{t}$, then either $F_t(X - h \iota , X + h \iota) = 0  $ or there exists some $y_{X} \in S_t$ such that $ \|  X - y_{X} \|_{\infty} \leq h$ (because $\underset{X  \stackrel{}{\sim} F_{t} }{\mathbb{P}}(S_t) = 1$). We have the inclusion $ \{  t \in \R^q : \| t -X \|_{\infty} \leq h_n   \}  \subseteq \{ t \in \R^q : \|t - y_{X}   \|_{\infty} \leq 2 h_n      \} $ and it follows that $$  F_{t}(X - h_n \iota , X +h_n \iota) \leq F_{t}(y_X - 2h_n \iota , y_X +2h_n \iota) \leq M_{t} 2^{s_tq} h_n^{s_{t} q} \leq M_{t} 2^{s_t q} h_n^{s_{\delta} q} $$
because $s_{t} \geq s_{\delta}$.
From combining the cases, we obtain (\ref{delres}). \\ 

With (\ref{delres}) established, the majority of the proof is analogous to the proof of Theorem 3.10, with some arguments modified to take into account the weaker requirement that $\delta \in L^2(X)$. We provide the details for the main changes in the argument.

 By Lemma \ref{varest2}, we have that $ \hat{\sigma}_n^2 / 2 = \E(H_n^2)[1+ o_{\mathbb{P}}(1)]$. Therefore, it suffices to study the limiting behavior of  $ n \hat{I}_n / \sqrt{2 \E(H_n^2)} $.  Let $A_2, A_4, A_6$ be as in the proof of Theorem 3.10. Define $s = \min_{t \in T} s_t$.  By Lemma 3.7, we have $F_{X} \in \mathcal{D}(s)$ and $\E[F_X(X - h_n \iota , X + h_n \iota)] \asymp h_n^{sq}$.  Assumptions (4,  6) are automatically satisfied.  By Lemma 3.5 and Assumption 4(i) we have $ \E[H_n^2] \asymp h_n^{sq} $ as well. \\

\begin{enumerate}

\item[(a)]

We verify that $A_4 = o_{\mathbb{P}} \big( n^{-1} \sqrt{\E(H_n^2)}        \big) $ under both parts (i) and (ii) of Theorem 3.11. An analogous argument to the proof of Theorem 3.10 shows that \begin{align*}
 A_4 = n^{-1/2} \gamma_n O_{\mathbb{P}} \bigg(    \E \bigg[ K \bigg(  \frac{X_1 - X_2}{h_n}   \bigg) M(X_1) \left| \delta(X_2) \right|       \bigg]         \bigg).
\end{align*}
If $\delta \in L^{\infty}(X)$, we use Lemma \ref{aux-new} and (\ref{delres}) to obtain \begin{align*}
 \E \bigg[ K \bigg(  \frac{X_1 - X_2}{h_n}   \bigg) M(X_1) \left| \delta(X_2) \right|       \bigg]   &  \lessapprox    \E \big[  \mathbbm{1} \big \{  X \in \mathcal{S}_{\delta}    \big \}   M(X) F_X(X - h_n \iota , X + h_n \iota)              \big] \\ & \lessapprox h_n^{s_{\delta} q}.
\end{align*}
Substituting $\gamma_n \lessapprox n^{-1/2} \{ \E[ F_X(X - h_n \iota , X + h_n \iota) ]   \}^{1/4} h_n^{-s_{\delta}q/2}$ yields $$ \frac{n A_4}{ \sqrt{\E(H_n^2)}  } = O_{\mathbb{P}} \bigg(    \frac{h_n^{s_{\delta}q}}{h_n^{sq/4} h_n^{s_{\delta}q/2} }                  \bigg) = o_{\mathbb{P}}(1)  $$
because $s_{\delta} \geq s$.

If $\delta \in L^2(X)$, we use Lemma \ref{aux-new} to obtain  \begin{align*}
\E \bigg[ K \bigg(  \frac{X_1 - X_2}{h_n}   \bigg) M(X_1) \left| \delta(X_2) \right|       \bigg]   &  \lessapprox    \E \big[  \mathbbm{1} \big \{  X \in \mathcal{S}_{\delta}    \big \} \left| \delta(X) \right|  \Omega_{M}  (X - h_n \iota , X + h_n \iota)              \big] .
\end{align*}
Let $\epsilon > 0 $ be such that $M(X) \in L^{4 + \epsilon}(X)$. Let $ \zeta = (3+\epsilon) / (4+ \epsilon) > 3/4  $. By H\"{o}lder's inequality, we obtain \begin{align*}
\Omega_{M}(X- h_n \iota , X + h_n \iota) \leq  \{  \E(M^{4 + \epsilon})  \}^{1/(4 + \epsilon)} \{  F_X(X - h_n \iota , X + h_n \iota)     \}^{ \zeta  }.
\end{align*}
Substituting this to bound the expectation and using (\ref{delres}) yields   \begin{align*} & \E \big[  \mathbbm{1} \big \{  X \in \mathcal{S}_{\delta}    \big \} \left| \delta(X) \right|  \Omega_{M}  (X - h_n \iota , X + h_n \iota)              \big] \lessapprox h_n^{s_{\delta} \zeta q}        \; ,  \\ &   \frac{n A_4}{ \sqrt{\E(H_n^2)}  } = O_{\mathbb{P}} \bigg(    \frac{h_n^{s_{\delta} \zeta q}}{h_n^{sq/4} h_n^{s_{\delta}q/2} }                  \bigg) = o_{\mathbb{P}}(1)      \end{align*}
because $s_{\delta} \geq s $ and $\zeta > 3/4$. \\

\item[(b)]

We verify that $A_6 = o_{\mathbb{P}} \big( n^{-1} \sqrt{\E(H_n^2)}        \big) $ under both parts (i) and (ii) of Theorem 3.11. Suppose $\delta \in L^2(X)$. From the argument in Theorem 3.10, we obtain \begin{align}   \E( \left| A_6  \right|^2) =  \gamma_n^2 O(n^{-2}) T_1 + \gamma_n^2 O(n^{-1}) T_2   \label{a6new}   \end{align}
where \begin{align*}
&  T_1 =  \E \bigg[  K^2 \bigg( \frac{X_1- X_2}{h_n}     \bigg) \delta^2(X_1)    \bigg]       \\ & T_2 =  \E \bigg[ K \bigg(  \frac{X_1 - X_2}{h_n}     \bigg)  K \bigg(  \frac{X_3 - X_2}{h_n}   \bigg)   \left| \delta(X_1) \right| \left|  \delta(X_3) \right|         \bigg] .
\end{align*}
Without loss of generality, we take $\delta(.)$ to be non-negative so that the absolute values can be dropped. From Lemma \ref{aux-new} and (\ref{delres}) we obtain \begin{align*}
T_1 \lessapprox \E \big[ \mathbbm{1} \{ X \in \mathcal{S}_{\delta}   \} \delta^2(X) F_{X}(X - h_n \iota , X + h_n \iota)               \big] \lessapprox  \E[\delta^2(X)]   h_n^{s_{\delta} q} \lessapprox  h_n^{s_{\delta} q} 
\end{align*}
For $T_2$, we note that  \begin{align*}
T_2 =    \E \bigg[ \delta(X_1)     \delta(X_3) \int_{\R^q} K \bigg(  \frac{X_1- x}{h_n}    \bigg) K \bigg(  \frac{X_3 - x}{h_n}      \bigg) d F_X(x)           \bigg].
\end{align*}
By an analogous argument to the expansion of $\E(G_n^2)$ in Lemma 3.3  and  (\ref{delres}) we obtain \begin{align*}
T_2 & \lessapprox  \E \big[  \delta(X) F(X - h_n \iota , X + h_n \iota) \Omega_{\delta} (X - 2 h_n \iota , X + 2 h_n \iota)                          \big] \\ &  =   \E \big[  \mathbbm{1} \{  X \in \mathcal{S}_{\delta}      \} \delta(X) F(X - h_n \iota , X + h_n \iota) \Omega_{\delta} (X - 2 h_n \iota , X + 2 h_n \iota)                          \big] \\\ & \lessapprox h_n^{s_{\delta}q} \, \E[ \mathbbm{1} \{  X \in \mathcal{S}_{\delta}      \} \delta(X)  \Omega_{\delta} (X - 2 h_n \iota , X + 2 h_n \iota)          ].
\end{align*}
Given a positive sequence $(\alpha_n)_{n=1}^{\infty}$, define
\begin{align*}&  \Omega_{\delta > \alpha_n}(X - h_n \iota ,X + h_n  \iota )=\int \limits_{ t \in \R^q : \|t - X  \|_{\infty} \leq h_n  } \delta(t) \mathbbm{1}\{\delta(t) > \alpha_n \} dF_{X}(t) \; , \\ &   \Omega_{\delta \leq \alpha_n}(X - h_n \iota ,X + h_n  \iota )=\int \limits_{ t \in \R^q : \|t - X  \|_{\infty} \leq h_n  } \delta(t) \mathbbm{1} \{\delta(t) \leq \alpha_n \} dF_{X}(t) \:.
\end{align*}
From an application of Cauchy-Schwarz on $\Omega_{\delta > \alpha_n}$, we obtain \begin{align*}   & \Omega_{\delta}(X - 2 h_n \iota , X + 2 h_n \iota)   \\ &  =   \Omega_{\delta \leq \alpha_n}(X - 2h_n \iota ,X +2 h_n  \iota ) + \Omega_{\delta > \alpha_n}(X - 2h_n \iota ,X + 2h_n  \iota )     \\ &   \leq \alpha_n F_X(X - 2 h_n \iota , X + 2 h_n \iota) + \sqrt{ \E[ \delta^2(X)  \mathbbm{1} \{ \delta > \alpha_n    \} ]  } \sqrt{F_X(X  - 2 h_n \iota , X + 2 h_n \iota)} .
\end{align*}
Substituting this into the expectation and using (\ref{delres}) yields $$ T_2 \lessapprox h_n^{s_{\delta}q} \big[  \alpha_n h_n^{s_{\delta}q}  +  \sqrt{ \E[ \delta^2(X)  \mathbbm{1} \{ \delta > \alpha_n    \} ]  }  h_n^{s_{\delta}q/2}            \big]   . $$
Letting $\alpha_n \uparrow \infty$ sufficiently slowly, for e.g. $\alpha_n \asymp h_n^{- s_{\delta}q/4} $, yields $T_2 = o( h_n^{1.5 s_{\delta}q }  ) $. From substituting the bounds for $(T_1,T_2)$ and $\gamma_n \lessapprox n^{-1/2} \{ \E[ F_X(X - h_n \iota , X + h_n \iota) ]   \}^{1/4} h_n^{-s_{\delta}q/2}$ into (\ref{a6new}) we obtain \begin{align*}  \frac{n^2 (A_6)^2}{ \E(H_n^2)} & = O_{\mathbb{P}} \bigg(  \frac{\gamma_n^2 T_1}{ \E(H_n^2)}      \bigg)   + O_{\mathbb{P}} \bigg(   \frac{n \gamma_n^2 T_2}{\E(H_n^2)}        \bigg) \\ & =  O_{\mathbb{P}} \bigg(  \frac{1}{n h_n^{sq/2}}   \bigg) + O_{\mathbb{P}} \bigg(  \frac{o( h_n^{1.5 s_{\delta}q }  )}{h_n^{s_{\delta}q}  h_n^{sq/2} }        \bigg).
\end{align*}
This second term is $o_{\mathbb{P}}(1)$ because $s \leq s_{\delta}  $ and the first is $o_{\mathbb{P}}(1)$ because $s \leq 1 $ and  $n h_n^{q} \uparrow \infty$. \\

\item[(c)]

We aim to show that \begin{align}
\E(A_2)   =     \gamma_n^2 \bigg \{ &  \E \bigg[  \delta^2(X) \int_{[0,1]^q} F_X(X - h_n v , X + h_n v) \partial_{v} K(-v) dv               \bigg]  + o(1) h_n^{s_{\delta}q}            \bigg \}.
\label{biasa2new}
\end{align}

\begin{align}
  \frac{n}{\sqrt{\E(H_n^2)}} \left| A_2 - \E(A_2)   \right| = o_{\mathbb{P}}(1). 
\label{vara2new}
\end{align}

Suppose $\delta \in L^2(X)$. A derivation with steps similar to the proof of Lemma 3.3(i) yields   \begin{align*}
\E(A_2) & =  \gamma_n^2 \E \bigg[ K \bigg( \frac{X_1 - X_2}{h_n} \bigg) \delta(X_1) \delta(X_2)             \bigg] \\ & =  \gamma_n^2 \E \bigg(   \delta(X) \int 
 _{[0,1]^q} \Omega_{\delta} (X - h_nv , X + h_nv) \partial_{v} K(-v) dv          \bigg) \\ &  =  \gamma_n^2  R_2.
\end{align*}
Define, as in the proof of Theorem $3.10$, the quantity $$ Z_{n,v}(x) =  \begin{cases}   \frac{\Omega_{\delta}(x - h_n v , x + h_n v)}{F_X(x - h_n v , x + h_n v)} & F_X(x - h_n v , x + h_n v) > 0 \; , \\ 0 & \text{else}      . \end{cases}  $$ 
It follows that \begin{align*}
  & \left| R_2 - \E \bigg[  \delta^2(X) \int_{[0,1]^q} F_X(X - h_n v , X + h_n v) \partial_{v} K(-v) dv               \bigg]  \right| \\ & \leq   \E \bigg( \left| \delta(X)        \right|   \int_{[0,1]^q} \left| Z_{n,v}(X)  - \delta(X)          \right| F_X(X - h_n v , X + h_n v)   \partial_{v} K(-v) dv   \bigg) \\  &= o(1) \E \big[ \mathbbm{1} \big \{ X \in \mathcal{S}_{\delta}  \big \}  \left| \delta(X) \right|  F_X(X - h_n \iota , X + h_n \iota)   \big] \\ & = o(1) h_n^{s_{\delta}q} \; ,
\end{align*}
where the last equality follows from (\ref{delres}).

We now verify (\ref{vara2new}). Let $ \bar{\delta}(X_i) = \delta(X_i) \mathbbm{1} \{  \delta^2(X_i) \leq n \} $ and define $$ \bar{A}_2 = \frac{\gamma_n^2}{n(n-1)} \sum_{i=1}^{n}\sum_{j\neq i}K\bigg(\frac{X_{i}-X_{j}%
}{h_{n}}\bigg) \bar{\delta}(X_i)  \bar{\delta}(X_j) .  $$
To show (\ref{vara2new}), it suffices to verify \begin{align*}
& B_1 = \frac{n}{\sqrt{\E(H_n^2)}} \left| A_2 - \bar{A_2}         \right| = o_{\mathbb{P}}(1) \; , \\ &  B_2 = \frac{n}{\sqrt{\E(H_n^2)}} \left|  \bar{A}_2 - \E( \bar{A}_2 )       \right| = o_{\mathbb{P}}(1) \; , \\  & B_3 = \frac{n}{\sqrt{\E(H_n^2)}} \left| \E(A_2) - \E(\bar{A}_2)    \right| = o(1).
\end{align*}
The result for $B_1$ follows immediately from the assumption that $\delta \in L^2(X)$. Indeed, for any $\epsilon > 0 $, we have that \begin{align*} \mathbb{P}(B_1 > \epsilon) \leq  \mathbb{P} \big(  \sup_{i=1, \dots , n} \delta^2(X_i) > n \big ) \leq n \mathbb{P}( \delta^2(X) > n )&  \leq  \E\big[  \delta^2(X) \mathbbm{1} \{ \delta^2(X) > n \}      \big]    \\ &  = o(1).
\end{align*}
Next, we verify the estimate for $B_3$. Note that \begin{align*}
\left| \E(A_2) - \E(\bar{A}_2)    \right| \leq \gamma_n^2 \bigg(  & \E \bigg[ K \bigg(  \frac{X_1 - X_2}{h_n}     \bigg) (\delta - \bar{\delta}) (X_1)  \delta(X_2)      \bigg] \\ &  + \E \bigg[ K \bigg(  \frac{X_1 - X_2}{h_n}     \bigg) \bar{\delta}(X_1) (\delta - \bar{\delta})(X_2)       \bigg]    \bigg).
\end{align*}
By arguing in an analogous way to the derivation of $\E(A_2)$ in (\ref{biasa2new}), we obtain \begin{align*}
& \E \bigg[ K \bigg(  \frac{X_1 - X_2}{h_n}     \bigg) (\delta - \bar{\delta}) (X_1)  \delta(X_2)      \bigg] \\ &   = o(1) h_n^{s_{\delta}q} + \E \bigg[  \delta^2(X) \mathbbm{1} \{ \delta^2(X) > n   \} \int_{[0,1]^q} F_X(X - h_n v , X + h_n v) \partial_{v} K(-v) dv               \bigg] \\ & = o(1) h_n^{s_{\delta}q} +  h_n^{s_{\delta}q} \E \big[ \delta^2(X) \mathbbm{1} \{ \delta^2(X) > n   \} \big ] \\ & = o(1) h_n^{s_{\delta}q}
\end{align*}
Similarly, we obtain that  $$  \E \bigg[ K \bigg(  \frac{X_1 - X_2}{h_n}     \bigg) \bar{\delta}(X_1) (\delta - \bar{\delta})(X_2)       \bigg]      = o(1) h_n^{s_{\delta}q}.  $$
Substituting $\gamma_n \lessapprox n^{-1/2} \{ \E[ F_X(X - h_n \iota , X + h_n \iota) ]   \}^{1/4} h_n^{-s_{\delta}q/2}$, we obtain $$  B_3 \leq \gamma_n^2  \frac{n}{\sqrt{\E(H_n^2)}} o(1) h_n^{s_{\delta}q} = o(1). $$ 
Finally, we verify the estimate for $B_2$. By arguing analogously as in the proof of Theorem 3.10, the variance of $  \bar{A_2} $ has order $$  \text{Var}(\bar{A_2})  \lessapprox  \gamma_n^4 \big( n^{-2} T_1 + n^{-1} T_2 \big) $$
 where \begin{align*}
 & T_1 = \E \bigg[ K^2 \bigg(  \frac{X_1 - X_2}{h_n}     \bigg)  \bar{\delta}^2(X_1) \bar{\delta}^2(X_2)          \bigg] \leq n \E \bigg[ K^2 \bigg(  \frac{X_1 - X_2}{h_n}     \bigg)  \delta^2(X_1) \bigg]  \\ & T_2 = \E \bigg[  K \bigg(  \frac{X_1 - X_2}{h_n}        \bigg) K \bigg(  \frac{X_3 - X_2}{h_n}        \bigg)  \bar{\delta}(X_1) \bar{\delta}(X_3) \bar{\delta}^2(X_2)        \bigg] \\ & \; \; \; \; \; \; \; \; \; \leq \E \bigg[  K \bigg(  \frac{X_1 - X_2}{h_n}        \bigg) K \bigg(  \frac{X_3 - X_2}{h_n}        \bigg)  \delta(X_1) \delta(X_3) \delta^2(X_2)        \bigg]
 \end{align*}
 By Lemma \ref{aux-new} and (\ref{delres}), we obtain \begin{align*}
 T_1 \lessapprox n \E \big[ \mathbbm{1} \{ X \in \mathcal{S}_{\delta}  \} \delta^2(X) F_X(X- h_n \iota, X + h_n \iota)        \big] \lessapprox n h_n^{s_{\delta}q}.
 \end{align*}
Since $X_1,X_2,X_3$ are i.i.d we obtain that   \begin{align}
&     \mathbb{E}\bigg[K\bigg(\frac{X_{1}-X_{2}}{h_{n}}\bigg)K\bigg(\frac
{X_{3}-X_{2}}{h_{n}}\bigg) \delta(X_1) \delta(X_3)  \delta^2(X_2) \bigg] \nonumber \\ &  = \mathbb{E}\bigg[  \delta^2(X_2)  \bigg(\int_{\mathbb{R}^{q}}K\bigg(\frac{t-X_{2}%
}{h_{n}}\bigg) \delta(t)  d F_{X}(t)\bigg)\bigg(\int_{\mathbb{R}^{q}}K\bigg(\frac{t-X_{2}}{h_{n}}%
\bigg) \delta(t)  d F_{X}%
(t)\bigg)\bigg]  \nonumber \\ & = \mathbb{E}\bigg[ \delta^2(X_2)  \bigg(\int_{\mathbb{R}^{q}}K\bigg(\frac{t-X_{2}%
}{h_{n}}\bigg) \delta(t)  d F_{X}(t)\bigg)^2 \bigg]. \label{boundsq0}
\end{align}
 
Define
\[
f(x)= M(x),\;g(x)=K\bigg(\frac{x-X_{2}}{h_{n}}\bigg).
\]
Let $X_2^i$ denote the $i^{th}$ coordinate of $X_2$. Conditional on $X_{2}$, $(f,g)$ satisfy the hypothesis of Lemma 3.2 with  $ \mathcal{O}  = (X_{2}^{1}%
-h_{n},X_{2}^{1}+h_{n}) \times \dots \times (X_{2}^{q}%
-h_{n},X_{2}^{q}+h_{n})  $.  Applying Lemma 3.2 yields
\begin{align}
 &  \mathbb{E}\bigg[\bigg(\int_{\mathbb{R}^{q}}K\bigg(\frac{x-X_{2}}{h_{n}%
}\bigg) \delta(x)  dF_{X}(x)\bigg)^{2}%
\bigg]  \nonumber
\\ &  = \E \bigg[  \bigg( \int_{\mathbb{R}^{q}} \Omega_{\delta}(X_{2}-h_{n}\iota,t)\partial_{t}%
K^{}\bigg(\frac{t-X_{2}}{h_{n}}\bigg)dt \bigg)^2 \bigg] \nonumber \\
&  = \E \bigg[  \bigg( \int_{\left[  -1,1\right]  ^{q}}  \Omega_{\delta} (X_{}-h_{n}\iota,X_{}%
+h_{n}v)\partial_{v}K^{}(v)dv \bigg)^2 \bigg] \nonumber \\ &  \lessapprox     \E \big[  \Omega_{|\delta|}^2 (X - h_n \iota , X + h_n \iota )              \big]              \nonumber                 \label{boundsq} \; ,
\end{align}
 where the third equality  follows from the change of variables $t\rightarrow
X_{2}+h_{}v$ and the support of $k(\,\cdot\,)$.  It follows that
  \begin{align*}
 T_2 \lessapprox \E \big[ \delta^2(X) \Omega_{\left| \delta \right|}^2(X - h_n \iota , X + h_n \iota)         \big].
 \end{align*}
 By Cauchy-Schwarz, $\Omega_{\left| \delta \right|}^2(X - h_n \iota , X + h_n \iota)   \leq \E [ \delta^2(X) ] F_X(X - h_n \iota , X + h_n \iota) $ and we obtain using (\ref{delres}) that $$ T_2 \lessapprox \E \big[ \mathbbm{1} \{ X \in S_{\delta} \}  \delta^2(X)  F_X(X - h_n \iota , X + h_n \iota)  \big] \lessapprox h_n^{s_{\delta}q} . $$
By substituting $\gamma_n \lessapprox n^{-1/2} \{ \E[ F_X(X - h_n \iota , X + h_n \iota) ]   \}^{1/4} h_n^{-s_{\delta}q/2}$, it follows that \begin{align*}
B_2 =   \frac{n^2}{\E(H_n^2)} O_{\mathbb{P}} \big( \text{Var}(\bar{A}_2)       \big) & = O_{\mathbb{P}} \bigg(  \frac{\gamma_n^4 T_1}{\E(H_n^2)}  + \frac{n \gamma_n^4 T_2}{\E(H_n^2)}                   \bigg) \\ & = O_{\mathbb{P}} \bigg( \frac{1}{n h_n^{s_{\delta}q}}  + \frac{1}{n h_n^{s_{\delta}q}}        \bigg).
\end{align*}
 This term is $o_{\mathbb{P}}(1)$ because $s_{\delta} \leq 1 $ and $n h_n^{q} \uparrow \infty$. \\ 
 
 \item[(d)]

It remains to prove the statement of the theorem. Combining the bounds derived in $(a-c)$ shows that $$  n \frac{\hat{I}_n}{\sqrt{\hat{\sigma}_n^2}} = N(0,1) +  L_n +             o_{\mathbb{P}}(1)     $$
where  \begin{align*}
L_n =   \frac{n \gamma_n^2}{\sqrt{\E(H_n^2)}} \bigg \{  & \E \bigg[  \delta^2(X) \int_{[0,1]^q} F_X(X - h_n v , X + h_n v) \partial_{v} K(-v) dv               \bigg] + o(1) h_n^{s_{\delta}q} \bigg \}.
\end{align*}
Define $$ \alpha_n = n^{-1/2} \{ \E[ F_X(X - h_n \iota , X + h_n \iota) ]   \}^{1/4} h_n^{-s_{\delta}q/2}.  $$
 Note that by (\ref{delres}) we have \begin{align*}
& \frac{ \E \big[  \delta^2(X) \int_{[0,1]^q} F_X(X - h_n v , X + h_n v) \partial_{v} K(-v) dv               \big]}{ h_n^{s_{\delta}q} } \\ & \lessapprox  \frac{ \E \big[  \mathbbm{1} \{ X \in \mathcal{S}_{\delta}  \} \delta^2(X) F_X(X - h_n \iota , X +h_n \iota)               \big]}{h_n^{s_{\delta}q}} \\ & \lessapprox 1.
\end{align*}
It follows that
  \begin{align*}
 L_n  \asymp \begin{cases}  o(1) & \gamma_n =o( \alpha_n) \\   \E \big[  \delta^2(X) \int_{[0,1]^q} h_n^{- s_{\delta}q} F_X(X - h_n v , X + h_n v) \partial_{v} K(-v) dv               \big] + o(1)  & \gamma_n \asymp \alpha_n .   \end{cases}
 \end{align*}
 For any fixed $\varepsilon \in (0,1)$, we have that \begin{align*}
 & \E \bigg[ \delta^2(X) \int_{[0,1]^q}  F_X(X - h_n v , X + h_n v) \partial_{v} K(-v) dv   \bigg]    \\ & \geq    \E \bigg[  \delta^2(X)     \int_{[\varepsilon,1]^q}  F_X(X - h_n v , X + h_n v) \partial_{v} K(-v) dv \bigg] \\ & \gtrapprox \E \big[ \delta^2(X) F_X(X - h_n \varepsilon \iota , X + h_n \varepsilon \iota) \big] \\ & \gtrapprox   \underset{X  \stackrel{}{\sim} F_{R} }{\E} \big[ \delta^2(X) F_{R} \big( X - h_n \varepsilon \iota , X + h_n \varepsilon \iota        \big)    \big].
 \end{align*}
 The result  follows from substituting $ \underset{X  \stackrel{}{\sim} F_{R} }{\E}  \big[  F_R(X - h_n \varepsilon \iota , X + h_n \varepsilon \iota)      \big] \asymp  h_n^{ s_{\delta} q} $ .

\end{enumerate}

\end{proof}

\end{document}